\documentclass[letterpaper, 11pt]{article}

\usepackage{amsmath, amssymb, amsfonts, amsthm}
\usepackage{bm}
\usepackage{mathtools}
\usepackage{mathrsfs}
\usepackage{geometry}
\usepackage{setspace}
\usepackage[utf8x]{inputenc}
\usepackage{algpseudocode,algorithmicx,algorithm}
\usepackage[T1]{fontenc}
\usepackage{lmodern}
\usepackage{graphicx}
\usepackage{enumerate}
\usepackage[natural, dvipsnames, pdftex]{xcolor}
\usepackage{tikz}
\usepackage{verbatim}
\usepackage{hyperref}
\usepackage{cancel}
\usepackage[all]{xy}
\usepackage{mdframed}
\usepackage[capitalize]{cleveref}
\usepackage{todonotes}
\usepackage{algpseudocode}
\usepackage{comment}
\usepackage{booktabs}
\usepackage{tabularx}
\usepackage{array}
\usepackage{float}

\usepackage[numbers]{natbib}

\usepackage{enumitem}

\usepackage{indentfirst}
\theoremstyle{plain}
\newtheorem{thm}{Theorem}[section]
\newtheorem{lemma}[thm]{Lemma}
\newtheorem{prop}[thm]{Proposition}
\newtheorem{cor}[thm]{Corollary}

\theoremstyle{definition}
\newtheorem{defn}[thm]{Definition}

\theoremstyle{remark}
\newtheorem{rmk}[thm]{Remark}

\newcommand{\F}{\mathbb{F}}

\newcommand{\cC}{\mathcal{C}}

\newcommand{\supp}{\mathrm{supp}}

\newcommand{\wt}{\mathrm{wt}}
\newcommand{\rank}{\mathrm{rank}}

\begin{document}

\title{A Syndrome-Space Approach to Proximity Gaps and Correlated Agreement for Random Linear Codes and Random Reed--Solomon Codes
\thanks{
C. Yuan is with School of Computer Science, Shanghai Jiao Tong University. (Email: \href{chen_yuan@sjtu.edu.cn}{chen\_yuan@sjtu.edu.cn}) R. Zhu is with School of Computer Science, Shanghai Jiao Tong University. (Email: \href{sjtuzrq7777@sjtu.edu.cn}{sjtuzrq7777@sjtu.edu.cn})
}}

\author{Chen Yuan, Ruiqi Zhu}
\date{}
\maketitle

\begin{abstract}
Proximity gaps and correlated agreement have become central tools in the analysis of interactive oracle proofs of proximity (IOPPs) and code-based SNARKs. Informally, a proximity-gap statement says that for a structured set of words---such as a line, an affine space, or a curve---either all points are close to the code, or most are far from it. Such statements are essential in sampling-based proof systems, where a verifier queries only a few random locations on a structured object but must still obtain a global soundness guarantee. In Reed--Solomon-based proof systems, one would ideally like the proximity parameter to approach the information-theoretic limit $1-R$, since this is the largest possible radius for a rate-$R$ code and directly affects protocol efficiency. While recent work has substantially strengthened the picture for algebraic codes and linked proximity gaps to decoding-related structural properties, it remains unclear whether analogous results for random linear codes can be proved directly, rather than by first proving list-decoding-type properties. Random Reed--Solomon codes provide a further natural benchmark, as they retain the algebraic structure of Reed--Solomon codes while introducing randomness through the evaluation set.

In this work, we establish a direct approach to proximity gaps and correlated agreement for random linear codes in the random parity-check-matrix model, without relying on list decoding as the main engine of the proof. Our approach is based on a syndrome-space reformulation together with a witness-based reduction argument, and it yields strong results for affine lines, affine spaces, and polynomial curves. It is conceptually different from the existing decoding-driven route for random linear codes, and it also leads to sharper parameters, including the optimal-up-to-$\varepsilon$ large-alphabet radius bound $\rho<1-R-\varepsilon$ for $q=\Theta(n)$, as well as near-capacity bounds over constant alphabets with improved alphabet-size requirements.

We further apply the same syndrome-space reductions to random Reed--Solomon codes. Since Reed--Solomon parity-check matrices are highly structured, the final random parity-check union bound is replaced by a local-profile argument. This yields correlated agreement for random Reed--Solomon codes over affine spaces and polynomial curves up to radius $\rho\le 1-R-\varepsilon$, with field size $q\ge n\cdot 2^{O(\varepsilon^{-3})}$ for affine spaces and $q\ge n\cdot 2^{O_\ell(\varepsilon^{-3})}$ for degree-$\ell$ curves.
\end{abstract}

\section{Introduction}
Proximity gaps, introduced by~\cite{rothblum2013interactive}, capture a basic local-to-global phenomenon for codes and have become important in the analysis of interactive oracle proofs of proximity (IOPPs) and code-based SNARKs. In protocols such as FRI and DEEP-FRI~\cite{ben2018fast,ben2019deep}, and in more recent systems such as STIR and WHIR~\cite{arnon2024stir,arnon2025whir}, the verifier reasons about a structured family of words---for example, a line, an affine space, or a low-degree curve---while querying only a few random locations. For such sampling-based checks to be sound, one needs a code-theoretic statement showing that the following bad situation cannot happen: many sampled points are individually close to the code, yet there is no single structured family of codewords that explains these nearby points at once. This is exactly the content of proximity-gap and correlated-agreement theorems.

Informally, a proximity-gap theorem says that if too many points of a structured family are within distance $\rho n$ of a code $C$, then the whole family must lie inside a slightly larger neighborhood of $C$. A stronger conclusion, usually called correlated agreement, asserts that the nearby codewords can be chosen coherently: they themselves form a structured family on the code side that agrees with the original family outside a small exceptional set of coordinates. These statements are exactly what allow local sampling to be converted into global soundness, and they have therefore become a basic ingredient in modern analyses of IOPPs and SNARKs. This connection has become even more explicit in recent code-based proof systems and their analyses, where stronger forms of correlated agreement are used to simplify or sharpen soundness arguments~\cite{zeilberger2024khatam,habock2025note}.

Beyond their conceptual role in local testing, proximity-gap theorems are also motivated by concrete applications to modern proof systems. In Reed--Solomon-based IOPPs and their descendants, one would ideally like the proximity parameter to approach the information-theoretic limit $1-R$, since this is the largest possible radius for a rate-$R$ code and directly affects the efficiency of the resulting protocols~\cite{ben2016interactive,ben2023proximity,zeilberger2024basefold}. Prior to recent progress, the strongest general Reed--Solomon proximity-gap results beyond the unique-decoding regime were limited to the Johnson-radius range, and improving this barrier has been repeatedly highlighted as an important challenge in the literature ~\cite{ben2023proximity}. Recent work has further sharpened the Reed--Solomon-side picture, including refined positive and negative results for proximity gaps, direct studies of mutual correlated agreement, and investigations of the conjectural limits of up-to-capacity statements~\cite{BSCHKS25,habock2025note,crites2025reed}. Recent impossibility results~\cite{diamond2025distribution} further indicate that unrestricted up-to-capacity conjectures must be formulated with care: for MDS codes, such statements fail in certain regimes where the relevant parameters vanish with $n$. More recently,~\cite{crites2025reed} shows that such conjectures also cannot hold in regimes where small-list list decoding is impossible. From this perspective, it is natural to ask whether strong proximity gaps also hold for random code ensembles, and what proof argument is responsible for them.

Random linear codes and random Reed--Solomon codes are two natural benchmark ensembles. Random linear codes model generic linear codes. Random Reed--Solomon codes, on the other hand, keep the low-degree structure of Reed--Solomon codes but choose the evaluation set at random. For random linear codes, a long line of work shows that they nearly match fully random codes for list decoding, average-radius list decoding, and related notions~\cite{guruswami2010list,wootters2013list,rudra2018average,guruswami2021bounds}. For random Reed--Solomon codes, recent work shows that they achieve list-decoding capacity with linear-sized alphabets~\cite{alrabiah2025random}, and that random Reed--Solomon codes and random linear codes are locally equivalent for a broad class of local linear properties~\cite{levi2025random}. Proximity gaps, however, ask a different kind of question. They are not about how a code intersects one Hamming ball, but about how the code interacts with a whole structured family of nearby words, such as a line, an affine space, or a curve. Understanding this local geometry helps clarify whether proximity gaps come from special algebraic structure, or whether they also appear in random code ensembles. It also raises a methodological question: can one prove such results directly, without first proving a list-decoding-type statement?

A direct proof is not automatic. Most known approaches first prove an intermediate property related to decoding, and then use that property to derive a proximity gap. For example, recent work uses subspace designs, line- or curve-decodability, and local-property frameworks to prove proximity gaps and correlated agreement for several algebraic and random code families~\cite{guruswami2016explicit,guruswami2008explicit,kopparty2015list,levi2025random,jeronimo2025probabilistic}. These results show that proximity gaps are closely related to decoding, but they do not show that list decoding is necessary for proving them. They have also motivated recent work on stronger generator-based or mutual-correlated-agreement formulations, as well as follow-up results for explicit code families such as folded Reed--Solomon codes~\cite{bordage2025all,jeronimo2026optimal}.

For random linear codes, a natural first attempt is therefore to use their strong list-decoding properties. This is the route taken, either explicitly or implicitly, by several available approaches ~\cite{gao2025list,levi2025random,goyal2025structure,brakensiek2025random,jeronimo2025probabilistic}.  ~\cite{gao2025list} establishes the relationship between list decoding and proximity gaps. In the Reed--Solomon setting, the first strong proximity-gap results beyond unique decoding also used substantial algebraic decoding ingredients~\cite{ben2023proximity}. More recently, Goyal and Guruswami~\cite{goyal2025optimal} gave a broad framework connecting proximity gaps, correlated agreement, subspace-design codes, curve-decoding-type properties, and local properties. These works leave open a more basic question: does one actually need the list-decoding theory of a code in order to prove a proximity gap for that code?

For random linear codes, we show that list decoding is not needed for the main proof. We work directly with a random parity-check matrix. The key idea is to translate proximity into a statement about syndromes, encode a bad line, space, or curve by a low-weight witness matrix, and then use a rank-reduction argument to show that this witness matrix must have high rank. A random parity-check matrix is very unlikely to map many independent low-weight vectors to a prescribed low-dimensional syndrome structure, and this gives the desired proximity-gap and correlated-agreement statements.

The same first part of the proof also applies to random Reed--Solomon codes. A bad affine space or curve again gives a high-rank low-weight witness matrix. The only difference is the last step. For random linear codes, a direct union bound over the random parity-check matrix shows that such matrices cannot occur with high probability. For random Reed--Solomon codes, the parity-check matrix is structured, so this union bound is not available. Instead, the high-rank low-weight witness produced by the rank-reduction argument would make the code contain a low-dimensional local profile, and the local-profile theorem of~\cite{levi2025random} shows that this cannot happen with high probability.

Beyond this conceptual separation from list decoding, our approach also yields stronger parameters. In the constant-alphabet setting, we prove proximity-gap results under the natural entropy condition $H_q(\rho) < 1-R$ once $q$ is a sufficiently large constant. In the large-alphabet setting with $q=\Theta(n)$, we obtain proximity gaps and correlated agreement up to every radius $\rho < 1-R-\varepsilon$, which is optimal up to the $\varepsilon$-loss.

While the line case already contains the core ideas, our framework also yields affine-space and curve-based statements. These results improve the currently available random-linear-code consequences of~\cite{goyal2025optimal} both in the achievable radius and in the required alphabet size. We hope that the syndrome-space viewpoint developed here will be useful more broadly for studying local geometric properties of codes beyond the regimes currently accessible through list decoding.

\subsection{Related Work}

The study of proximity gaps originated in earlier work on proximity testing and interactive proofs for general linear codes, where slack versions of the property were already used for soundness amplification and local testing~\cite{rothblum2013interactive,ames2023ligero,ben2018worst}. In the Reed--Solomon setting, these ideas were substantially strengthened over the past several years. A major milestone was the work of~\cite{ben2023proximity}, which established strong proximity-gap theorems beyond the unique-decoding regime, up to the Johnson radius, by importing powerful ingredients from algebraic unique decoding, list decoding, and low-degree testing; their proof also builds on an extension of the Arora--Sudan low-degree-testing paradigm to the function-field setting~\cite{arora1997improved}. Very recent work has further refined our understanding of Reed--Solomon proximity gaps, including both positive and negative results on proximity gaps, direct studies of mutual correlated agreement, and investigations of conjectural limits~\cite{BSCHKS25,habock2025note,crites2025reed}. Meanwhile, systems such as STIR and WHIR use correlated agreement, and stronger variants of it, as an explicit part of their soundness analysis of practical proof systems~\cite{arnon2024stir,arnon2025whir}.

These advances made the Johnson-radius range the benchmark for strong Reed--Solomon proximity-gap results, and sharpened the broader question of whether one can approach the information-theoretic limit $1-R$ in more general settings. A complementary line of work has developed a broader structural viewpoint around subspace-design codes, line- and curve-decodability, and local coordinate-wise linear properties~\cite{guruswami2016explicit,levi2025random,brakensiek2025random,jeronimo2025probabilistic}. This line of work treats proximity gaps and correlated agreement as consequences of local structural properties related to decoding. It also gives a way to compare random, algebraic, and explicit code families within a common framework. These developments have also motivated stronger formulations such as generator-based and mutual correlated agreement, as well as follow-up results for explicit code families such as folded Reed--Solomon codes~\cite{bordage2025all,jeronimo2026optimal}.

On a different front, random linear codes have been intensively studied from the viewpoint of global decoding. A long line of work shows that they essentially achieve the best known tradeoffs for list decoding, average-radius list decoding, and related notions~\cite{guruswami2010list,wootters2013list,rudra2018average,guruswami2021bounds}. These results give a detailed picture of the global geometry of random linear codes, but they do not by themselves address the structured local questions that arise in proximity-gap theorems.

Random Reed--Solomon codes have recently emerged as a parallel probabilistic benchmark for algebraic codes. Early work on generic or randomly punctured Reed--Solomon codes established capacity-type list-decoding guarantees over large or polynomial-size fields, and subsequent work showed that random Reed--Solomon codes achieve list-decoding capacity with linear-sized alphabets~\cite{alrabiah2025random,brakensiek2023generic}. More recently,~\cite{levi2025random} proved that random Reed--Solomon codes and random linear codes are locally equivalent for broad classes of local coordinate-wise linear properties. This provides a powerful way to transfer local-property threshold statements between the two ensembles. However, proximity gaps and correlated agreement are not themselves local properties: they concern whole structured families of words, such as lines, affine spaces, and curves.

The works closest to our random-Reed--Solomon results are~\cite{goyal2025optimal} and~\cite{levi2025random}.~\cite{goyal2025optimal} developed a broad framework connecting proximity gaps, correlated agreement, subspace-design codes, curve-decoding-type properties, and local properties. Their framework gives consequences for random linear codes and random Reed--Solomon codes by routing the problem through decoding and local-property-type structure.~\cite{levi2025random} proved the local equivalence of random linear and random Reed--Solomon codes and supplied the local-profile theorem that we use in the final probabilistic step of the random-RS argument.

Our proof starts from the bad structured family itself. For random linear codes, we work directly in the random parity-check-matrix model and prove the proximity-gap argument through syndrome-space witnesses and rank reduction. The resulting high-rank low-weight witnesses are then shown to be unlikely by a direct union bound over the random parity-check matrix. For random Reed--Solomon codes, the same witness reduction remains valid, but the last step is different: such a witness gives rise to a low-dimensional local profile, and the local-profile theorem of~\cite{levi2025random} shows that this cannot happen with high probability. This also leads to sharper parameters. In the large-alphabet random-linear-code setting with $q=\Theta(n)$, the route of~\cite{goyal2025optimal} reaches $ \rho < 1-R-2\varepsilon $, whereas our direct argument reaches $ \rho < 1-R-\varepsilon $. For random Reed--Solomon codes, the consequences obtained from~\cite{goyal2025optimal} give correlated agreement up to $ \rho \le 1-R-2\varepsilon $ with alphabet-size dependence $q\ge n\cdot 2^{O(\varepsilon^{-7})}$ for affine spaces and $q\ge n\cdot 2^{O_\ell(\varepsilon^{-(2\ell+5)})}$ for degree-$\ell$ curves; our random-RS theorem reaches $ \rho\le 1-R-\varepsilon $ with $q\ge n\cdot 2^{O_m(\varepsilon^{-3})}$ and $q\ge n\cdot 2^{O_\ell(\varepsilon^{-3})}$, respectively.

\subsection{Overview of Our Approach and Contributions}
Our contributions have two parts. The first is a direct syndrome-space proof of proximity gaps and correlated agreement for random linear codes in the random parity-check-matrix model. The second is an application of the same deterministic syndrome-space reductions to random Reed--Solomon codes, where the final probabilistic step is supplied by the local-profile theorem of~\cite{levi2025random}.

\paragraph{Random linear codes.}
Our goal is to establish proximity-gap and correlated-agreement theorems for random linear codes by a route that is largely independent of their list-decoding theory. The central novelty is a direct analysis in the random parity-check-matrix model. The proof has four main ingredients: a syndrome-space reformulation of proximity, a witness-matrix encoding of bad structured families, a deterministic rank-reduction argument, and a probabilistic counting-and-union-bound argument showing that the resulting witness configurations do not exist with high probability.

The starting point is the observation that for a linear code $\cC$ with parity-check matrix $H\in \F_q^{r\times n}$, a vector $\mathbf{y}$ is within Hamming distance $E$ of $\cC$ if and only if its syndrome $H\mathbf{y}$ lies in the syndrome ball
\[
H_E := \{H\mathbf{x}\in \F_q^r:\wt(\mathbf{x})\leq E\}.
\]
Thus, it suffices to consider the problem that 
$H\mathbf{x}_i=\mathbf{s}_0+\alpha_i\mathbf{s}_1, i=1,\ldots,K$ where $\wt(\mathbf{x}_i)\leq E$ and $\alpha_i\in \F_q$. We want to show that if $K$ is large enough, it is very unlikely that there exist $K$ low weight vectors $\mathbf{x}_1,\ldots,\mathbf{x}_K$ that were all mapped to some points in a given line $\mathbf{s}_0+\alpha\mathbf{s}_1$ by a uniformly random parity-check matrix. 
%This allows us to replace a geometric question in the ambient word space by an intersection problem in syndrome space. In particular, if $U$ is an affine line or affine space in the word space, then its image under $H$ is a corresponding affine syndrome line or affine syndrome space $S$, and the set of points of $U$ that are $E$-close to $\cC$ is exactly the set of points of $S$ lying in $H_E$. Thus, a violation of proximity gap or correlated agreement is converted into the existence of a structured family in syndrome space containing too many points of $H_E$ without lying entirely in $H_{E^+}$.

To analyze the failure probability, we introduce the notion of the witness matrix $X:=[\mathbf{x}_1|\cdots|\mathbf{x}_K]$. We call $X$ a witness matrix for syndrome line $L=\{\mathbf{s}_0+\alpha \mathbf{s}_1: \alpha\in \F_q\}$ if $H\mathbf{x}_i\in L$ for every column vector $\mathbf{x}_i$ in $X$. In the meanwhile, we require $\wt(\mathbf{x}_i)\leq E$. We prove that if $K$ is large enough, we will find a submatrix $X_J$ of $X$ which is still a witness matrix of smaller rank for $L$. The key step is a rank-reduction lemma: if the witness matrix $X$ has rank $t\ge 3$, then one can find a large subset of its columns that still witnesses the same syndrome line but whose rank drops by one. In fact, the size of this subset is at least $\frac{K(d-E)}{d}$ where $d$ is the minimum distance of $\cC$. Repeating this reduction for $t-2$ times, we finally arrive at the case $\rank(X_J)=2$ with $|J|\geq K(\frac{d-E}{d})^{t-2}$. In this case, we use Lemma \ref{lem:gap-line-ball} to show that the size of subset $J$ should be upper bounded by $B_{E,E^+}=\frac{E^{+} + 1}{E^{+} - E + 1}$ where $E^{+}\geq E$ represents the slackness of the proximity gap. This means if $K$ is large enough, we can expect that the rank of $X$ should also be large at the beginning of our reduction. However if $\rank(X)=t$ is large enough, the probability is exponentially small that $HX=Y$ for a uniformly random parity-matrix $H$ and a given matrix $Y$. This allows us to take a union bound over all $K$-subsets $\alpha_1,\ldots,\alpha_K\in \F_q$, all lines $L$ and all $t$-subsets of low weight vectors to complete our probabilistic argument.

%at which point the surviving witness columns admit an affine form $\mathbf{x}_j = \tilde{\mathbf{a}} + \alpha_j \tilde{\mathbf{b}}$. 

%Once the  reduction is established, the random-coding analysis becomes a counting and union-bound argument over possible low-rank witnesses. 

%Combining this with the rank lower bound obtained from the deterministic step, and then taking a union bound over parameter choices and over syndrome lines or spaces, yields exponentially small upper bounds on the probability that a bad configuration exists. In this way, the proof separates cleanly into a deterministic geometry-of-witnesses part and a probabilistic counting part in the random-parity-check-matrix model.

This framework first yields a line proximity-gap theorem for random linear codes. In the constant-alphabet setting, for $E=\lfloor \rho n\rfloor$ and $E^+=E+\lceil \varepsilon n\rceil$, we prove that random linear codes satisfy the $(E,E^+,K/q)$-line proximity gap property with high probability, provided $q$ is a sufficiently large constant depending only on $\varepsilon$, and the relevant radius is governed by the natural entropy condition $H_q(\rho)<1-R$. By the standard line-to-space lifting in syndrome space, this also gives the corresponding space proximity gap statement.

We then extend the same perspective from lines to affine spaces and correlated agreement. The notion of correlated agreement is reformulated in syndrome space in terms of low row-weight witness matrices for affine syndrome spaces. The same rank-reduction argument generalizes to this higher-dimensional setting: a bad affine syndrome space with too many points in $H_E$ but no correlated-agreement explanation in radius $E^+$ must give rise to a high-rank witness matrix, while the probabilistic moment argument shows that such witnesses do not exist with high probability. This yields affine-space correlated-agreement theorems in the constant-alphabet radius setting, as well as large-alphabet radius consequences after invoking the relevant lifting input. We summarize our main results for random linear codes informally below.

\begin{thm}[Informal]
Fix a rate $R\in(0,1)$.

\begin{enumerate}
    \item Large-alphabet setting.
    For every $\varepsilon>0$, if $q=\Theta(n)$ and $\rho < 1-R-\varepsilon$, then a random linear code $\cC$ of rate $R$ over $\mathbb{F}_q$ satisfies proximity gap and correlated-agreement properties at radius $\rho n$ with high probability.

    \item Constant-alphabet setting.
    For every $\varepsilon>0$, if $\rho < 1-R-\varepsilon-\frac{\varepsilon}{\log_2(1/\varepsilon)}$ and $q \ge \left(\frac{2}{\varepsilon}\right)^{2/\varepsilon}$, then, with high probability, a random linear code $\cC$ of rate $R$ over $\F_q$ satisfies proximity-gap and correlated-agreement statements: the existence of many points that are $\rho n$-close to $\cC$ forces the entire structured family to be $(\rho+\varepsilon)n$-close to $\cC$.
\end{enumerate}
\end{thm}

%In the linear-alphabet one-radius regime $q=\Theta(n)$, we obtain proximity gaps and correlated agreement up to every radius $\rho < 1-R-\varepsilon$, which is optimal up to the $\varepsilon$-loss and improves the $\rho<1-R-2\varepsilon$ radius obtained through the existing approach~\cite{goyal2025optimal}. In the constant-alphabet regime, our explicit theorem yields proximity gaps up to $\rho < 1 - R - \varepsilon - \frac{\varepsilon}{\log_2(1/\varepsilon)}$, and, for the alphabet size, we require only that $q \ge \left(\frac{2}{\varepsilon}\right)^{2/\varepsilon}$.

While the line case already contains the main conceptual ideas, the framework also yields affine-space results and correlated-agreement theorems for polynomial curves. We therefore expect the syndrome-space viewpoint developed here to be useful more broadly for local geometric properties beyond the settings currently accessible through list-decoding techniques.

\begin{rmk}[Comparison with~\cite{goyal2025optimal} for random linear codes]
For random linear codes, the framework of~\cite{goyal2025optimal} yields proximity gap and correlated agreement statements up to $\rho<1-R-2\varepsilon$ in the large-alphabet setting. Our argument improves this to $\rho<1-R-\varepsilon$, which is optimal up to the $\varepsilon$-loss.

In the constant-alphabet setting, our theorem applies for $\rho < 1-R-\varepsilon-\frac{\varepsilon}{\log_2(1/\varepsilon)}$ with alphabet size only $q \ge \left(\frac{2}{\varepsilon}\right)^{2/\varepsilon}$, improving over the $1-R-2\varepsilon$ radius benchmark and the $q\ge 2^{O(1/\varepsilon^3)}$ alphabet-size requirement obtained via~\cite{goyal2025optimal}.
\end{rmk}

\paragraph{Random Reed--Solomon codes.}
We also apply the same syndrome-space viewpoint to random Reed--Solomon codes. The deterministic part of the argument remains unchanged.  If an affine space or a degree-$\ell$ curve has many $E$-close points but admits no correlated agreement in radius $E^+$, then the rank-reduction theorems from Sections~\ref{section:correlated-agreement} and~\ref{section:curve-correlated-agreement}  force the existence of many linearly independent low-weight witness columns. Collecting these columns gives a high-rank low-weight witness matrix $X$.

The new difficulty is the final probabilistic step. For random linear codes, a direct union bound over the random parity-check matrix shows that such matrices cannot occur with high probability. For random RS codes, the parity-check matrix is structured, so this direct argument is no longer available. We instead show that such a witness matrix would make the random RS code contain a low-dimensional local profile, and then use the local-profile theorem of~\cite{levi2025random}.

The conversion is the following.  In the affine $m$-space case, the syndrome columns of $X$ lie in $\mathrm{span}(\mathbf{s}_0,\ldots,\mathbf{s}_m)$; in the degree-$\ell$ curve case, they lie in $\mathrm{span}(\mathbf{s}_0,\ldots,\mathbf{s}_\ell)$. Thus, after choosing enough independent columns, the matrix $HX$ has a large kernel. Taking a subspace $U\subseteq \ker(HX)$ of the required dimension gives $X\cdot U\subseteq \cC.$ In other words, many linear combinations of the columns of $X$ are actual codewords.

Now choose a basis $\mathbf{u}_1,\ldots,\mathbf{u}_b$ of $U$ and form $A=[X\mathbf{u}_1|\cdots|X\mathbf{u}_b].$ Since $U\subseteq \ker(HX)$, every column $X\mathbf{u}_j$ lies in $\cC$, so the columns of $A$ are codewords. At the same time, the sparsity of $X$ controls the rows of $A$. Indeed, at coordinate $i$, the row $A_{i,*}$ only depends on those columns $\mathbf{x}_j$ of $X$ with $X_{i,j}\ne 0$. Therefore, $A_{i,*}$ lies in a subspace whose dimension is at most the number of nonzero entries in the $i$-th row of $X$.  Summing over all coordinates, the total local dimension is at most the total number of nonzero entries of $X$, which is at most $tE$.

This is exactly the local-profile step formalized in Lemma~\ref{lem:rs-structured-witness-profile}: a high-rank low-weight witness matrix $X$, together with a subspace $U\subseteq \ker(HX)$, would force the random RS code to contain a local profile of small total dimension. Proposition~\ref{prop:rs-profile-exclusion-uniform} shows that, with high probability, a random RS code contains no such local profile. Combining this probabilistic statement with the deterministic rank-reduction theorems gives the random-RS correlated-agreement theorem, stated precisely as Theorem~\ref{thm:iid-rs-ca-main}.

\begin{thm}[Informal]
Fix a rate $R\in(0,1)$ and fixed integers $m,\ell\ge 1$. For every $\varepsilon>0$ and every radius $\rho\le 1-R-\varepsilon,$ an i.i.d. random Reed--Solomon code of rate $R$ over $\mathbb F_q$ satisfies, with high probability, affine $m$-space correlated agreement and degree-$\ell$ curve correlated agreement at radius $\rho n$, provided $q\ge n\cdot 2^{O_{R,m,\ell}(\varepsilon^{-3})}.$
\end{thm}

The precise thresholds, the output-radius parameter $E^+$, and the quantitative statements are given in Theorem~\ref{thm:iid-rs-ca-main}. In particular, the theorem includes the no-slack case $E^+=E$ in the i.i.d. random-RS model. For the standard distinct-evaluation random-RS model, the fixed-slack consequences $E^+=E+\Omega(n)$ transfer by the comparison lemma of~\cite{levi2025random}, as explained after Lemma~\ref{lem:distinct-rs}.

\begin{rmk}[Comparison with~\cite{goyal2025optimal} for random RS codes]
For random RS codes, the consequences obtained from~\cite{goyal2025optimal} give correlated agreement up to radius $\rho \le 1-R-2\varepsilon$. In the line/affine-space case, their alphabet-size requirement is $q\ge n\cdot 2^{O(\varepsilon^{-7})},$ and for degree-$\ell$ curves it is $q\ge n\cdot 2^{O_\ell(\varepsilon^{-(2\ell+5)})}.$

Our random-RS result reaches the larger radius $\rho\le 1-R-\varepsilon$. At this radius, Theorem~\ref{thm:iid-rs-ca-main} gives alphabet size $q\ge n\cdot 2^{O_m(\varepsilon^{-3})}$ for affine $m$-spaces, and $q\ge n\cdot 2^{O_\ell(\varepsilon^{-3})}$ for degree-$\ell$ curves.  Thus, compared with the random-RS consequences of~\cite{goyal2025optimal}, our result improves both the reachable radius and the alphabet-size dependence; see Remark~\ref{rmk:compare-random-rs} for the precise comparison.
\end{rmk}

\section{Preliminaries}
This section collects the notation and basic ingredients used throughout the paper.
The presentation is organized so as to simultaneously cover the large-alphabet radius
setting and the constant-alphabet radius setting.

\subsection{General Notation}

Let $q$ be a prime power, and let $\mathbb{F}_q$ denote the finite field of size $q$. For an integer $n \ge 1$, write $[n] := \{1,\ldots,n\}$. Let $H \in \mathbb{F}_q^{r\times n}$, and define the linear code
\[
\cC := \ker(H) = \{\mathbf{x} \in \mathbb{F}_q^n : H\mathbf{x} = 0\} \subseteq \mathbb{F}_q^n.
\]
When $H$ is sampled uniformly from $\mathbb{F}_q^{r\times n}$, the corresponding code
$\cC$ is a random linear code with rate $R=1-r/n$.

For a vector $\mathbf{x}=(x_1,\dots,x_n)\in\mathbb{F}_q^n$, let
\[
\supp(\mathbf{x}) := \{i\in[n] : x_i \neq 0\},
\qquad
\wt(\mathbf{x}) := |\supp(\mathbf{x})|.
\]
For $\mathbf{x},\mathbf{y}\in\mathbb{F}_q^n$, the Hamming distance is $d(\mathbf{x},\mathbf{y}) := \wt(\mathbf{x}-\mathbf{y})$. The minimum distance of the linear code $\cC$ is denoted by $d(\cC) := \min\{\wt(\mathbf{c}) : \mathbf{c}\in \cC\setminus\{0\}\}$.

For an integer $E\ge 0$, define the Hamming ball $B_E := \{\mathbf{x} \in \mathbb{F}_q^n : \wt(\mathbf{x})\le E\}$. More generally, for $\mathbf{y}\in\mathbb{F}_q^n$, let $B_E(\mathbf{y}) := \{\mathbf{x}\in\mathbb{F}_q^n : d(\mathbf{x},\mathbf{y})\le E\}$.

\begin{defn}[Syndrome ball]
For the parity-check matrix $H\in\mathbb{F}_q^{r\times n}$ and an integer $E\ge 0$, define the syndrome ball of radius $E$ by
\[
H_E := \{H\mathbf{x} : \mathbf{x}\in\mathbb{F}_q^n,\ \wt(\mathbf{x})\le E\}
\subseteq \mathbb{F}_q^r.
\]
\end{defn}

The following elementary equivalence will be used repeatedly.

\begin{lemma}\label{lem:syndrome-characterization}
For every $\mathbf{y}\in\mathbb{F}_q^n$ and every integer $E\ge 0$,
\[
d(\mathbf{y},C)\le E
\quad\Longleftrightarrow\quad
H\mathbf{y}\in H_E.
\]
\end{lemma}

\begin{proof}
By definition,
\[
d(\mathbf{y},\cC)\le E
\iff
\exists\,\mathbf{c}\in \cC \text{ such that } \wt(\mathbf{y}-\mathbf{c})\le E.
\]
Writing $\mathbf{x}:=\mathbf{y}-\mathbf{c}$, and using $\mathbf{c}\in \cC=\ker(H)$, we obtain $H\mathbf{x} = H(\mathbf{y}-\mathbf{c})=H\mathbf{y}$ with $\wt(\mathbf{x})\le E$ which is equivalent to $H\mathbf{y}\in H_E$. The converse is identical by reversing the argument.
\end{proof}

\subsection{Affine spaces and the proximity-gap property}

For $\mathbf{u}_0,\ldots,\mathbf{u}_{\ell}\in\mathbb{F}_q^n$, define the affine space
\[
U:=\mathbf{u}_0+\mathrm{span}\{\mathbf{u}_1,\ldots,\mathbf{u}_{\ell}\}\subseteq \F_q^n.
\]
Likewise, for $\mathbf{s}_0,\ldots,\mathbf{s}_{\ell}\in\mathbb{F}_q^r$, define the affine syndrome space
\[
S:=\mathbf{s}_0+\mathrm{span}\{\mathbf{s}_1,\ldots,\mathbf{s}_{\ell}\}\subseteq \F_q^r.
\]
Define $U_{H}:=H\mathbf{u}_0+\mathrm{span}\{H\mathbf{u}_1,\ldots,H\mathbf{u}_{\ell}\}$. If $H\mathbf{u}_i=\mathbf{s}_i$ for every $0\leq i\leq \ell$, then $U_H=S$.

Combining this with Lemma~\ref{lem:syndrome-characterization}, one can see that for every
$\alpha_1,\ldots,\alpha_{\ell}\in\F_q$,
\[
\mathbf{s}_0+\alpha_1 \mathbf{s}_1+\ldots+\alpha_{\ell}\mathbf{s}_{\ell} \in H_E
\quad\Longleftrightarrow\quad
d(\mathbf{u}_0+\alpha_1 \mathbf{u}_1+\ldots+\alpha_{\ell}\mathbf{u}_{\ell}, \cC)\le E.
\]
Moreover, we have $\frac{\left|\{\mathbf{u}\in U:d(\mathbf{u},\cC)\leq E\}\right|}{|U|}=\frac{|S\cap H_E|}{|S|}$.

\begin{defn}[$(E,E^+,\tau)$-space Proximity Gap]
\label{def:space-proximity-gap}
Let $0\leq E\leq E^{+}\leq n$ and $0\leq \tau\leq 1$. A linear code $\cC=\ker(H)$ is said to satisfy the $(E,E^{+},\tau)$-space proximity-gap property if every affine space $S\subseteq \F_q^r$ with $\frac{|S\cap H_E|}{|S|}>\tau$ implies $S\subseteq H_{E^{+}}$.
\end{defn}

\begin{lemma}\label{lem:line-to-space-gap}
Assume a linear code $\cC$ satisfies the following $(E,E^{+},\tau)$-line proximity-gap property. Then $\cC$ satisfies the $(E,E^+,\tau q/(q-1))$-space proximity-gap property.
\end{lemma}

\begin{proof}[Proof Sketch.] The argument is similar to that of Lemma 2.13 in~\cite{goyal2025optimal}, except that it is reformulated in the syndrome space; therefore, we defer it to Appendix~\ref{appendix:A}.
\end{proof}

\begin{defn}[($m, E, E^+, \tau$)-space Correlated Agreement]\label{def:correlated-agreement-space}
Let $m \ge 1$, $0 \le E \le E^+ \le n$, and $0 \le \tau \le 1$. A linear code $\cC$ with parity-check matrix $H$ is said to have the $(m,E,E^+,\tau)$-space correlated agreement property if for every affine space $U = \mathbf{u}_0 + \mathrm{span}(\mathbf{u}_1,\dots,\mathbf{u}_m) \subseteq \mathbb{F}_q^n$ with syndrome space
\[
S:= U_H= \mathbf{s}_0 + \mathrm{span}(\mathbf{s}_1,\dots,\mathbf{s}_m)\subseteq \mathbb{F}_q^r,
\qquad \mathbf{s}_i:=H\mathbf{u}_i \ \ (0\le i\le m),
\]
the condition 
$\frac{|S\cap H_E|}{|S|}>\tau$, implies that there exists an affine code space $V = \mathbf{c}_0 + \mathrm{span}(\mathbf{c}_1,\dots,\mathbf{c}_m)\subseteq \cC$
such that
\[
\left|\bigcup_{i=0}^{m}\mathrm{supp(\mathbf{c}_i-\mathbf{u}_i)}\right|
\le E^+.
\]
\end{defn}

\begin{lemma}\label{lem:correlated-line-to-space}
Assume that for every integer $E'$ with $0\le E'\le E$, a linear code $\cC$ has the $(1,E',E',\tau)$-line correlated agreement property. Assume further that for all $\mathbf{x}\in \F_q^n$,
\[
\bigl|\{\mathbf{c}\in \cC:\ d(\mathbf{x},\mathbf{c})\le E\}\bigr|<q.
\]
Then for every integer $m\ge 2$, the code $\cC$ has the $(m,E,E,\tau q/(q-1))$-space correlated agreement property.
\end{lemma}

\begin{proof}[Proof Sketch.]
     The argument is similar to that of Lemma 2.15 in~\cite{goyal2025optimal}, except that it is reformulated in the syndrome space; therefore, we defer it to Appendix~\ref{appendix:A}.
\end{proof}

\begin{lemma}[Corollary~1.4~\cite{alrabiah2025random}]\label{lem:list-decoding}
    For all $\varepsilon>0$, if $q\geq 2^{O(\frac{1}{\varepsilon^2})}$ and $n$ is sufficiently large, then a random linear code is, with high probability, $(1-R-\varepsilon,O(1/\varepsilon))$-average-radius-list-decodable.
\end{lemma}

Therefore, it suffices to consider the case of affine lines, since the case of affine spaces can be derived from the aforementioned theorem. For the remainder of this paper, we use the notation $\ell(\mathbf{a},\mathbf{b})=\{\mathbf{a}+\alpha \mathbf{b}:\alpha\in \F_q\}\subseteq \F_q^n$ for an affine line, and $L(\mathbf{s}_0,\mathbf{s}_1)=\{\mathbf{s}_0+\alpha \mathbf{s}_1:\alpha\in \F_q\}\subseteq \F_q^r$ for an affine syndrome line.

\section{Proximity Gap of Random Linear Codes}
\begin{lemma}\label{lem:degenerate-lines}
Let an affine syndrome line $L=L(\mathbf{s}_0,\mathbf{s}_1)\subseteq\mathbb{F}_q^r$. If $\dim_{\F_q} \mathrm{span}\{\mathbf{s}_0,\mathbf{s}_1\}\le 1$, then $|L\cap H_E|\in\{0,1,q\}$. In particular, if $|L\cap H_E|\ge 2$, then only the case $\dim_{\F_q} \mathrm{span}\{\mathbf{s}_0,\mathbf{s}_1\}=2$ can produce a bad line.
\end{lemma}

\begin{proof}
If $\mathbf{s}_1=0$, then $L$ is a singleton, so the claim is immediate.
Now assume $\mathbf{s}_1\neq 0$ and $\mathbf{s}_0=\beta \mathbf{s}_1$ for some $\beta\in\mathbb{F}_q$. Then $L=\{\gamma \mathbf{s}_1 : \gamma\in\mathbb{F}_q\}$.

Since $\mathbf{0}\in L$ and $\mathbf{0}=H\mathbf{0}\in H_E$, we always have $|L\cap H_E|\ge 1$.
If there exists $\gamma_0\neq 0$ and $\mathbf{x}\in B_E$ such that $H\mathbf{x}=\mathbf{\gamma}_0 \mathbf{s}_1$, then for every $\gamma\in\mathbb{F}_q$, the vector $\mathbf{x}_\gamma := \gamma\gamma_0^{-1}\mathbf{x}$ still belongs to $B_E$ and satisfies $H\mathbf{x}_\gamma=\gamma \mathbf{s}_1$. Thus $L\subseteq H_E$, and hence $|L\cap H_E|=q$. Otherwise, the only element of $L$ lying in $H_E$ is $0$, so $|L\cap H_E|=1$.
\end{proof}

Henceforth, whenever a syndrome line $L=\{\mathbf{s}_0+\alpha \mathbf{s}_1:\alpha\in\mathbb{F}_q\}$ is considered, it will be implicitly assumed that $\dim \mathrm{span}\{\mathbf{s}_0,\mathbf{s}_1\}=2$. We first show a bound on the number of elements of an affine line that can lie in a low-radius Hamming ball unless the entire line lies in a slightly larger ball.

The next lemma captures a simple counting principle in Hamming space. Along an affine line $\mathbf{a}+\alpha \mathbf{b}$, a fixed coordinate can vanish for at most one value of $\alpha$. Hence, if many elements on the line have small Hamming weight, there must be many such coordinate-wise cancellations. This can happen only to a limited extent unless the entire line is already contained in a slightly larger Hamming ball. This is the same basic idea used in the early $d/3$ proximity-gap argument of~\cite{ben2018worst}: at each coordinate where two affine descriptions differ, at most one element on the line can cancel the discrepancy.

\begin{lemma} \label{lem:gap-line-ball}
Let $\ell=\{\mathbf{a}+\alpha \mathbf{b}:\alpha\in\mathbb{F}_q\}\subseteq\mathbb{F}_q^n$, where $\mathbf{b}\neq \mathbf{0}$. Let $0\le E\le E^{+}\le n$. If $\ell\nsubseteq B_{E^{+}}$, 
then $|\ell\cap B_E|\le \frac{E^{+} + 1}{E^{+} - E + 1}$. When $E^{+}=E$, the statement reduces to if $\ell\nsubseteq B_E$, then $|\ell\cap B_E|\leq E+1$.
\end{lemma}

\begin{proof}
Although this result appears implicitly in other references, we present it here for completeness. 
Let $T:=\supp(\mathbf{b})$ and $t:=|\supp(\mathbf{a})\setminus T|$. For each $\alpha\in\mathbb{F}_q$, we have $\wt(\mathbf{a}+\alpha \mathbf{b})=t+\wt(\mathbf{b})-|\{i\in T:a_i+\alpha b_i=0\}|$ and $\wt(\mathbf{b})=|T|=\Sigma_{\alpha\in \F_q} |\{i\in T:a_i+\alpha b_i=0\}|$. Since $\ell\nsubseteq B_{E^{+}}$, there exists $\alpha_0\in\mathbb{F}_q$ such that $\wt(\mathbf{a}+\alpha_0 \mathbf{b})>E^{+}$,
which implies
\[
t+\wt(\mathbf{b})-|\{i\in T:a_i+\alpha_0 b_i=0\}|>E^{+},
\qquad\text{hence}\qquad
t+\wt(\mathbf{b})\ge E^{+}+1.
\]
For any $\mathbf{a}+\alpha\mathbf{b}\in \ell\cap B_E$, we have $|\{i\in T:a_i+\alpha b_i=0\}|\ge t+\wt(\mathbf{b})-E$. 
Therefore,
\[
\wt(\mathbf{b})=\sum_{\alpha\in\mathbb{F}_q} |\{i\in T:a_i+\alpha b_i=0\}|
\ge \sum_{\mathbf{a}+\alpha\mathbf{b}\in \ell\cap B_E} |\{i\in T:a_i+\alpha b_i=0\}|
\ge |\ell\cap B_E|(t+\wt(\mathbf{b})-E).
\]
It follows that $|\ell\cap B_E|\le \frac{t+\wt(\mathbf{b})}{t+\wt(\mathbf{b})-E}$. Since the function $x\mapsto x/(x-E)$ is decreasing for $x>E$, and $t+\wt(\mathbf{b})\ge E^{+}+1$, we obtain
\[
|\ell\cap B_E|
\le
\frac{t+\wt(\mathbf{b})}{t+\wt(\mathbf{b})-E}
\le
\frac{E^{+} + 1}{E^{+} - E + 1}.
\]
This completes the proof.
\end{proof}

Based on the proof of Lemma~\ref{lem:gap-line-ball}, we prove in Appendix~\ref{appendix:A} that, for the proximity-gap property over constant-size alphabets, a positive slack term, namely $E^{+}>E$ is necessary.

\begin{thm}\label{thm:no-slack-obstruction}
Let $E,K$ be integers satisfying $1 \le K \le E+1 < n$ and assume $K<q$. Fix pairwise distinct elements $\alpha_1,\dots,\alpha_K \in \F_q$. Then there exist vectors $\mathbf{x}_1,\mathbf{x}_2 \in \F_q^n$ such that $\wt(\mathbf{x}_1+\alpha_j \mathbf{x}_2)=E$ for every $j\in [K]$ and $\wt(\mathbf{x}_1+\alpha \mathbf{x}_2)=E+1$ for every other $\alpha\in \F_q$.

Consequently, for every linear code $\cC\subseteq \F_q^n$ with parity-check matrix $H$ and minimum distance $d(\cC)\geq 2E+2$, the associated syndrome line $L:=\{Hx_1+\alpha Hx_2:\alpha\in\mathbb F_q\}\subseteq \mathbb F_q^r$ satisfies $|L\cap H_E|\ge K$ and $L\not\subseteq H_E$. In particular, $L$ violates the strong $(E,E,(K-1)/q)$ proximity-gap conclusion.
\end{thm}

\begin{defn}[Witness matrix]\label{def:witness-matrix}
Fix a nondegenerate affine syndrome line $L=\{\mathbf{s}_0+\alpha \mathbf{s}_1:\alpha\in\mathbb{F}_q\}\subseteq\mathbb{F}_q^r$ and pairwise distinct elements $\alpha_1,\dots,\alpha_K\in\mathbb{F}_q$.
A matrix $X=[\mathbf{x}_1|\cdots|\mathbf{x}_K]\in\mathbb{F}_q^{n\times K}$ is called an $E$-witness matrix for $(L;\alpha_1,\dots,\alpha_K)$ if, for every $j\in [K]$,
\[
H\mathbf{x}_j=\mathbf{s}_0+\alpha_j \mathbf{s}_1,
\qquad
\wt(\mathbf{x}_j)\le E.
\]
We say that a matrix $X=[\mathbf{x}_1|\cdots|\mathbf{x}_K]\in\mathbb{F}_q^{n\times K}$ is an $E$-witness matrix covering $K$ elements of $L$ if there exist pairwise distinct elements $\alpha_1,\ldots,\alpha_K\in\mathbb F_q$ such that
\[
H\mathbf{x}_j=\mathbf{s}_0+\alpha_j \mathbf{s}_1,\qquad\wt(\mathbf{x}_j)\le E,\qquad j\in[K].
\]
\end{defn}

\begin{thm}\label{thm:rank-reduction}
Let $\cC\subseteq\F_q^n$ be a linear code with minimum distance $d(\cC)\ge d$ and $0< E<d$, $\gamma:=\frac{d-E}{d}\in(0,1)$. Let $L=\{\mathbf{s}_0+\alpha \mathbf{s}_1:\alpha\in \F_q\}\subseteq \F_q^r$ be a nondegenerate affine syndrome line. Suppose that $X=[\mathbf{x}_1|\ldots|\mathbf{x}_K]\in \F_q^{n\times K}$ is an $E$-witness matrix covering $K$ elements of $L$ and $\mathrm{rank}(X)=t\geq 3$. Then there exists a subset $J\subseteq[K]$ such that $|J|\ge K\gamma$ and the submatrix $X_J := [\mathbf{x}_j]_{j\in J}$ is still an $E$-witness matrix covering $|J|$ elements of the same syndrome line $L$ with $\rank(X_J)\le t-1$.
\end{thm}

\begin{proof}
Since $X$ is an $E$-witness matrix covering $K$ elements of $L$, there exist pairwise distinct elements
$\alpha_1,\ldots,\alpha_K\in\mathbb F_q$ such that $H\mathbf{x}_j=\mathbf{s}_0+\alpha_j \mathbf{s}_1$ and $\mathrm{wt}(\mathbf{x}_j)\leq E$ for all $j\in[K]$. Define $Y:=[\mathbf{s}_0+\alpha_1 \mathbf{s}_1\mid\cdots\mid \mathbf{s}_0+\alpha_K \mathbf{s}_1]\in\mathbb{F}_q^{r\times K}$, and let $\mathbf{u}:=(1,\dots,1)\in\mathbb{F}_q^K,\mathbf{v}:=(\alpha_1,\dots,\alpha_K)\in\mathbb{F}_q^K$. Since the elements $\alpha_1,\ldots,\alpha_K$ are pairwise distinct and $\dim_{\F_q}\mathrm{span}\{\mathbf{s}_0,\mathbf{s}_1\}=2$, we have $Y=\mathbf{s}_0\cdot \mathbf{u}^\top+\mathbf{s}_1\cdot \mathbf{v}^\top$ and $\mathrm{Row}(Y)=\mathrm{span}\{\mathbf{u},\mathbf{v}\}\subseteq \mathrm{Row}(X)$. 

Extend $\{\mathbf{u},\mathbf{v}\}$ to a basis $\mathbf{u},\ \mathbf{v},\ \mathbf{w}_3,\dots,\mathbf{w}_t$ of $\mathrm{Row}(X)$. Then there exist vectors $\mathbf{a},\mathbf{b},\mathbf{c}_3,\dots,\mathbf{c}_t\in\F_q^n$ such that $X = \mathbf{a}\cdot \mathbf{u}^\top + \mathbf{b}\cdot \mathbf{v}^\top + \sum_{\ell=3}^t \mathbf{c}_\ell\cdot \mathbf{w}_\ell^\top$. Since $HX=Y$, it follows that
\[
(H\mathbf{a}-\mathbf{s}_0)\cdot \mathbf{u}^\top + (H\mathbf{b}-\mathbf{s}_1)\cdot \mathbf{v}^\top
+ \sum_{\ell=3}^t (H\mathbf{c}_\ell)\cdot \mathbf{w}_\ell^\top = \mathbf{0}.
\]
By linear independence, $H\mathbf{a}=\mathbf{s}_0, H\mathbf{b}=\mathbf{s}_1, H\mathbf{c}_\ell=\mathbf{0}\quad(\ell=3,\dots,t)$, so each $\mathbf{c}_\ell$ is a codeword. If $\mathbf{c}_t=\mathbf{0}$, then $\rank(X)\le t-1$, and there is nothing to prove. Thus, we assume $\mathbf{c}_t\neq \mathbf{0}$. Let $S:=\supp(\mathbf{c}_t)$ and $|S|=\wt(\mathbf{c}_t)\geq d$.

Each column $\mathbf{x}_j$ has weight at most $E$, and therefore has at least $|S|-E$ zero coordinates in $S$. By the pigeonhole principle, there exists $h\in S$ such that
\[
|\{j\in[K]:(\mathbf{x}_j)_h=0\}|
\ge
K\cdot \frac{|S|-E}{|S|}
\ge
K\cdot \frac{d-E}{d}
=
K\gamma.
\]
Let $J:=\{j\in[K]:(\mathbf{x}_j)_h=0\}$ and $|J|\ge K\gamma$. For each $j \in J$, the $j$-th column of $X$ is $\mathbf{x}_j = \mathbf{a}+\alpha_j \mathbf{b}+\sum_{\ell=3}^{t-1}(\mathbf{w}_\ell)_j \cdot \mathbf{c}_\ell + (\mathbf{w}_t)_j\cdot \mathbf{c}_t$. Since $(\mathbf{x}_j)_h=0$ for every $j\in J$ and $(\mathbf{c}_t)_h\neq 0$, we may solve for $(\mathbf{w}_t)_j$:
\[
(\mathbf{w}_t)_j
=
-\frac{a_h+\alpha_j b_h+\sum_{\ell=3}^{t-1}(\mathbf{w}_\ell)_j\cdot (\mathbf{c}_\ell)_h}{(\mathbf{c}_t)_h}.
\]
Substituting this into the above expression gives, for every $j\in J$,
\[
\mathbf{x}_j
=
\Bigl(\mathbf{a}-\frac{a_h}{(\mathbf{c}_t)_h}\mathbf{c}_t\Bigr)
+\alpha_j\Bigl(\mathbf{b}-\frac{b_h}{(\mathbf{c}_t)_h}\mathbf{c}_t\Bigr)
+\sum_{\ell=3}^{t-1}(\mathbf{w}_\ell)_j
\Bigl(\mathbf{c}_\ell-\frac{(\mathbf{c}_\ell)_h}{(\mathbf{c}_t)_h}\mathbf{c}_t\Bigr).
\]
Define
\[
\mathbf{a}':=\mathbf{a}-\frac{\mathbf{a}_h}{(\mathbf{c}_t)_h}\mathbf{c}_t,\qquad
\mathbf{b}':=\mathbf{b}-\frac{\mathbf{b}_h}{(\mathbf{c}_t)_h}\mathbf{c}_t,\qquad
\mathbf{c}'_\ell:=\mathbf{c}_\ell-\frac{(\mathbf{c}_{\ell})_h}{(\mathbf{c}_t)_h}\mathbf{c}_t
\quad (\ell=3,\dots,t-1).
\]
Then for all $j\in J$ we have $\mathbf{x}_j = \mathbf{a}' + \alpha_j \mathbf{b}' + \sum_{\ell=3}^{t-1}(\mathbf{w}_\ell)_j\cdot \mathbf{c}_\ell'$.
This means that let $\mathbf{u}_J,\mathbf{v}_J,(\mathbf{w}_{\ell})_J\in\mathbb F_q^{|J|}$ denote the restrictions of $\mathbf{u},\mathbf{v},\mathbf{w}_\ell$ to the coordinate set $J$, then $X_J= \mathbf{a}'\cdot \mathbf{u}_J^\top + \mathbf{b}'\cdot \mathbf{v}_J^\top + \sum_{\ell=3}^{t-1} \mathbf{c}_\ell'\cdot (\mathbf{w}_{\ell})_J^\top$. Thus, the row vectors in $X_J$ are the linear combinations of $\mathbf{u}_J,\mathbf{v}_J,(\mathbf{w}_{3})_J,\dots,(\mathbf{w}_{t-1})_J$, and therefore $\mathrm{rank}(X_J)\le t-1.$

Moreover, since $H\mathbf{c}_t=\mathbf{0}$, we have $H\mathbf{a}'=H\mathbf{a}=\mathbf{s}_0,\, H\mathbf{b}'=H\mathbf{b}=\mathbf{s}_1,\, H\mathbf{c}_\ell'=H\mathbf{c}_\ell=\mathbf{0}$. Since the weight of each column in $X_J$ is at most $E$, $X_J$ remains an $E$-witness matrix for the same syndrome line $L$.
\end{proof}

\begin{cor} \label{cor:rank-two-extraction}
Under the assumptions of~\cref{thm:rank-reduction}, suppose that $\rank(X)=t\ge 2$. Then there exists a subset $J\subseteq[K]$ with $|J|\ge K\gamma^{\,t-2}$ and vectors $\widetilde{\mathbf{a}},\widetilde{\mathbf{b}}\in\mathbb{F}_q^n$, with $\widetilde{\mathbf{b}}\neq \mathbf{0}$, such that for all $j\in J$, $\mathbf{x}_j=\widetilde{\mathbf{a}}+\alpha_j \widetilde{\mathbf{b}},\,H\widetilde{\mathbf{a}}=\mathbf{s}_0,\,H\widetilde{\mathbf{b}}=\mathbf{s}_1$. In particular, the surviving witness vectors lie on an affine line $\widetilde\ell:=\{\widetilde{\mathbf{a}}+\alpha \widetilde{\mathbf{b}}:\alpha\in\mathbb{F}_q\}$
satisfying $\widetilde\ell_H=L$.

\end{cor}

\begin{proof}
If $t=2$, the conclusion is immediate. Otherwise, apply
Lemma~\ref{thm:rank-reduction} repeatedly until the rank drops to $2$.
Each step preserves at least a $\gamma$-fraction of the columns, so after $t-2$ iterations,
at least $K\gamma^{t-2}$ columns remain. When we arrive at rank $2$, by the proof of Lemma~\ref{thm:rank-reduction}, the witness vectors are of the form $\mathbf{x}_j=\widetilde{\mathbf{a}}+\alpha_j\widetilde{\mathbf{b}}$, with $H\widetilde{\mathbf{a}}=\mathbf{s}_0$ and $H\widetilde{\mathbf{b}}=\mathbf{s}_1$.
\end{proof}

\begin{thm}\label{thm:min-rank-threshold}
Let $\cC$ be a linear code with parity-check matrix $H$ and minimum distance $d(\cC)\ge d$, and let $0<E\leq E^{+}<d$, $K\ge 2$. Let $\gamma:=\frac{d-E}{d},\,
B_{E,E^{+}}:=
\left\lfloor \frac{E^{+} + 1}{E^{+} - E + 1}\right\rfloor.$
Let $L=\{\mathbf{s}_0+\alpha \mathbf{s}_1:\alpha \in \F_q\}\subseteq\mathbb{F}_q^r$ be a nondegenerate syndrome line. Suppose there exists an $E$-witness matrix $X$ of rank $t$ covering $K$ elements of $L$ as in Definition~\ref{def:witness-matrix}. If $L\nsubseteq H_{E^+}$, then $K\cdot \gamma^{\,t-2}\le B_{E,E^{+}}$.
\end{thm}

\begin{proof}
By Corollary~\ref{cor:rank-two-extraction}, we can find pairwise distinct elements $\alpha_1,\ldots,\alpha_K\in\mathbb F_q$ associated with the $K$ elements covered by $X=[\mathbf{x}_1|\ldots|\mathbf{x}_K]$, and a subset $J\subseteq[K]$ with $|J|\ge K\gamma^{t-2}$, such that the corresponding witness vectors $\mathbf{x}_j$ lie on an affine line $\widetilde\ell=\{\widetilde{\mathbf{a}}+\alpha\widetilde{\mathbf{b}}:\alpha\in\mathbb{F}_q\}$ with $\widetilde\ell_H=L$ and for every $j\in J$, $\wt(\widetilde{\mathbf{a}}+\alpha_j\widetilde{\mathbf{b}})\le E$. Thus, $|\widetilde\ell\cap B_E|\ge |J|\ge K\gamma^{\,t-2}$.

If $\widetilde\ell\subseteq B_{E^{+}}$, then $L=\widetilde\ell_H\subseteq H_{E^+},$ contrary to the assumption. Hence $\widetilde\ell\nsubseteq B_{E^{+}}$, and
Lemma~\ref{lem:gap-line-ball} implies $K\gamma^{\,t-2}\le|\widetilde\ell\cap B_E|\le B_{E,E^{+}}$ as claimed.
\end{proof}

Before proceeding to the main probabilistic analysis, we present some auxiliary lemmas in Appendix~\ref{appendix:B} that will be useful for probability estimates below. Our following theorem utilizes Theorem~\ref{thm:min-rank-threshold} to derive a union bound for bad syndrome lines.

\begin{thm}\label{thm:bad-line-moment}
Let $0<E\le E^+$ and $E<d$. Let $s\ge 0$ be an integer, and define $\gamma:=\frac{d-E}{d},\,B_{E,E^+}:=\left\lfloor \frac{E^+ + 1}{E^+ - E + 1}\right\rfloor$. Assume that $K\ge 2$ with $K > B_{E,E^+}\gamma^{-s}$. Let $\cC$ be a random linear code with parity-check matrix $H$. Then the probability that there exists an affine syndrome line $L\subseteq \F_q^r$ such that $|L\cap H_E|\geq K$, $L\nsubseteq H_{E^+}$ and $d(\cC)\ge d$ is at most $q^{2r}\binom{q}{K}\binom{K}{s+3}|B_E|^{s+3}q^{-r(s+3)}$.
\end{thm}

\begin{proof}
Fix an affine syndrome line $L=\{\mathbf{s}_0+\alpha \mathbf{s}_1:\alpha\in\mathbb F_q\}\subseteq \mathbb F_q^r$ and  a $K$-element subset $\{\alpha_1,\ldots,\alpha_K\}\subseteq\mathbb F_q$. Define the event
\[
\mathcal{A}(L;\alpha_1,\dots,\alpha_K)
:=
\bigl[\forall j\in[K],\ \mathbf{s}_0+\alpha_j \mathbf{s}_1\in H_E\bigr]
\wedge
\bigl[L\not\subseteq H_{E^+}\bigr]
\wedge
[d(\cC)\ge d].
\]
If $\mathcal A(L;\alpha_1,\dots,\alpha_K)$ occurs, then there exist vectors $\mathbf{x}_1,\dots,\mathbf{x}_K\in B_E$ such that $H\mathbf{x}_j=\mathbf{s}_0+\alpha_j \mathbf{s}_1$ where $j\in [K]$.

Since $K>B_{E,E^{+}}\cdot \gamma^{-s}$, Theorem~\ref{thm:min-rank-threshold} implies that no  matrix of rank at most $s+2$ is an $E$-witness matrix for $(L,\alpha_1,\ldots,\alpha_K)$. Hence every such $E$-witness matrix must satisfy $\mathrm{rank}(X)\ge s+3$. In particular, among the $K$ columns of $X$, we can choose a subset $J\subseteq[K]$ of size $s+3$ such that the columns $\{\mathbf{x}_j:j\in J\}$ are linearly independent. There are at most $\binom{K}{s+3}|B_E|^{s+3}$ possible choices for these columns.

By Lemma~\ref{lem:random-image}, the probability that these $s+3$ linearly independent vectors were mapped by $H$ to the prescribed syndrome columns $\{\mathbf{s}_0+\alpha_j \mathbf{s}_1:j\in J\}$ is $q^{-r(s+3)}$. Therefore,
\[
\Pr_{\cC}\bigl[\mathcal A(L;\alpha_1,\dots,\alpha_K)\bigr]
\le
\binom{K}{s+3}|B_E|^{s+3}q^{-r(s+3)}.
\]
Now we fix the line $L$. If $|L\cap H_E|\geq K,\, L\not\subseteq H_{E^+}$ and $d(\cC)\ge d$, then we can choose a $K$-element subset $\{\alpha_1,\dots,\alpha_K\}\subseteq \mathbb{F}_q$ such that $\mathbf{s}_0+\alpha_j \mathbf{s}_1\in H_E$ where $j\in[K]$.

Taking a union bound over all $\binom{q}{K}$ possible $K$-subsets $\{\alpha_1,\ldots,\alpha_K\}$ of $\mathbb{F}_q$ and all affine lines $L\subseteq \mathbb{F}_q^r$, the failure probability is at most 
\[
\sum_{L,(\alpha_1,\ldots,\alpha_K)}\Pr_{\cC}\!\bigl[\mathcal{A}(L;\alpha_1,\dots,\alpha_K)\bigr]\leq q^{2r}\binom{q}{K}\binom{K}{s+3}|B_E|^{s+3}q^{-r(s+3)}.
\]
This completes the proof.
\end{proof}

We are ready to present our proximity gap theorem for random linear codes.

\begin{thm}\label{thm:explicit-two-radius}
Fix $R\in(0,1)$ and $0<\varepsilon<(1-R)/2$. Define $a_\varepsilon:=\frac{\varepsilon}{\log_2(1/\varepsilon)}$. Assume that $0<\rho<1-R-\varepsilon-a_\varepsilon$. Let $\delta=1-R-a_\varepsilon,\,\ell=\left\lceil \frac{2(1-R)}{\varepsilon}\right\rceil-1,\,K=\left\lceil \left(1+\frac{\rho}{\varepsilon}\right) \left(\frac{\delta}{\delta-\rho}\right)^\ell \right\rceil,\,r=(1-R)n,\,E=\lfloor \rho n\rfloor,\,E^+:=E+\lceil \varepsilon n\rceil$ and let $q$ be a fixed prime power such that $q\ge \max\{\left(\frac{2}{\varepsilon}\right)^{1/\varepsilon},\ K+1\}$.

Then, with probability at least $1-q^{-\Omega(n)}$, a random linear code $\cC$ of rate $R$ and length $n$ has minimum distance at least $\lfloor\delta n\rfloor$ and satisfies the $(E,E^+,K/q)$-line proximity-gap property.

\end{thm}

\begin{proof}
Define $b_\varepsilon:=\frac{\varepsilon}{1+\log_2(1/\varepsilon)}$. Let $d:=\lfloor \delta n\rfloor,\,\gamma:=\frac{d-E}{d},\,B_{E,E^+}:=\left\lfloor \frac{E^+ + 1}{E^+ - E + 1}\right\rfloor$. Since $\delta-(\rho+\varepsilon)=1-R-\rho-\varepsilon-a_\varepsilon>0$, we have $E^+<d$ for all sufficiently large $n$.

We first estimate the probability that $\cC$ has minimum distance less than $d$. By Lemma~\ref{lem:H_q-bound} and Lemma~\ref{lem:ball-volume-bound}, $|B_{d}| \le q^{(\delta+b_\varepsilon)n+o(n)}$. Therefore, $\Pr_{\cC}[d(\cC)\le d] \le |B_{d}|\,q^{-r} \le q^{(\delta+b_\varepsilon-(1-R))n+o(n)}$. Since $\delta+b_\varepsilon-(1-R)=-(a_\varepsilon-b_\varepsilon)$, we obtain
\[
\Pr_{\cC}[d(\cC)\le d]
\le q^{-(a_\varepsilon-b_\varepsilon)n+o(n)}
= q^{-\Omega(n)}.
\]

Next we estimate the bad-line probability on the event $d(\cC)\ge d$. As
\[
B_{E,E^+}
=
\left\lfloor
\frac{E+\lceil \varepsilon n\rceil+1}{\lceil \varepsilon n\rceil+1}
\right\rfloor
\le
1+\frac{E}{\lceil \varepsilon n\rceil+1}
\le
1+\frac{\rho}{\varepsilon},
\]
and $\gamma^{-1}=\frac{d}{d-E}\le \frac{\delta}{\delta-\rho}$, we obtain $B_{E,E^+}\cdot \gamma^{-\ell}
\le
\left(1+\frac{\rho}{\varepsilon}\right)
\left(\frac{\delta}{\delta-\rho}\right)^\ell< K+1$. We next bound the
probability that there exists affine syndrome line $L\subseteq \F_q^r$ such that
\[
\left[|L\cap H_E|\geq K+1\right]\wedge \left[L\nsubseteq H_{E^+}\right]\wedge \left[d(\cC)\geq d\right].
\] 
Applying Theorem~\ref{thm:bad-line-moment} with $d=d$, $s=\ell$, the probability is at most 
%\[
%\sum_{L,(\alpha_1,\ldots,\alpha_K)}\Pr_{\cC}\!\bigl[\mathcal{A}(L;\alpha_1,\dots,\alpha_K)\bigr]\leq %q^{2r}\binom{q}{K+1}\binom{K+1}{\ell+3}|B_E|^{\ell+3}q^{-r(\ell+3)}.
%\]
$q^{2r}\binom{q}{K+1}\binom{K+1}{\ell+3}|B_E|^{\ell+3}q^{-r(\ell+3)}.$

Since $K$ and $q$ depend only on $(R,\rho,\varepsilon)$, the two binomial factors contribute
only $q^{O_{R,\rho,\varepsilon}(1)}$.
On the other hand, Lemma~\ref{lem:ball-volume-bound} gives $|B_E|\le q^{(\rho+b_\varepsilon)n+o(n)}$. Hence the exponent in the preceding probability bound is at most $2(1-R)n+(\ell+3)(\rho+b_\varepsilon-(1-R))n+o(n)$. By the assumption on $\rho$, $1-R-\rho-b_\varepsilon > \varepsilon+(a_\varepsilon-b_\varepsilon) >\varepsilon$. Since $\ell+3=\left\lceil \frac{2(1-R)}{\varepsilon}\right\rceil+2>\frac{2(1-R)}{\varepsilon}$, it follows that $(\ell+3)(1-R-\rho-b_\varepsilon)>2(1-R)$. Thus, $2(1-R)+(\ell+3)(\rho+b_\varepsilon-(1-R))<0$,
and therefore the probability is at most $q^{-\Omega(n)}$.

Combining this with the minimum-distance bound, we conclude that with probability at least $1-q^{-\Omega(n)}$, every affine syndrome line $L\subseteq \F_q^r$ satisfies: if $L\not\subseteq H_{E^+}$, then $|L\cap H_E|\le K$. This is exactly the $(E,E^+,K/q)$-line proximity gap property, and the $(E,E^+,K/(q-1))$-space version follows from Lemma~\ref{lem:line-to-space-gap}.
\end{proof}

\begin{rmk}
A simpler but slightly coarser sufficient condition is $q\ge \left(\frac{2}{\varepsilon}\right)^{2/\varepsilon}$, since for fixed $R$ and sufficiently small $\varepsilon$, this implies $q> K+1$.
\end{rmk}

Taking $E^+=E$ yields the corresponding no-slack statement in the large-alphabet setting, where the required alphabet size grows linearly with $n$.

\begin{cor}
Fix $R \in (0,1)$ and $0<\varepsilon < (1-R)/2$. Assume that $0<\rho<1-R-\varepsilon$. Let $\delta := 1-R-\varepsilon, \ell := \left\lceil \frac{2(1-R)}{\varepsilon}\right\rceil - 1,E := \lfloor \rho n\rfloor,d := \lfloor \delta n\rfloor$, and define $K :=\left\lfloor(E+1)\left(\frac{d}{d-E}\right)^\ell \right\rfloor$. Let $q$ be a prime power such that $q=\Theta(n)$ and $
q \ge \max\!\left\{
\left(\frac{2}{\varepsilon}\right)^{1/\varepsilon},\,
K+2
\right\}$.

Then, with probability at least $1-q^{-\Omega(n)}$, a random linear code $\cC$ of rate $R$ and code length $n$ has minimum distance at least $d$ and satisfies the $(E,E,K/q)$-line proximity gap property.
\end{cor}

\begin{proof}
This is the specialization of Theorem~\ref{thm:bad-line-moment} to $E^{+}=E$, for which $B_{E,E}=E+1$. As $K+1>(E+1)\left(\frac{\delta-\rho}{\delta}\right)^{-\ell}$, every bad line with at least $K+1$ elements in $H_E$ must have witness rank at least $\ell+3$.

As above, since $q\ge \left(\frac{2}{\varepsilon}\right)^{1/\varepsilon}$, Lemma~\ref{lem:ball-volume-bound} gives $|B_{d}|\le q^{d+o(n)}< q^{(1-R-\varepsilon)n+o(n)}$, and therefore $\Pr_{\cC}[d(\cC)\le d]
\le
|B_{d}|\,q^{-r}
<
q^{-\varepsilon n+o(n)}$. We bound the probability that there exists affine syndrome line $L\subseteq \F_q^r$ such that
\[
\left[|L\cap H_E|\ge K+1\right]\wedge \left[L\nsubseteq H_E\right]\wedge \left[d(\cC)\geq d\right].
\]
By Theorem~\ref{thm:bad-line-moment} with $d=d$, $s=\ell$, the probability is at most $q^{2r}\binom{q}{K+1}\binom{K+1}{\ell+3}|B_E|^{\ell+3}q^{-r(\ell+3)}$.

Because $q=\Theta(n)$ and $K=\Theta(n)$, we have $\log_q \binom{q}{K+1}=o(n),\,\log_q \binom{K+1}{\ell+3}=o(n)$. Hence the exponent in the above bound is at most $2(1-R)n+(\ell+3)(\rho-(1-R))n+o(n)$. Since $\rho<1-R-\varepsilon$, the exponent is bounded above by $\bigl(2(1-R)-(\ell+3)\varepsilon\bigr)n+o(n)$. By the definition of $\ell$, we have $\ell+3=\Bigl\lceil \frac{2(1-R)}{\varepsilon}\Bigr\rceil+2>\frac{2(1-R)}{\varepsilon}$, and hence $2(1-R)-(\ell+3)\varepsilon<0$. Combining this with the minimum-distance bound proves the claim.
\end{proof}

\begin{rmk}
    In particular, $K \le (\rho n+1)\left(\frac{\delta n}{(\delta-\rho)n-1}\right)^\ell = \Theta(n),$ so any prime power $q \ge cn$ with $c>\rho\left(\frac{1-R-\varepsilon}{1-R-\varepsilon-\rho}\right)^\ell$ is sufficient for all large $n$.
\end{rmk}

\section{Correlated Agreement of Random Linear Codes}
%\begin{comment}
\label{section:correlated-agreement}
\noindent
For a matrix $X=[\mathbf{x}_0|\cdots|\mathbf{x}_m]\in \mathbb{F}_q^{n\times (m+1)}$, define its row support and row weight by
\[
\mathrm{rowsupp}(X):=\{i\in [n]: (\mathbf{x}_j)_i\neq 0 \text{ for some }0\le j\le m\},
\qquad
\mathrm{rowwt}(X):=|\mathrm{rowsupp}(X)|.
\]
Likewise, for an affine space $U\subseteq \mathbb{F}_q^n$, define $\mathrm{Supp}(U):=\bigcup_{\mathbf{u}\in U}\mathrm{supp}(\mathbf{u})$.

\begin{defn}[$(m,E,E^+,\tau)$ correlated agreement]
Let $m\ge 1$, $0\le E\le E^+\le n$, and $0\le \tau\le 1$. We say that a linear code $\cC$ with parity-check matrix $H$ has the $(m,E,E^+,\tau)$-correlated agreement property if, for every affine syndrome space $S=\mathbf{s}_0+\mathrm{span}(\mathbf{s}_1,\dots,\mathbf{s}_m)\subseteq \mathbb{F}_q^r$ of dimension $m$, the condition $\frac{|S\cap H_E|}{|S|}>\tau$, implies that there exists a matrix $X=[\mathbf{x}_0|\cdots|\mathbf{x}_m]\in \mathbb{F}_q^{n\times (m+1)}$
such that $H\mathbf{x}_i=\mathbf{s}_i\ (0\le i\le m)$ and $\mathrm{rowwt}(X)\le E^+$. 

For a fixed affine syndrome space $S=\mathbf{s}_0+\mathrm{span}(\mathbf{s}_1,\ldots,\mathbf{s}_m)$, if no such matrix $X$ exists, we say that $S$ has no $(H,m,E,E^+)$-correlated agreement.
\end{defn}

The next lemma shows that this is exactly the syndrome-space reformulation of Definition~\ref{def:correlated-agreement-space} for affine spaces of dimension $m$.

\begin{lemma}\label{lem:syndrome-CA-reformulation}
Let $U=\mathbf{u}_0+\mathrm{span}(\mathbf{u}_1,\dots,\mathbf{u}_m)\subseteq \mathbb{F}_q^n,\, S=U_H=\mathbf{s}_0+\mathrm{span}(\mathbf{s}_1,\dots,\mathbf{s}_m)\subseteq \mathbb{F}_q^r,$ where $\mathbf{s}_i=H\mathbf{u}_i$ for $0\le i\le m$. Then the following are equivalent:
\begin{enumerate}
\item there exists an affine code space $V=\mathbf{c}_0+\mathrm{span}(\mathbf{c}_1,\dots,\mathbf{c}_m)\subseteq \cC$ such that
\[
\left|\bigcup_{j=0}^{m} \mathrm{supp}(\mathbf{c}_j-\mathbf{u}_j)\right|\le E^+;
\]
\item there exists a matrix $X=[\mathbf{x}_0|\cdots|\mathbf{x}_m]\in \mathbb{F}_q^{n\times (m+1)}$ such that $H\mathbf{x}_i=\mathbf{s}_i\ (0\le i\le m)$ and $\mathrm{rowwt}(X)\le E^+.$
\end{enumerate}
\end{lemma}

\begin{proof}
Assume (1). Define $\mathbf{x}_j:=\mathbf{u}_j-\mathbf{c}_j,\, 0\le j\le m$. As $\mathbf{c}_j\in \cC=\ker(H)$, we have $H\mathbf{x}_j=H(\mathbf{u}_j-\mathbf{c}_j)=H\mathbf{u}_j=\mathbf{s}_j$ for every $0\le j\le m$. Moreover,
\[
\mathrm{rowsupp}(X)
= \bigcup_{j=0}^{m}\mathrm{supp}(\mathbf{x}_j)=\bigcup_{j=0}^{m} \mathrm{supp}(\mathbf{c}_j-\mathbf{u}_j)
\]
so $\mathrm{rowwt}(X)\le E^+$.

Conversely, assume (2), and define $\mathbf{c}_j:=\mathbf{u}_j-\mathbf{x}_j,\, 0\le j\le m$. Then $H\mathbf{c}_j=H\mathbf{u}_j-H\mathbf{x}_j=\mathbf{s}_j-\mathbf{s}_j=\mathbf{0}$, so each $\mathbf{c}_j\in \cC$. Hence $V:=\mathbf{c}_0+\mathrm{span}(\mathbf{c}_1,\dots,\mathbf{c}_m)\subseteq \cC$. Also,
\[
\bigcup_{j=0}^{m} \mathrm{supp}(\mathbf{c}_j-\mathbf{u}_j)=\bigcup_{j=0}^{m}\mathrm{supp}(\mathbf{x}_j)=\mathrm{rowsupp}(X)
\]
so the disagreement support has size at most $E^+$. This proves the equivalence.
\end{proof}

We next generalize Lemma~\ref{lem:gap-line-ball} from affine lines to affine spaces of dimension $m$.

\begin{lemma}\label{lem:affine-space-ball}
Let $U = \left\{ \mathbf{u}_0+\sum_{j=1}^m \alpha_j \mathbf{u}_j: \mathbf{\alpha}=(\alpha_1,\dots,\alpha_m)\in \mathbb{F}_q^m \right\} \subseteq \mathbb{F}_q^n$. If $|\mathrm{Supp}(U)|>E^+$, then $|U\cap B_E| \le \frac{E^+ +1}{E^+ - E +1}\cdot q^{m-1}$.
\end{lemma}

\begin{proof}
Let $T:=\mathrm{Supp}(U)$ with $N:=|T|$. For each $i\in T$, define the affine linear form $f_i(\mathbf{\alpha}):=(\mathbf{u}_0)_i+\sum_{j=1}^m \alpha_j (\mathbf{u}_j)_i$ where $\mathbf{\alpha}=(\alpha_1,\ldots,\alpha_m)\in \mathbb{F}_q^m$. Since $i\in T$, the form $f_i$ is not identically zero. Hence its zero set has size at
most $q^{m-1}$: $|\{\mathbf{\alpha}\in \mathbb{F}_q^m: f_i(\mathbf{\alpha})=0\}|\le q^{m-1}$. Summing over all $\mathbf{\alpha}\in \F_q^m$, we obtain $\sum_{\mathbf{\alpha}\in \mathbb{F}_q^m} |\{i\in T: f_i(\mathbf{\alpha})=0\}|=\sum_{i\in T} |\{\mathbf{\alpha}\in \mathbb{F}_q^m: f_i(\mathbf{\alpha})=0\}|\le N\cdot q^{m-1}$.

For $\mathbf{\alpha}=(\alpha_1,\ldots,\alpha_m)\in \F_q^m$, if $\mathbf{u}_0+\sum_{j=1}^m \alpha_j \mathbf{u}_j\in U\cap B_E$, then
\[
N\cdot q^{m-1}
\ge \sum_{\alpha:\, \mathbf{u}_0+\sum_{j=1}^m \alpha_j \mathbf{u}_j\in B_E}
|\{i\in T:f_i(\alpha)=0\}|
\ge |U\cap B_E|\cdot (N-E).
\]
It follows that $|U\cap B_E| \le \frac{N}{N-E}\cdot q^{m-1}$. Since $N\ge E^+ +1$ and the function $x\mapsto x/(x-E)$ is decreasing for $x>E$, we obtain $|U\cap B_E| \le \frac{E^+ +1}{E^+ - E +1}\cdot q^{m-1}$. This completes the proof.
\end{proof}

We now generalize the rank-reduction argument from affine lines to affine syndrome spaces of dimension $m$.

\begin{defn}[Witness matrix]\label{def:correlated-witness}
Fix an affine syndrome space $S=\mathbf{s}_0+\mathrm{span}(\mathbf{s}_1,\dots,\mathbf{s}_m)\subseteq \mathbb{F}_q^r$ of dimension $m$, and let $\mathbf{\beta}_1,\ldots,\mathbf{\beta}_K\in \F_q^m$ be pairwise distinct elements, where $\mathbf{\beta}_j=((\mathbf{\beta}_j)_1,\dots,(\mathbf{\beta}_j)_m)$. A matrix $X=[\mathbf{x}_1|\cdots|\mathbf{x}_K]\in \mathbb{F}_q^{n\times K}$ is called an $E$-witness matrix for $(S;\mathbf{\beta}_1,\dots,\mathbf{\beta}_K)$ if $H\mathbf{x}_j=\mathbf{s}_0+\sum_{i=1}^m (\mathbf{\beta}_j)_i\cdot \mathbf{s}_i$ where $\wt(\mathbf{x}_j)\leq E$ for $j\in [K]$.

More generally, we say that $X=[\mathbf{x}_1|\ldots|\mathbf{x}_K]\in \F_q^{n\times K}$ is an $E$-witness matrix covering $K$ elements of $S$ if there exist pairwise distinct elements
$\mathbf{\beta}_1,\ldots,\mathbf{\beta}_K\in\F_q^m$ such that $H\mathbf{x}_j=\mathbf{s}_0+\sum_{i=1}^m(\mathbf{\beta}_j)_i\cdot \mathbf{s}_i$ where $\mathrm{wt}(\mathbf{x}_j)\le E$ for $j\in [K]$.
\end{defn}

For an affine syndrome space $S=\mathbf{s}_0+\mathrm{span}(\mathbf{s}_1,\ldots,\mathbf{s}_m),$ let $h_S:=\dim_{\mathbb F_q}\mathrm{span}\{\mathbf{s}_0,\mathbf{s}_1,\ldots,\mathbf{s}_m\}$. Since $S$ has dimension $m$, we have $h_S\in\{m,m+1\}$.

\begin{lemma}\label{lem:rank-reduction-m}
Let $S=\mathbf{s}_0+\mathrm{span}(\mathbf{s}_1,\dots,\mathbf{s}_m)\subseteq \mathbb{F}_q^r$ be an affine syndrome space of dimension $m$, and let $h:=\dim_{\mathbb F_q}\mathrm{span}\{\mathbf{s}_0,\mathbf{s}_1,\ldots,\mathbf{s}_m\}$. Let $\mathbf{\beta}_1,\dots,\mathbf{\beta}_K\in \mathbb{F}_q^m$ be pairwise distinct elements that are not contained in any affine hyperplane of $\mathbb{F}_q^m$.

Let $X=[\mathbf{x}_1|\cdots|\mathbf{x}_K]\in \mathbb{F}_q^{n\times K}$ be an $E$-witness matrix for $(S;\mathbf{\beta}_1,\dots,\mathbf{\beta}_K)$, and suppose $\mathrm{rank}(X)=t\ge h+1$. If $d(\cC)\ge d$ and $\gamma:=\frac{d-E}{d}$, then there exists a subset $J\subseteq[K]$ such that $|J|\ge K\gamma,$ the restricted matrix $X_J=[x_j]_{j\in J}$ is still an $E$-witness matrix for the same affine syndrome space $S$, and $\mathrm{rank}(X_J)\le t-1$.
\end{lemma}

\begin{proof}
Define $\mathbf{u}_0:=(1,\dots,1)\in \mathbb{F}_q^K,\, \mathbf{u}_i:=((\mathbf{\beta}_1)_i,\dots,(\mathbf{\beta}_K)_i)\in \mathbb{F}_q^K\, (1\le i\le m)$,
and let
\[
Y:=
\left[
\mathbf{s}_0+\sum_{i=1}^m (\mathbf{\beta}_1)_i\cdot \mathbf{s}_i
\ \Big|\ \cdots\ \Big|\
\mathbf{s}_0+\sum_{i=1}^m (\mathbf{\beta}_K)_i\cdot \mathbf{s}_i
\right]
\in \mathbb{F}_q^{r\times K}.
\]
Since $\mathbf{\beta}_1,\dots,\mathbf{\beta}_K\in \mathbb{F}_q^m$ are not contained in any affine hyperplane, the vectors $\mathbf{u}_0,\mathbf{u}_1,\dots,\mathbf{u}_m$ are linearly independent. Moreover, $Y=\mathbf{s}_0\cdot \mathbf{u}_0^\top+\sum_{i=1}^m \mathbf{s}_i\cdot \mathbf{u}_i^\top$, and hence  $\mathrm{Row}(Y)\subseteq \mathrm{span}\{\mathbf{u}_0,\dots,\mathbf{u}_m\}$ and $\mathrm{dim} \mathrm{Row}(Y)=h$. Choose a basis $\mathbf{v}_1,\ldots,\mathbf{v}_h$ of $\mathrm{Row}(Y)$ and extend it to a basis $\mathbf{v}_1,\dots,\mathbf{v}_h,\mathbf{w}_{h+1},\dots,\mathbf{w}_t$ of $\mathrm{Row}(X)$. Then there exist vectors $\mathbf{a}_1,\dots,\mathbf{a}_h,\mathbf{c}_{h+1},\dots,\mathbf{c}_t\in \mathbb{F}_q^n$ such that $X= \sum_{i=1}^h \mathbf{a}_i\cdot \mathbf{v}_i^\top +\sum_{\ell=h+1}^t \mathbf{c}_\ell\cdot \mathbf{w}_\ell^\top$. Since $HX=Y$ and $\mathbf{v}_1,\ldots,\mathbf{v}_h$ span $\mathrm{Row}(Y)$, it follows that $H\mathbf{c}_\ell=\mathbf{0}\, (h+1\le \ell\le t)$. Hence each $\mathbf{c}_\ell$ is a codeword.

If $\mathbf{c}_t=\mathbf{0}$, then already $\mathrm{rank}(X)\le t-1$, and there is nothing to prove. Thus assume $\mathbf{c}_t\neq 0$, and let $T:=\mathrm{supp}(\mathbf{c}_t)$. Since $\mathbf{c}_t\in \cC\setminus\{0\}$ and $d(\cC)\ge d$, we have $|T|=\mathrm{wt}(\mathbf{c}_t)\ge d$. Each column $\mathbf{x}_j$ has at least $|T|-E$ zero coordinates inside $T$. By the pigeonhole principle, there exists $i_0\in T$ such that $|\{j\in[K]: (\mathbf{x}_j)_{i_0}=0\}|\ge K\frac{|T|-E}{|T|} \ge K\frac{d-E}{d} = K\gamma$. Let $J:=\{j\in[K]: (\mathbf{x}_j)_{i_0}=0\}$. Then $|J|\ge K\gamma$.

For $j\in J$, the $j$-th column of $X$ is $\mathbf{x}_j=\sum_{i=1}^h (\mathbf{v}_i)_j \mathbf{a}_i+\sum_{\ell=h+1}^{t-1}(\mathbf{w}_\ell)_j \mathbf{c}_\ell+(\mathbf{w}_t)_j \mathbf{c}_t$. Since $(\mathbf{x}_j)_{i_0}=0$ and $(\mathbf{c}_t)_{i_0}\neq 0$, we may solve for $(\mathbf{w}_t)_j$ as in the proof of Theorem~\ref{thm:rank-reduction}. This yields a representation of $X_J$ using $(\mathbf{u}_0)_J,\ldots,(\mathbf{u}_m)_J,(\mathbf{w}_{m+1})_J,\ldots,(\mathbf{w}_{t-1})_J$. The restricted matrix $X_J$ therefore has rank at most $t-1$, and it remains an $E$-witness matrix for the same affine syndrome space $S$.
\end{proof}

\begin{cor}\label{cor:rank-reduction-m-iterated}
Under the assumptions of Lemma~\ref{lem:rank-reduction-m}, suppose that $\mathrm{rank}(X)=t\ge h$ and $K\gamma^{t-h}>q^{m-1}$. Then there exists a subset $J\subseteq[K]$ with $|J|\ge K\gamma^{\,t-h}$ and vectors $\mathbf{a}_0,\mathbf{a}_1,\dots,\mathbf{a}_m\in \mathbb{F}_q^n$ such that for every $j\in J$, $\mathbf{x}_j=\mathbf{a}_0+\sum_{i=1}^m (\mathbf{\beta}_j)_i\cdot \mathbf{a}_i$, $H\mathbf{a}_i=\mathbf{s}_i\ (0\le i\le m)$. In particular, the surviving witness vectors lie in an affine $m$-space $\widetilde U = \mathbf{a}_0 + \mathrm{span}(\mathbf{a}_1,\dots,\mathbf{a}_m)\subseteq \mathbb{F}_q^n$ satisfying $\widetilde U_H=S$.
\end{cor}

\begin{proof}
If $t=h$, we skip the iteration and apply the final-stage argument below directly. Otherwise, iterate Lemma~\ref{lem:rank-reduction-m} until the rank drops to $h$. Each step preserves at least a $\gamma$-fraction of the columns, so the final set $J$ satisfies $|J|\ge K\gamma^{t-h}>q^{m-1}$. Thus the surviving element set is never contained in an affine hyperplane, and the iteration is valid.

At the final stage, $\mathrm{rank}(X_J)=h$. Let $Y_J:=HX_J$. Since $|J|>q^{m-1}$, the vectors $(\mathbf{u}_0)_J,\ldots,(\mathbf{u}_m)_J$ are linearly independent. Hence $\mathrm{rank}(Y_J)
=\dim\mathrm{span}\{\mathbf{s}_0,\ldots,\mathbf{s}_m\}=h$. Therefore
\[
\mathrm{Row}(X_J)=\mathrm{Row}(Y_J)
\subseteq \mathrm{span}\{(\mathbf{u}_0)_J,\ldots,(\mathbf{u}_m)_J\}.
\]
It follows that there exist $\mathbf{a}_0,\ldots,\mathbf{a}_m\in\mathbb F_q^n$ such that $\mathbf{x}_j=\mathbf{a}_0+\sum_{i=1}^m(\mathbf{\beta}_j)_i \mathbf{a}_i \,\,(j\in J).$ Finally, applying $H$ and comparing coefficients gives $H\mathbf{a}_i=\mathbf{s}_i$ for $0\le i\le m$. This proves the claim.
\end{proof}

We now state the $m$-dimensional affine-space analogue of Theorem~\ref{thm:min-rank-threshold}.

\begin{thm}\label{thm:min-rank-bad-m-space}
Let $\cC$ be a linear code with parity-check matrix $H$ and minimum distance $d(C)\ge d$, and let $0<E\le E^+<d$. Let $\gamma:=\frac{d-E}{d},B_{E,E^+}:= \left\lfloor\frac{E^+ + 1}{E^+-E+1}\right\rfloor$. Let $S=\mathbf{s}_0+\mathrm{span}(\mathbf{s}_1,\ldots,\mathbf{s}_m)\subseteq\mathbb F_q^r$ be an affine syndrome space of dimension $m$ and let $h:=\mathrm{dim\, span}\{\mathbf{s}_0,\ldots,\mathbf{s}_m\}$. Suppose that $S$ has no $(H,m,E,E^+)$-correlated agreement and there exists an $E$-witness matrix $X$ of rank $t\ge h$ covering $K$ elements of $S$ as in Definition~\ref{def:correlated-witness}. Then $K\gamma^{t-h}\le B_{E,E^+}\,q^{m-1}$.
\end{thm}

\begin{proof}
Assume for contradiction that $K\gamma^{t-h}>B_{E,E^+}q^{m-1}$. Since $B_{E,E^+}\ge 1$, we also have $K\gamma^{t-h}>q^{m-1}$. Hence Corollary~\ref{cor:rank-reduction-m-iterated} applies. Therefore there exists a subset $J\subseteq[K]$ with $|J|\ge K\gamma^{t-h}$ and vectors $\mathbf{a}_0,\mathbf{a}_1,\ldots,\mathbf{a}_m\in\mathbb F_q^n$ such that, for every
$j\in J$, $\mathbf{x}_j=\mathbf{a}_0+\sum_{i=1}^m(\mathbf{\beta}_j)_i \mathbf{a}_i$ and $H\mathbf{a}_i=\mathbf{s}_i\quad(0\le i\le m)$.

Let $\widetilde U=\mathbf{a}_0+\mathrm{span}(\mathbf{a}_1,\ldots,\mathbf{a}_m)$. Then $\widetilde U_H=S$. Since $S$ has no $(H,m,E,E^+)$-correlated agreement, Lemma~\ref{lem:syndrome-CA-reformulation} implies that $|\mathrm{Supp}(\widetilde U)|>E^+$. Every surviving witness vector has weight at most $E$, so $|J|\le |\widetilde U\cap B_E|$. By Lemma~\ref{lem:affine-space-ball}, $|\widetilde U\cap B_E|\le B_{E,E^+}q^{m-1}.$ Therefore $K\gamma^{t-h}\le |J| \le B_{E,E^+}q^{m-1}$, contradicting the assumption. This proves the claim.
\end{proof}

\begin{comment}
\begin{proof}
By Corollary~\ref{cor:rank-reduction-m-iterated}, there exists a subset
$J\subseteq[K]$ with $|J|\ge K\gamma^{\,t-(m+1)}$ such that the corresponding witness vectors lie in an affine $m$-space $\widetilde U = \mathbf{a}_0 + \mathrm{span}(\mathbf{a}_1,\dots,\mathbf{a}_m)\subseteq \mathbb{F}_q^n$ with $\widetilde U_H=S.$ Because $S$ has no $(H,m,E,E^+)$-correlated agreement, Lemma~\ref{lem:syndrome-CA-reformulation}
implies that no such preimage space can satisfy $|\mathrm{Supp}(\widetilde U)|\le E^+$. Thus, $|\mathrm{Supp}(\widetilde U)|>E^+$.

Moreover, every surviving witness column belongs to $\widetilde U$ and has weight at most $E$. Hence, $|J| \le|\widetilde U\cap B_E|$. Applying Lemma~\ref{lem:affine-space-ball}, we obtain $|J| \le\frac{E^+ +1}{E^+-E+1}\,q^{m-1}$. Since $|J|\ge K\cdot \gamma^{t-(m+1)}$, this yields $K\gamma^{\,t-(m+1)}\le B_{E,E^+}\cdot q^{m-1}$ as claimed.
\end{proof}

\end{comment}

We next generalize the preceding union-bound argument, Theorem~\ref{thm:bad-line-moment}, from affine lines to affine syndrome spaces of dimension $m$.

\begin{thm}\label{thm:moment-bad-m-space}
Fix $m\ge 1$. Let $0<E\le E^+$, let $s\ge 0$ be an integer, and let $d>E$. Define $\gamma:=\frac{d-E}{d},\, B_{E,E^+}:=\left\lfloor \frac{E^+ + 1}{E^+ - E + 1}\right\rfloor$.
Assume that $K > B_{E,E^+}\,q^{m-1}\gamma^{-s}$. Let $\cC$ be a random linear code with parity-check matrix $H\in \F_q^{r\times n}$. Then the probability that there exists an affine syndrome space $S\subseteq \F_q^r$ of dimension $m$ such that $S$ has no $(H,m,E,E^+)$-correlated agreement, $|S\cap H_E|\geq K$ and $d(\cC)\geq d$ is at most $\binom{q^m}{K}q^{-r(s+1)}\left(\binom{K}{m+s+1}|B_E|^{m+s+1}+\binom{K}{m+s+2}|B_E|^{m+s+2}\right)$.
\end{thm}

\begin{proof}
Let $S:=\mathbf{s}_0+\mathrm{span}(\mathbf{s}_1,\dots,\mathbf{s}_m)\subseteq \F_q^r$ be an affine syndrome space and a $K$-element subset $\mathbf{\beta}_1,\dots,\mathbf{\beta}_K\in \F_q^m$ of pairwise distinct elements. Let $h:=\dim_{\mathbb F_q}\mathrm{span}\{\mathbf{s}_0,\ldots,\mathbf{s}_m\}$. Since $S$ has dimension $m$, we have $h\in\{m,m+1\}$. Consider the event
\[
\begin{aligned}
\mathcal{A}(S;\mathbf{\beta}_1,\dots,\mathbf{\beta}_K)
:=
&\Bigl[
\mathbf{s}_0+\sum_{i=1}^m (\mathbf{\beta}_j)_i\cdot \mathbf{s}_i \in H_E \text{ for all } j\in[K]
\Bigr] \\
&\wedge
\Bigl[
S \text{ has no } (H,m,E,E^+)\text{-correlated agreement}
\Bigr] \wedge
\bigl[d(\cC)\ge d\bigr].
\end{aligned}
\]
If this event occurs, then there exists an $E$-witness matrix $X=[\mathbf{x}_1|\cdots|\mathbf{x}_K]$ for $(S;\mathbf{\beta}_1,\dots,\mathbf{\beta}_K)$. Since $d(\cC)\ge d$ and $K > B_{E,E^+}\,q^{m-1}\gamma^{-s}\ge q^{m-1}$, the $K$ distinct elements cannot lie in any affine hyperplane of $\F_q^m$. Hence Theorem~\ref{thm:min-rank-bad-m-space} applies and shows that every such witness matrix must satisfy $\rank(X)\ge h+s+1$.

Therefore one may choose a subset of $h+s+1$ columns that is linearly independent. There are at most $\binom{K}{h+s+1}|B_E|^{h+s+1}$ ways to choose these columns. By Lemma~\ref{lem:random-image}, the probability that a fixed linearly independent $h+s+1$-tuple maps to the prescribed syndrome columns is exactly $q^{-r(h+s+1)}$. Hence
\[
\Pr_{\cC}\!\bigl[\mathcal{A}(S;\mathbf{\beta}_1,\dots,\mathbf{\beta}_K)\bigr]
\le
\binom{K}{h+s+1}|B_E|^{h+s+1}\cdot q^{-r(h+s+1)}.
\]
Now fix $S$. If $|S\cap H_E|\ge K,\, S \text{ has no }(H,m,E,E^+)\text{-correlated agreement},\,d(\cC)\ge d$, then we can choose a $K$-element subset $\{\mathbf{\beta}_1,\dots,\mathbf{\beta}_K\}\subseteq \mathbb{F}_q^m$ such that $\mathbf{s}_0+\sum_{i=1}^m (\mathbf{\beta}_j)_i\mathbf{s}_i\in H_E$ for $j\in [K]$.

%For each such choice, the event $\mathcal{A}(S;\mathbf{\beta}_1,\dots,\mathbf{\beta}_K)$ occurs. We next bound the probability that there exists affine syndrome space $S\subseteq \F_q^r$ such that
%\[
%\Bigl[|S\cap H_E|\geq K\Bigr] \wedge\Bigl[S \text{ has no } (H,m,E,E^+)\text{-correlated agreement}\Bigr] \wedge\bigl[d(\cC)\ge d\bigr].
%\]
Taking a union bound over all $\binom{q^m}{K}$ possible $K$-subsets $\{\beta_1,\ldots,\beta_K\}$ of $\mathbb{F}_q^m$ and all affine spaces $S$ of dimension $m$, the failure probability is at most 
\[
\sum_{S,(\beta_1,\ldots,\beta_K)}\Pr_{\cC}\!\bigl[\mathcal{A}(S;\mathbf{\beta}_1,\dots,\mathbf{\beta}_K)\bigr]\leq 
\binom{q^m}{K}q^{-r(s+1)}
\left(
\binom{K}{m+s+1}|B_E|^{m+s+1}
+
\binom{K}{m+s+2}|B_E|^{m+s+2}
\right).
\]
This completes the proof.
\end{proof}

We now derive an explicit random-coding corollary in the same spirit as Theorem~\ref{thm:explicit-two-radius}.

\begin{thm}\label{thm:explicit-direct-CA-m}
Fix an integer $m \ge 1$, a rate $R \in (0,1)$, and $0<\varepsilon < (1-R)/2$. Let $a_\varepsilon := \frac{\varepsilon}{\log_2(1/\varepsilon)}$, $0<\rho<1-R-\varepsilon-a_\varepsilon$, $\delta := 1-R-a_\varepsilon, \lambda := \left\lceil \frac{(m+1)(1-R)}{\varepsilon}\right\rceil - m - 2$, $\tau :=\left\lceil\left(1+\frac{\rho}{\varepsilon}\right)\left(\frac{\delta}{\delta-\rho}\right)^\lambda\right\rceil,\,E := \lfloor \rho n\rfloor,\, E^+ := E+\lceil \varepsilon n\rceil$. Let $q$ be a prime power such that $q\ge \max\left\{\left(\frac{2}{\varepsilon}\right)^{1/\varepsilon},\tau+1\right\}$. Then, with probability at least $1-q^{-\Omega(n)}$, a random linear code $\cC$ of rate $R$ and length $n$ has minimum distance at least $\lfloor\delta n\rfloor$ and satisfies the $(m,E,E^+,\tau/q)$-space correlated agreement property.
\end{thm}

\begin{proof}
Define $b_\varepsilon:=\frac{\varepsilon}{1+\log_2(1/\varepsilon)}$. Let $d:=\lfloor \delta n\rfloor,\, \gamma:=\frac{d-E}{d},\, B_n:=\left\lfloor \frac{E^+ +1}{E^+ - E +1}\right\rfloor$. Since $\delta>\rho$, we have $d>E$ for all sufficiently large $n$.

We first estimate the probability that $\cC$ has minimum distance less than $d$. By Lemma~\ref{lem:H_q-bound} and Lemma~\ref{lem:ball-volume-bound}, $|B_{d}| \le q^{(\delta+b_\varepsilon)n+o(n)}$. Therefore, $\Pr_{\cC}[d(\cC)\le d]\le |B_{d}|\,q^{-r} \le q^{(\delta+b_\varepsilon-(1-R))n+o(n)}$. As $\delta+b_\varepsilon-(1-R)=-(a_\varepsilon-b_\varepsilon)$, we obtain $\Pr_{\cC}[d(\cC)\le d] \le q^{-(a_\varepsilon-b_\varepsilon)n+o(n)}= q^{-\Omega(n)}$.

Next we estimate the bad-space probability on the event $d(\cC)\ge d$. Let $K:=\tau q^{m-1}+1$. Since $B_n \le 1+\frac{\rho}{\varepsilon},\,\gamma^{-1}=\frac{d}{d-E}\le \frac{\delta}{\delta-\rho}$, we have $B_n q^{m-1}\gamma^{-\lambda}\le q^{m-1} \left(1+\frac{\rho}{\varepsilon}\right)\left(\frac{\delta}{\delta-\rho}\right)^\lambda\le\tau q^{m-1}< K$. Also, since $q\ge \tau+1$, we have $\tau<q$, hence $K \le q^m$. If $\frac{|S\cap H_E|}{|S|}>\frac{\tau}{q}$, then $|S\cap H_E|>\tau q^{m-1}$, hence $|S\cap H_E|\ge K$. We bound
the probability that there exists affine syndrome space $S\subseteq \F_q^r$ such that
\[
\Bigl[|S\cap H_E|\ge K\Bigr] \wedge\Bigl[S \text{ has no } (H,m,E,E^+)\text{-correlated agreement}\Bigr] \wedge\bigl[d(\cC)\ge d\bigr].
\]
By Theorem~\ref{thm:moment-bad-m-space} with $d=d$ and $s=\lambda$, we obtain the probability is at most 
\[
\binom{q^m}{K}q^{-r(\lambda+1)}\left(\binom{K}{m+\lambda+1}|B_E|^{m+\lambda+1}+\binom{K}{m+\lambda+2}|B_E|^{m+\lambda+2}\right).
\]

Since $m$, $q$, and $K$ depend only on $(R,\rho,\varepsilon,m)$, the two binomial factors contribute only $q^{O(1)}$. On the other hand, Lemma~\ref{lem:ball-volume-bound} gives $|B_E|\le q^{(\rho+b_\varepsilon)n+o(n)}.$ Thus, the exponent in the preceding probability bound is at most $(m+1)(1-R)n + (m+\lambda+2)(\rho+b_\varepsilon-(1-R))n + o(n)$. By the assumption on $\rho$, we have $1-R-\rho-b_\varepsilon>\varepsilon$. Since $m+\lambda+2\ge \frac{(m+1)(1-R)}{\varepsilon}$, it follows that $(m+\lambda+2)(1-R-\rho-b_\varepsilon)>(m+1)(1-R)$. Hence, $(m+1)(1-R)+(m+\lambda+2)(\rho+b_\varepsilon-(1-R))<0$ and therefore the probability is at most $q^{-\Omega(n)}$.

Combining this with the minimum-distance bound, we conclude that with probability at least $1-q^{-\Omega(n)},$ every affine syndrome space $S$ of dimension $m$ with $\frac{|S\cap H_E|}{|S|}>\tau/q$ admits a matrix $X=[\mathbf{x}_0|\cdots|\mathbf{x}_m]$ with $H\mathbf{x}_i=\mathbf{s}_i\ (0\le i\le m)$ and $\mathrm{rowwt}(X)\le E^+$. The equivalent correlated-agreement formulation now follows from Lemma~\ref{lem:syndrome-CA-reformulation}.
\end{proof}

We next record the no-slack analogue. Since Theorem~\ref{thm:explicit-direct-CA-m} is stated in the constant-alphabet setting form $E^+=E+\lceil \varepsilon n\rceil$, the large-alphabet result is obtained instead by specializing Theorem~\ref{thm:moment-bad-m-space} to the case $m=1$ and $E^+=E$, and then lifting the line statement to affine spaces via Lemma~\ref{lem:correlated-line-to-space} together with the list-decoding bound in Lemma~\ref{lem:list-decoding}.

\begin{thm}\label{thm:space-one-radius}
Fix $R \in (0,1)$ and $0<\varepsilon < (1-R)/2$. Assume that $0<\rho<1-R-\varepsilon$. Let $\delta := 1-R-\varepsilon, \ell := \left\lceil \frac{2(1-R)}{\varepsilon}\right\rceil - 1, E := \lfloor \rho n\rfloor, d := \lfloor \delta n\rfloor$, and define $K := \left\lfloor(E+1)\left(\frac{d}{d-E}\right)^\ell \right\rfloor$. Let $q$ be a prime power such that $q=\Theta(n)$ and $q \ge K+2$.

Then, with probability at least $1-q^{-\Omega(n)}$, a random linear code $\cC$ of rate $R$, length $n$ has minimum distance at least $d$ and satisfies the $(1,E,E,K/q)$-correlated agreement property.
\end{thm}

\begin{proof}
For $m=1$, $E^+=E$, we have $B_{E,E}=E+1$. Since $K+1>(E+1)\left(\frac{d}{d-E}\right)^\ell$, Theorem~\ref{thm:moment-bad-m-space} gives the corresponding union bound. We first estimate the minimum-distance event. Since $q=\Theta(n)$ and $d=\Theta(n)$, Lemma~\ref{lem:ball-volume-bound} yields $|B_{d}|\le q^{d+o(n)}<q^{(1-R-\varepsilon)n+o(n)}$.
Therefore,
\[
\Pr_{\cC}[d(\cC)\le d]\le |B_{d}|\,q^{-r}
\le q^{-\varepsilon n+o(n)}.
\]
Next we estimate the bad-line probability on the event $d(\cC)\ge d$. We bound the probability that there exists affine syndrome space $L\subseteq \F_q^r$ such that
\[
\left[|L\cap H_E|\ge K+1\right]\wedge \left[L\text{ has no }(H,1,E,E)\text{-correlated agreement}\right]\wedge \left[d(\cC)\geq d\right].
\]
Applying Theorem~\ref{thm:moment-bad-m-space} with $m=1$, $E^+=E$, and $s=\ell$, we obtain the probability is at most 
$\binom{q}{K+1}q^{-r(\ell+1)}\left(\binom{K+1}{\ell+2}|B_E|^{\ell+2}+\binom{K+1}{\ell+3}|B_E|^{\ell+3}\right)$

Because $q=\Theta(n)$ and $K=\Theta(n)$, we have $\log_q\binom{q}{K+1}=o(n),\, \log_q\binom{K+1}{\ell+3}=o(n)$. Also, by Lemma~\ref{lem:ball-volume-bound}, $|B_E|\le q^{E+o(n)}\le q^{\rho n+o(n)}$. Hence the exponent in the preceding probability bound is at most $2(1-R)n+(\ell+3)(\rho-(1-R))n+o(n)$. Since $\rho<1-R-\varepsilon$, this is bounded above by $\bigl(2(1-R)-(\ell+3)\varepsilon\bigr)n+o(n)$. By the definition of $\ell$, $\ell+3=\left\lceil \frac{2(1-R)}{\varepsilon}\right\rceil+2 > \frac{2(1-R)}{\varepsilon}$ and therefore $2(1-R)-(\ell+3)\varepsilon<0$. Thus, the probability is at most $q^{-\Omega(n)}$.

Combining this with the minimum-distance bound, we conclude that with probability at least $1-q^{-\Omega(n)}$, every affine syndrome line $L$ with $\frac{|L\cap H_E|}{|L|}>\frac{K}{q}$ admits a matrix $X=[\mathbf{x}_0|\mathbf{x}_1]$ such that $H\mathbf{x}_0=\mathbf{s}_0,\,H\mathbf{x}_1=\mathbf{s}_1$ and $\mathrm{rowwt}(X)\le E$. This is exactly the $(1,E,E,K/q)$-correlated agreement property.
\end{proof}

\begin{cor}
Under the assumptions of Theorem~\ref{thm:space-one-radius}, with probability at least $1-q^{-\Omega(n)}$, the same random linear code satisfies the $(m,E,E,K/(q-1))$-space correlated agreement property simultaneously for every integer $m\ge 1$.
\end{cor}

\begin{proof}
By the proof of Theorem~\ref{thm:space-one-radius}, the line correlated-agreement statement holds simultaneously for all radii $0\le E'\le E$. Indeed, for such $E'$, the corresponding threshold $(E'+1)\left(\frac{d}{d-E'}\right)^\ell$ is at most $(E+1)\left(\frac{d}{d-E}\right)^\ell < K+1$, since the function $x\mapsto (x+1)\left(\frac{d}{d-x}\right)^\ell$ is increasing on $0\le x<d$. Moreover, $|B_{E'}|\le |B_E|$, so the same union bound as in Theorem~\ref{thm:space-one-radius} gives failure probability $q^{-\Omega(n)}$, uniformly over $E'\le E$. Taking a union bound over the $E+1=O(n)$ choices of $E'$ still gives failure probability $q^{-\Omega(n)}$. Hence, with high probability, $\cC$ satisfies the $(1,E',E',K/q)$-line correlated-agreement property for every $0\le E'\le E$.

It remains to verify the hypothesis of Lemma~\ref{lem:correlated-line-to-space}. Since $\rho<1-R-\varepsilon$, for all sufficiently large $n$ we have $E=\lfloor \rho n\rfloor < \left(1-R-\varepsilon\right)n$. Hence, Lemma~\ref{lem:list-decoding} implies that, with probability $1-q^{-\Omega(n)}$, for every
$\mathbf{y}\in\F_q^n$, $|\{\mathbf{c}\in \cC : d(\mathbf{y},\mathbf{c})\le E\}|=O(\frac{1}{\varepsilon})<q$.

Applying Lemma~\ref{lem:correlated-line-to-space} with $\tau=\frac{K}{q}$, we conclude that $\cC$ has the
$(m,E,E,K/(q-1))$
-space correlated agreement property for every integer $m\ge 1$.
\end{proof}

\section{Curve Based Correlated Agreement of Random Linear Codes}
\label{section:curve-correlated-agreement}
In this section we adapt the argument of Section~\ref{section:correlated-agreement} from affine $m$-spaces to polynomial curves of degree at most $\ell$. The only substantive new element is that the element space is now one-dimensional. Accordingly, the analogue of Lemma~\ref{lem:affine-space-ball} is proved by counting roots of univariate polynomials, and the factor $q^{m-1}$ from Section~\ref{section:correlated-agreement} is replaced by $\ell$.

For vectors $\mathbf{u}_0,\dots,\mathbf{u}_\ell \in \F_q^n$, define the degree-$\ell$ polynomial curve $\Gamma(\alpha):=\mathbf{u}_0+\sum_{j=1}^{\ell}\alpha^j \mathbf{u}_j,\,\alpha\in\mathbb F_q$. Likewise, for $\mathbf{s}_0,\dots,\mathbf{s}_\ell \in \F_q^r$, define the degree-$\ell$ syndrome curve $S(\alpha):=\mathbf{s}_0+\sum_{j=1}^{\ell}\alpha^j \mathbf{s}_j,\, \alpha\in\mathbb F_q$. If $H\mathbf{u}_j = \mathbf{s}_j$ for every $0 \le j \le \ell$, then $\Gamma_H :=\{H\mathbf{u}_{0}+\sum_{j=1}^{\ell}\alpha^j  H\mathbf{u}_j:\alpha\in \F_q\}=S$. By Lemma~\ref{lem:syndrome-characterization}, for every $\alpha \in \F_q$,
\[
\mathbf{s}_0 + \sum_{j=1}^{\ell} \alpha^j \mathbf{s}_j \in H_E
\iff
d\!\left(\mathbf{u}_0 + \sum_{j=1}^{\ell} \alpha^j \mathbf{u}_j,\ \cC\right) \le E .
\]

\begin{defn}[$(\ell,E,E^+,\tau)$-correlated agreement for degree-$\ell$ curves]
\label{def:curve-direct-ca}
Let $\ell\ge1$, $0\le E\le E^+\le n$, and $0\le\tau\le1$. We say that a linear code $\cC$ with parity-check matrix $H$ has the $(\ell,E,E^+,\tau)$-correlated agreement property for degree-$\ell$ curves if every degree-$\ell$ syndrome curve $S(\alpha):=\mathbf{s}_0+\sum_{j=1}^{\ell}\alpha^j \mathbf{s}_j,\, \alpha\in \F_q$ satisfies the following: if $\frac{1}{q}\left|\left\{\alpha\in\mathbb F_q:S(\alpha) \in H_E\right\}\right|>\tau$, then there exists a matrix $X=[\mathbf{x}_0|\cdots|\mathbf{x}_\ell]\in \F_q^{n\times(\ell+1)}$ such that $H\mathbf{x}_j = \mathbf{s}_j\, (0\le j\le \ell)$ and $\mathrm{\mathrm{rowwt}}(X)\le E^+$.

For a fixed degree-$\ell$ syndrome curve $S(\alpha):=\mathbf{s}_0+\sum_{j=1}^{\ell}\alpha^j \mathbf{s}_j,\, \alpha\in \F_q$, if no such matrix $X$ exists, we say that $S$ has no $(H,\ell,E,E^+)$-correlated agreement.
\end{defn}

The next lemma shows that this is exactly the syndrome-space reformulation of the coefficient-wise correlated-agreement notion for degree-$\ell$ curves.

\begin{lemma}\label{lem:curve-ca-equivalence}
Let $\Gamma(\alpha):=\mathbf{u}_0+\sum_{j=1}^{\ell}\alpha^j \mathbf{u}_j,\,\alpha\in\mathbb F_q$ and $S=\Gamma_H=\mathbf{s}_0+\sum_{j=1}^{\ell}\alpha^j \mathbf{s}_j,\,\alpha\in \F_q$ where $\mathbf{s}_j = H\mathbf{u}_j$ for $0\le j\le \ell$.
Then the following are equivalent:
\begin{enumerate}
\item there exist codewords $\mathbf{c}_0,\dots,\mathbf{c}_\ell \in \cC$ such that
\[
\left|\bigcup_{j=0}^{\ell} \mathrm{supp}(\mathbf{c}_j-\mathbf{u}_j) \right|\leq E^+
\]
\item there exists a matrix $X=[\mathbf{x}_0|\cdots|\mathbf{x}_\ell]\in \F_q^{n\times(\ell+1)}$ such that $H\mathbf{x}_j=\mathbf{s}_j \, (0\le j\le \ell),\, \mathrm{\mathrm{rowwt}}(X)\le E^+.$
\end{enumerate}
\end{lemma}

\begin{proof}[Proof Sketch.]
The proof is similar to Lemma~\ref{lem:syndrome-CA-reformulation}. Therefore, we defer it to the Appendix~\ref{appendix:C}.
\end{proof}

We next prove the degree-$\ell$ analogue of Lemma~\ref{lem:affine-space-ball}.

\begin{lemma}\label{lem:curve-ball-intersection}
Let $\Gamma(\alpha)=\mathbf{a}_0+\sum_{j=1}^{\ell}\alpha^j \mathbf{a}_j,\, \alpha\in\mathbb F_q$. If the coefficient matrix $A=[\mathbf{a}_0|\cdots|\mathbf{a}_\ell]\in \F_q^{n\times(\ell+1)}$ satisfies $\mathrm{\mathrm{rowwt}}(A) > E^+$, then $\left|\{\alpha\in\mathbb F_q:\Gamma(\alpha)\in B_E\}\right|\le \ell\cdot \frac{E^++1}{E^+-E+1}$.
\end{lemma}

\begin{proof}
Let $T := \mathrm{rowsupp}(A), \,N:=|T|=\mathrm{\mathrm{rowwt}}(A)$. For each $i\in T$, define the polynomial $f_i(z) := (\mathbf{a}_0)_i + \sum_{j=1}^{\ell} z^j (\mathbf{a}_j)_i \in \F_q[z]$. Since $i\in T$, the coefficient tuple $\bigl( (\mathbf{a}_0)_i,\dots,(\mathbf{a}_\ell)_i \bigr)$ is not identically zero, so $f_i$ is a nonzero polynomial of degree at most $\ell$. Hence $\bigl|\{\alpha\in \F_q : f_i(\alpha)=0\}\bigr| \le \ell$.

Summing over $\alpha\in \F_q$, we obtain
\[
\sum_{\alpha\in \F_q} |{i\in T: f_i(\alpha)=0}|= \sum_{i\in T} \bigl|\{\alpha\in \F_q : f_i(\alpha)=0\}\bigr| \le N\ell.
\]
For each $\alpha \in \F_q$, if $\Gamma(\alpha)\in B_E$, then
\[
N\ell
\ge
\sum_{\alpha:\Gamma(\alpha)\in B_E} |\{i\in T:f_i(\alpha)=0\}|
\ge
\left|\{\alpha\in\mathbb F_q:\Gamma(\alpha)\in B_E\}\right|(N-E).
\]
It follows that $\left|\{\alpha\in\mathbb F_q:\Gamma(\alpha)\in B_E\}\right|\le\frac{N\ell}{N-E}$. Since $N\ge E^+ + 1$ and the function $x\mapsto x/(x-E)$ is decreasing for $x>E$, we obtain $\left|\{\alpha\in\mathbb F_q:\Gamma(\alpha)\in B_E\}\right|\le\ell \cdot \frac{E^+ + 1}{E^+ - E + 1}$. This completes the proof.
\end{proof}

Since the proofs for degree-$\ell$ curves are highly similar to those for $m$-dimensional affine spaces, we defer the proofs of many of the following results to the Appendix~\ref{appendix:C}. We now adapt the rank-reduction argument from affine $m$-spaces to degree-$\ell$ curves.

\begin{defn}[Witness matrix for degree-$\ell$ curves]\label{def:curve-witness}
\label{def:curve-witness}
Fix a degree-$\ell$ syndrome curve $S= \mathbf{s}_0+\sum_{i=1}^{\ell}\alpha^i \mathbf{s}_i : \alpha\in \F_q$, and pairwise distinct elements $\alpha_1,\dots,\alpha_K\in \F_q$. A matrix $X=[\mathbf{x}_1|\cdots|\mathbf{x}_K]\in \F_q^{n\times K}$ is called an $E$-witness matrix for $(S;\alpha_1,\dots,\alpha_K)$ if $H\mathbf{x}_j = \mathbf{s}_0+\sum_{i=1}^{\ell}\alpha_j^i \mathbf{s}_i$ with $\wt(x_j)\le E$ for $j\in[K]$.

More generally, we say that $X=[\mathbf{x}_1|\cdots|\mathbf{x}_K]\in\mathbb F_q^{n\times K}$ is an $E$-witness matrix covering $K$ elements of $S$ if there exist pairwise distinct elements $\alpha_1,\ldots,\alpha_K\in\mathbb F_q$ such that $H\mathbf{x}_j=\mathbf{s}_0+\sum_{i=1}^{\ell}\alpha_j^i \mathbf{s}_i$ where $\mathrm{wt}(x_j)\le E$ for $j\in[K]$.
\end{defn}

For a degree-$\ell$ curve $S= \mathbf{s}_0+\sum_{i=1}^{\ell}\alpha^i \mathbf{s}_i : \alpha\in \F_q $, let $h_S:=\dim_{\mathbb F_q}\mathrm{span}\{\mathbf{s}_0,\mathbf{s}_1,\ldots,\mathbf{s}_\ell\}$. Since $S$ has degree $\ell$, we have $1\leq h_S\leq \ell+1$.

\begin{lemma}\label{lem:curve-rank-reduction}
Let $S= \mathbf{s}_0+\sum_{i=1}^{\ell}\alpha^i \mathbf{s}_i : \alpha\in \F_q$ be a degree-$\ell$ syndrome curve, and let $\alpha_1,\dots,\alpha_K\in \F_q$ be pairwise distinct with $K\ge \ell+1$. Let $X=[\mathbf{x}_1|\cdots|\mathbf{x}_K]\in \F_q^{n\times K}$ be an $E$-witness matrix for $(S;\alpha_1,\dots,\alpha_K)$ and $h:=\dim_{\mathbb F_q}\mathrm{span}\{\mathbf{s}_0,\ldots,\mathbf{s}_\ell\}$, and suppose $\rank(X)=t\ge h+1.$ 

If $d(\cC)\ge d>E$ and $\gamma:=\frac{d-E}{d}$, then there exists a subset $J\subseteq [K]$ such that $|J|\ge K\gamma$, the restricted matrix $X_J=[\mathbf{x}_j]_{j\in J}$ is still an $E$-witness matrix for the same syndrome curve $S$, with $\rank(X_J)\le t-1$.
\end{lemma}

\begin{proof}[Proof Sketch.]
    The proof is similar to Lemma~\ref{lem:rank-reduction-m}. Therefore, we defer it to the Appendix~\ref{appendix:C}.
\end{proof}

\begin{cor}
\label{cor:curve-rank-reduction}
Under the assumptions of Lemma~\ref{lem:curve-rank-reduction}, suppose that $\rank(X)=t\ge h$ and $K\gamma^{t-h}>\ell$. Then there exists a subset $J\subseteq[K]$ with $|J|\ge K\gamma^{\,t-h}$ and vectors $\mathbf{a}_0,\dots,\mathbf{a}_\ell\in \F_q^n$ such that for every $j\in J$, $\mathbf{x}_j = \mathbf{a}_0+\sum_{i=1}^{\ell}\alpha_j^i\mathbf{a}_i$ with $H\mathbf{a}_i=\mathbf{s}_i \, (0\le i\le \ell)$.
\end{cor}

\begin{proof}
If $t=h$, there is nothing to prove. Otherwise, apply Lemma~\ref{lem:curve-rank-reduction} until the rank drops to $h$. Each step keeps at least a $\gamma$-fraction of the columns, so the final set $J$ satisfies $|J|\ge K\gamma^{t-h}>\ell$. 

At the last step, $\mathrm{rank}(X_J)=h$. Let $Y_J:=HX_J$. Since $|J|>\ell$ and the elements $\alpha_j$ are pairwise distinct, the restricted vectors $(\mathbf u_0)_J,\ldots,(\mathbf u_\ell)_J$ are linearly independent by the Vandermonde determinant. Hence $\mathrm{rank}(Y_J)=\dim_{\mathbb F_q}\mathrm{span}\{\mathbf s_0,\ldots,\mathbf s_\ell\}=h$. Therefore $\mathrm{Row}(X_J)=\mathrm{Row}(Y_J)\subseteq \mathrm{span}\{(\mathbf u_0)_J,\ldots,(\mathbf u_\ell)_J\}$. It follows that there exist $\mathbf a_0,\ldots,\mathbf a_\ell\in\mathbb F_q^n$
such that $\mathbf x_j=\mathbf a_0+\sum_{i=1}^{\ell}\alpha_j^i\mathbf a_i$ for $j\in J$. This implies $H\mathbf a_i=\mathbf s_i\, (0\le i\le \ell)$ and the proof is completed.
\end{proof}

We can now prove the degree-$\ell$ analogue of Theorem~\ref{thm:min-rank-bad-m-space}.

\begin{thm}
\label{thm:curve-witness-rank}
Let $S= \mathbf{s}_0+\sum_{i=1}^{\ell}\alpha^i \mathbf{s}_i : \alpha\in \F_q$ be a degree-$\ell$ syndrome curve and let $h:=\dim_{\mathbb F_q}\mathrm{span}\{\mathbf s_0,\ldots,\mathbf s_\ell\}$. Assume further that $d(\cC)\ge d>E$. Let $\gamma:=\frac{d-E}{d},\,B_{E,E^+}:= \left\lfloor\frac{E^+ +1}{E^+ - E +1}\right\rfloor$. Assume that $S$ has no $(H,\ell,E,E^+)$-correlated agreement and there exists an $E$-witness matrix $X$ of rank $t\ge h$ covering $K$ elements of $S$ as in Definition~\ref{def:curve-witness}. Then $K\gamma^{\,t-h} \le \ell\cdot B_{E,E^+}$.
\end{thm}

\begin{proof}[Proof Sketch.]
    The proof is similar to~\cref{thm:min-rank-bad-m-space}. Therefore, we defer it to the Appendix~\ref{appendix:C}.
\end{proof}

We next derive a quantitative union-bound estimate for bad syndrome curves.

\begin{thm}
\label{thm:curve-moment}
Fix $\ell\ge 1$. Let $0<E\le E^+$, let $s\ge 0$ be an integer, and let $d>E$.
Define $\gamma:=\frac{d-E}{d},\,B_{E,E^+}:=
\left\lfloor
\frac{E^+ +1}{E^+ - E +1}
\right\rfloor$.
Assume that $K > \ell\, B_{E,E^+}\,\gamma^{-s}$. Let $\cC$ be a random linear code with parity check matrix $H$. Then the probability of there exists a degree-$\ell$ syndrome curve $S=\mathbf{s}_0+\sum_{i=1}^{\ell}\alpha^i \mathbf{s}_i,\,\alpha\in\mathbb F_q,$ such that $\left|\{\alpha\in\F_q:\mathbf{s}_0+\sum_{i=1}^{\ell}\alpha^i \mathbf{s}_i\in H_E\}\right|\ge K$, $S\text{ has no }(H,\ell,E,E^+)\text{-correlated agreement}$ and $d(\cC)\ge d$ is at most 
\[
\binom{q}{K}q^{-r(s+1)}\sum_{h=1}^{\ell+1}q^{h(\ell+1)}\binom{K}{h+s+1}|B_E|^{h+s+1}.
\]
\end{thm}

\begin{proof}[Proof Sketch.]
    The proof is similar to~\cref{thm:moment-bad-m-space}. Therefore, we defer it to the Appendix~\ref{appendix:C}.
\end{proof}

We now derive an explicit random-coding corollary in the same spirit as Theorem~\ref{thm:explicit-direct-CA-m}.

\begin{thm}\label{thm:curve-explicit}
Fix an integer $\ell \ge 1$, a rate $R \in (0,1)$, and $0<\varepsilon < (1-R)/2$.
Define $a_\varepsilon := \frac{\varepsilon}{\log_2(1/\varepsilon)}$. Assume that $0<\rho<1-R-\varepsilon-a_\varepsilon$. Let $\delta := 1-R-a_\varepsilon,\, \lambda := \left\lceil \frac{(\ell+1)(1-R)}{\varepsilon}\right\rceil - \ell - 2,\,\tau :=\left\lceil\ell\left(1+\frac{\rho}{\varepsilon}\right)\left(\frac{\delta}{\delta-\rho}\right)^\lambda\right\rceil,\,E := \lfloor \rho n\rfloor,\,E^+ := E+\lceil \varepsilon n\rceil$. Let $q$ be a prime power such that $q \ge \max\!\left\{\left(\frac{2}{\varepsilon}\right)^{1/\varepsilon},\,\tau+1 \right\}$. Then, with probability at least $1-q^{-\Omega(n)}$, a random linear code $\cC$ of rate $R$, length $n$ and minimum distance at least $\lfloor \delta n\rfloor$ satisfies the $(\ell,E,E^+,\tau/q)$-correlated agreement property for degree-$\ell$ curves.
\end{thm}

\begin{proof}[Proof Sketch.]
    The proof is similar to~\cref{thm:explicit-direct-CA-m}. Therefore, we defer it to Appendix~\ref{appendix:C}.
\end{proof}

\begin{thm}\label{thm:curve-one-radius}
Fix an integer $\ell \ge 1$, a rate $R \in (0,1)$, and $0<\varepsilon < (1-R)/2$. Assume that $0<\rho<1-R-\varepsilon,\,\delta := 1-R-\varepsilon,\,\lambda := \left\lceil \frac{(\ell+1)(1-R)}{\varepsilon}\right\rceil - \ell - 1,\,E := \lfloor \rho n\rfloor,\, d := \lfloor \delta n\rfloor,\, K :=\left\lfloor\ell(E+1)\left(\frac{d}{d-E}\right)^\lambda\right\rfloor$ and let $q$ be a prime power such that $q=\Theta(n)$ and $q \ge K+2$.

Then, with probability at least $1-q^{-\Omega(n)}$, a random linear code $\cC$ of rate $R$ and code length $n$ has minimum distance at least $d$ and satisfies the $(\ell,E,E,K/q)$-correlated agreement property for degree-$\ell$ curves.
\end{thm}

\begin{proof}[Proof Sketch.]
    The proof is similar to Theorem~\ref{thm:space-one-radius}. Therefore, we defer it to the Appendix~\ref{appendix:C}.
\end{proof}

\section{Correlated Agreement for Random Reed-Solomon Codes}
In this section we show that the syndrome-space reductions proved above also apply to random Reed--Solomon codes. For random linear codes, this final step was a direct union bound in the random parity-check-matrix model. For random Reed--Solomon codes, a parity-check matrix is no longer random in this sense. We therefore replace the union bound by the local-profile theorem of~\cite{levi2025random}. The bridge is simple: a high-rank low-weight syndrome witness forces many linear combinations of its columns to be codewords, and hence gives a local profile. Random RS codes avoid such profiles with high probability.

This yields correlated-agreement statements for random RS codes up to radius $ \rho \le 1-R-\varepsilon $ over fields of size $q \ge n \cdot 2^{O(\varepsilon^{-3})},$ for affine $m$-spaces and $q\geq n \cdot 2^{O_\ell(\varepsilon^{-3})}$ for degree-$\ell$ curves. This improves the corresponding random-RS consequences obtained from the framework of~\cite{goyal2025optimal} which shows that the correlated-agreement for random RS codes up to radius $\rho\leq 1-R-2\varepsilon$ over fields of size $q\geq n\cdot 2^{O(\varepsilon^{-7})}$ for affine $m$-spaces and $q\geq n\cdot 2^{O_{\ell}(\varepsilon^{-(2\ell+5)})}$ for degree-$\ell$ curves.

A random Reed--Solomon code in this section means the iid evaluation model
\[
\cC=\mathrm{RS}_{\F_q}((\alpha_1,\ldots,\alpha_n);k)
:=\{(f(\alpha_1),\ldots,f(\alpha_n)): f\in \F_q[X],\deg f<k\}\subseteq \F_q^n,
\]
where $\alpha_1,\ldots,\alpha_n$ are sampled independently and uniformly from $\F_q$, and $k=Rn$. Throughout this section, $H$ denotes an arbitrary parity-check matrix for $\cC$.

We use the following standard minimum-distance estimate for the iid model. It is only needed to ensure that the deterministic rank-reduction theorems from Sections~4 and~5 can be applied directly, without passing to the without-repetitions model.

\begin{lemma}\label{lem:iid-rs-distance}
Fix $R\in(0,1)$ and $0<\varepsilon<(1-R)/2$.  If $q\ge (4/\varepsilon)n$, then, with probability at least $1-\exp(-\Omega_\varepsilon(n))$, the iid random RS code above has dimension $k$ and minimum distance at least $d:=\big\lfloor (1-R-\varepsilon/4)n\big\rfloor.$
\end{lemma}

This follows by deleting all repeated evaluation coordinates except their first occurrence. Let $D:=n-\big|\{\alpha_1,\ldots,\alpha_n\}\big|$ be the number of deleted coordinates. The remaining code is an RS code on distinct evaluation points, hence is MDS~\cite{RS60}. Moreover, $\mathbb E[D]\le \binom n2\frac1q\le \frac{\varepsilon n}{8}$ when $q\ge (4/\varepsilon)n$. Since changing one evaluation point changes $D$ by at most one, McDiarmid's inequality~\cite{McD89} gives $\Pr\left[D>\frac{\varepsilon n}{4}\right]\le \exp(-\Omega_\varepsilon(n)).$ Thus, with probability $1-\exp(-\Omega_\varepsilon(n))$, after deleting at most $\varepsilon n/4$ repeated coordinates, the remaining distinct-evaluation RS code has minimum distance at least $n-D-k+1\ge (1-R-\varepsilon/4)n$. Now we record the local-profile notation in~\cite{levi2025random}.

\begin{defn}
Let $b\ge 1$ be an integer. A $b$-local profile is a tuple $\mathcal V=(V_1,\ldots,V_n)$, where each $V_i\subseteq \F_q^b$ is a linear subspace.  A matrix $A\in\F_q^{n\times b}$ satisfies $\mathcal{V}$ if $A_{i,*}\in V_i$ for every $i\in[n]$.

Let $\cC\subseteq\F_q^n$ be a linear code and let $U\subseteq\mathbb F_q^b$.  We say that $\cC$ contains $(\mathcal V,U)$ if there is a matrix $A\in\mathbb F_q^{n\times b}$ such that
\[
  A_{*,j}\in \cC \quad (j=1,\ldots,b),
  \qquad
  A_{i,*}\in V_i \quad (i\in[n]),
  \qquad
  \mathrm{Row}(A)=U .
\]
For $U\subseteq \mathbb F_q^b$, define
\begin{equation}\label{eq:threshold-rate}
R_{\mathcal V,U}:=\max_{W\subsetneq U}\left( 1- \frac{\sum_{i=1}^n\bigl(\dim(V_i\cap U)-\dim(V_i\cap W)\bigr)}   {n(\dim U-\dim W)}\right).
\end{equation}
We call $R_{\mathcal{V},U}$ the threshold rate of $(\mathcal{V},U)$.
\end{defn}

\begin{lemma}[Theorem~3.10~\cite{levi2025random}]
\label{lem:rs-fixed-rowspan-lms}
Let $\cC=\mathrm{RS}_q((\alpha_1,\ldots,\alpha_n);k)$, where $\alpha_1,\ldots,\alpha_n$ are sampled independently and uniformly from $\mathbb F_q$, and let $k=Rn$.  Let $\mathcal V=(V_1,\ldots,V_n)$ be a $b$-local profile and let $0\ne U\le\mathbb F_q^b$.  Assume that $q>kb$ and fix $\eta\cdot n\ge 2b(b+1)$. If $R\leq R_{\mathcal{V},U}-\eta$, then $\Pr\bigl[\cC\text{ contains }(\mathcal V,U)\bigr]\le(2b-1)\left( \frac{(4b)^{4b}k}{\eta\cdot q} \right)^{\eta n/(2b)} .$

\end{lemma}

Lemma~\ref{lem:rs-fixed-rowspan-lms} is a fixed-profile corollary of Theorem 3.10 of~\cite{levi2025random}. We next convert high-rank syndrome witnesses into local profiles. We state the lemma in a form that applies uniformly to affine spaces and degree-$\ell$ curves, because the same statement will be used for affine spaces and polynomial curves.

The next lemma is the bridge between syndrome witnesses and the local-profile framework of~\cite{levi2025random}. Suppose a high-rank witness matrix $X$ has many linear combinations that are actual codewords, namely $X\cdot U\subseteq \cC$ for a subspace $U\subseteq \mathbb F_q^t$. Here $XU\subseteq \cC$ means that for every $\mathbf{u}=(\mathbf{u}_1,\ldots,\mathbf{u}_t)\in U$, $X\mathbf{u}=\sum_{j=1}^t \mathbf{u}_j\mathbf{x}_j\in \cC.$ Since each original witness column $\mathbf{x}_j$ has weight at most $E$, the matrix $X$ has at most $tE$ nonzero entries in total. Therefore, the rows of $X$ gives a local profile $\mathcal{V}=(V_1,\ldots,V_n)$ contained by $\cC$, and this profile satisfies $\sum_{i=1}^n \dim V_i \le tE.$

Let $X\in \mathbb F_q^{n\times t}$ and write $J_i(X):=\{j\in[t]:X_{i,j}\ne 0\}$ for $i\in [n]$. Thus, $J_i(X)$ records the columns of $X$ whose $i$-th coordinate is nonzero.

\begin{lemma}\label{lem:rs-structured-witness-profile}
Let $X=[\mathbf{x}_1|\cdots|\mathbf{x}_t]\in \mathbb F_q^{n\times t}$ have rank $t$, and assume $\mathrm{wt}(\mathbf{x}_j)\le E$ for all $j\in [t]$. Let $U\subseteq \mathbb F_q^t$ be a $b$-dimensional subspace with $1\le b<t$, and assume $X\cdot U\subseteq \cC.$

Choose a basis $\mathbf{u}_1,\ldots,\mathbf{u}_b$ of $U$. For each $j\in[t]$, define $\mathbf{m}_j:=((\mathbf{u}_1)_j,\ldots,(\mathbf{u}_b)_j)\in \mathbb F_q^b,$ and for each coordinate $i\in[n]$, define $V_i:=\mathrm{span}\{\mathbf{m}_j:j\in J_i(X)\}\subseteq \mathbb F_q^b.$ Let $\mathcal{V}=(V_1,\ldots,V_n)$. Then $\cC$ contains $(\mathcal{V},\F_q^b)$. Moreover, $\sum_{i=1}^n \dim V_i\le tE.$
\end{lemma}

\begin{proof}
Define $A:=[X\mathbf{u}_1|\cdots|X\mathbf{u}_b]\in \mathbb F_q^{n\times b}.$ Since $\mathbf{u}_1,\ldots,\mathbf{u}_b\in U$ and $X\cdot U\subseteq \cC$, every column of $A$ lies in $\cC$. Since $\mathbf{u}_1,\ldots,\mathbf{u}_b$ are linearly independent, the columns $X\mathbf{u}_1,\ldots,X\mathbf{u}_b$ are also linearly independent. Hence $\mathrm{rank}(A)=b$ and therefore $\mathrm{Row}(A)=\mathbb F_q^b$.

For each coordinate $i\in[n]$, the $i$-th row of $A$ is $A_{i,*}=\sum_{j\in J_i(X)} X_{i,j}\cdot \mathbf{m}_j,$ and hence $A_{i,*}\in V_i$. Thus $\cC$ contains $(\mathcal{V},\mathbb F_q^b)$.

Finally, $\sum_{i=1}^n \dim V_i\le \sum_{i=1}^n |J_i(X)|= \sum_{j=1}^t \mathrm{wt}(\mathbf{x}_j)\le tE.$ This proves the lemma.
\end{proof}

\begin{prop}
\label{prop:rs-profile-exclusion-uniform}
Fix $R\in(0,1)$, integers $1\le b<T$, and real numbers $\rho,\theta>0$ such that $T\rho<\theta<b(1-R).$ Let $E=\lfloor \rho n\rfloor$. Then there exists a constant $C_{R,b,T,\theta}$ such that, if $q\ge C_{R,b,T,\theta}\cdot n$, then with probability at least $1-\exp(-\Omega(n))$, the iid random RS code $\cC$ has the following property:

There do not exist an integer $t$ with $b+1\le t\le T$, a $b$-dimensional subspace $U\subseteq \F_q^t$, and a full-rank matrix $X=[\mathbf{x}_1|\cdots|\mathbf{x}_t]\in \mathbb F_q^{n\times t}$ such that every column of $X$ has weight at most $E$ and all $U$-linear combinations of the columns of $X$
are codewords. Equivalently, there is no such triple $(t,U,X)$ satisfying
\[
\mathrm{rank}(X)=t,\qquad
\mathrm{wt}(\mathbf{x}_j)\le E\ \text{ for all }j\in[t],
\qquad
X\cdot U\subseteq \cC.
\]
\end{prop}

\begin{proof}
Let $\eta:=\frac12\left(1-R-\frac{\theta}{b}\right)>0.$ We first describe the whole family of local profiles that may arise from a matrix $X$ as in the statement. For such a matrix, the induced profile depends only on which columns of $X$ are nonzero at each coordinate. Thus, in the first part of the proof we enumerate all possible choices of these column sets, before any specific matrix $X$ is fixed.

Fix an integer $t\in\{b+1,\ldots,T\}$ and a $b$-dimensional subspace $U\subseteq \mathbb F_q^t$. Choose one basis $\mathbf{u}_1,\ldots,\mathbf{u}_b$ of $U$, and write $\mathbf{m}_j:=((\mathbf{u}_1)_j,\ldots,(\mathbf{u}_b)_j)\in \mathbb F_q^b$ for $j\in [t]$. If a matrix $X$ is later given, then at coordinate $i$, the $i$-th row of $[X\mathbf{u}_1|\cdots|X\mathbf{u}_b]$ is a linear combination of exactly those vectors $\mathbf{m}_j$ for which $X_{i,j}\ne 0$. Thus, to cover all possible matrices $X$, it suffices to let the subsets $J_i\subseteq[t]$ range over all choices and define $V_i:=\mathrm{span}\{\mathbf{m}_j:j\in J_i\}\subseteq \mathbb F_q^b$ for $i\in [n]$. Let $\mathcal{V}=(V_1,\ldots,V_n)$.

Consider such a profile $\mathcal{V}$ with $\sum_{i=1}^n \dim V_i\le \theta n.$ Taking $W=0$ in the definition of the threshold rate eq~\eqref{eq:threshold-rate} gives 
\[
R_{\mathcal{V},\mathbb F_q^b}\ge 1-\frac{\sum_{i=1}^n \dim V_i}{nb}\ge 1-\frac{\theta}{b}=R+2\eta.
\]
Hence, Lemma~\ref{lem:rs-fixed-rowspan-lms} applies and gives
\[
\Pr\big[\cC\text{ contains }(\mathcal{V},\mathbb F_q^b)\big]\le(2b-1)\left(\frac{(4b)^{4b}Rn}{\eta q}\right)^{\eta n/(2b)}.
\]

We now take a union bound over all choices of $t$, $U$, and $J_1,\ldots,J_n$, and hence over all profiles in the family enumerated above. There are at most $T$ choices for $t$, at most $q^{b(t-b)}\le q^{T^2}$ choices for $U$, and at most $2^{tn}\le 2^{Tn}$ choices for the subsets $J_1,\ldots,J_n$. Therefore the probability that $\cC$ contains any such profile is at most $Tq^{T^2}2^{Tn}(2b-1)\left(\frac{(4b)^{4b}Rn}{\eta q}\right)^{\eta n/(2b)}.$ Since $b,T,\eta$ are fixed constants, choosing $C_{R,b,T,\theta}$ sufficiently large makes the last expression at most $\exp(-\Omega(n))$.

It remains to connect this enumeration back to the matrices $X$ in the proposition. Suppose that a tuple $(t,U,X)$ as in the statement exists. For this particular matrix $X$, take $J_i=J_i(X):=\{j\in[t]:X_{i,j}\ne 0\}.$ Then the profile induced by $X$ is one of the profiles enumerated above. Moreover, Lemma~\ref{lem:rs-structured-witness-profile} shows that $\cC$ contains this profile, and $\sum_{i=1}^n \dim V_i\le tE\le T\rho n<\theta n.$ This contradicts the preceding union bound, which showed that with high probability $\cC$ contains no profile in this enumerated family with total dimension at most $\theta n$. Hence no such tuple $(t,U,X)$ exists, except with probability $\exp(-\Omega(n))$.

\end{proof}

We can now combine the deterministic rank-reduction theorems from Sections~\ref{section:correlated-agreement} and~\ref{section:curve-correlated-agreement} with Proposition~\ref{prop:rs-profile-exclusion-uniform}. The deterministic part shows that a bad affine space or a bad degree-$\ell$ curve forces the existence of many linearly independent low-weight witness columns. After collecting these columns into a matrix $X$, the syndrome matrix $HX$ has small rank: in the affine-space case, all its columns lie in $\mathrm{span}(\mathbf{s}_0,\ldots,\mathbf{s}_m)$, and in the degree-$\ell$ curve case, all its columns lie in $\mathrm{span}(\mathbf{s}_0,\ldots,\mathbf{s}_\ell)$. Hence $\ker(HX)$ has positive dimension, and in the proof below we will choose a subspace $U\subseteq \ker(HX)$ of the dimension required by Proposition~\ref{prop:rs-profile-exclusion-uniform}. For every such $U$, we have $X\cdot U\subseteq \cC.$ Proposition~\ref{prop:rs-profile-exclusion-uniform} then shows that, with high probability, a random RS code admits no such full-rank low-weight matrix $X$ together with such a subspace $U$.

The theorem is stated with a free output radius $E^+$. Thus $E^+=E$ gives the no-slack statement in the i.i.d. random-RS model, while $E^+=E+\Omega(n)$ gives the fixed-slack form that will later transfer to the distinct-evaluation model. For $E\le E^+$, write $B_{E,E^+}:=\left\lfloor \frac{E^++1}{E^+-E+1}\right\rfloor$.

\begin{thm} \label{thm:iid-rs-ca-main}
Fix $R\in(0,1)$, $0<\varepsilon<(1-R)/2$, and fixed integers $m,\ell\ge 1$.  There exists a constant $C_{R,\varepsilon,m,\ell}$ such that the following holds for all sufficiently large $n$. Let $0<\rho\le 1-R-\varepsilon$, $E=\lfloor \rho n\rfloor$, and let $d=\big\lfloor (1-R-\varepsilon/4)n\big\rfloor$. Let $E^+$ be any integer with $E\le E^+<d$, and let $\cC=\mathrm{RS}_{\F_q}((\alpha_1,\ldots,\alpha_n);k)$ be the iid random RS code, where $q\ge n\cdot 2^{O_{R,m,\ell}(\varepsilon^{-3})}$.  Then, with probability at least $1-\exp(-\Omega(n))$, $\cC$ simultaneously satisfies:

\begin{enumerate}
\item \emph{Affine spaces.}  The code $\cC$ satisfies the $(m,E,E^+,K_m/q^m)$-space correlated-agreement property, where
\[
K_m:=\left\lfloor
q^{m-1}
\left\lfloor\frac{E^++1}{E^+-E+1}\right\rfloor
\left(\frac{d}{d-E}\right)^{
\left\lceil \frac{4(m+1)(1-R)}{\varepsilon}\right\rceil}
\right\rfloor .
\]

\item \emph{Curves.}  The code $\cC$ satisfies the $(\ell,E,E^+,K^{(\ell)}/q)$-correlated-agreement property for degree-$\ell$ curves, where
\[
K^{(\ell)}:=\left\lfloor
\ell
\left\lfloor\frac{E^++1}{E^+-E+1}\right\rfloor
\left(\frac{d}{d-E}\right)^{
\left\lceil \frac{4(\ell+1)(1-R)}{\varepsilon}\right\rceil}
\right\rfloor .
\]
\end{enumerate}
\end{thm}

\begin{proof}
By Lemma~\ref{lem:iid-rs-distance}, if the constant $C_{R,\varepsilon,m,\ell}$ is large enough, then with probability at least $1-\exp(-\Omega(n))$, the iid random RS code has dimension $k$ and minimum distance at least $d=\left\lfloor (1-R-\varepsilon/4)n\right\rfloor$. We will work on this event.

We first provide the parameter choice that will be used to apply Proposition~\ref{prop:rs-profile-exclusion-uniform}. For an integer $a\ge 1$, define $s_a:=\left\lceil \frac{4a(1-R)}{\varepsilon}\right\rceil,\,b_a:=s_a+1,\,T_a:=s_a+a+1$ and define $\theta_a:=\frac{T_a\rho+b_a(1-R)}2.$ Since $\rho\le 1-R-\varepsilon$, we have
\[
b_a(1-R)-T_a\rho=(s_a+1)(1-R)-(s_a+a+1)\rho\ge (s_a+a+1)\varepsilon-a(1-R)>0.
\]
Therefore $T_a\rho<\theta_a<b_a(1-R).$ Thus, Proposition~\ref{prop:rs-profile-exclusion-uniform} can be applied with parameters $(b,T,\theta)=(b_a,T_a,\theta_a).$ Let $\eta_a:=\frac12\left(1-R-\frac{\theta_a}{b_a}\right).$ By the proof of Proposition~\ref{prop:rs-profile-exclusion-uniform}, it is enough to require
\[
q\ge n\cdot\max\left\{2Rb_a,\,\frac{(4b_a)^{4b_a}R}{\eta_a}\cdot 2^{4b_aT_a/\eta_a}\right\}.
\]
We now estimate this quantity. Since $\eta_a=\frac12\left(1-R-\frac{\theta_a}{b_a}\right)=\frac{b_a(1-R)-T_a\rho}{4b_a}$ and $\rho\le 1-R-\varepsilon$, we have $b_a(1-R)-T_a\rho\ge (s_a+a+1)\varepsilon-a(1-R).$ By the choice $s_a=\left\lceil 4a(1-R)/\varepsilon\right\rceil$, the right-hand side is positive, and moreover $\eta_a=\Omega_a(\varepsilon)$. Also $b_a=O_a(\varepsilon^{-1})$ and $T_a=O_a(\varepsilon^{-1}).$ Therefore
\[
\log_2\left(\frac{(4b_a)^{4b_a}R}{\eta_a}\cdot 2^{4b_aT_a/\eta_a}\right)=O_a(\varepsilon^{-1}\log(1/\varepsilon))+O_a(\log(1/\varepsilon))+O_a(\varepsilon^{-3})=O_a(\varepsilon^{-3}).
\]
The term $2Rb_a$ is also bounded by $2^{O_a(\varepsilon^{-3})}$. Hence, the field-size condition required by Proposition~\ref{prop:rs-profile-exclusion-uniform} is of the form $q\ge n\cdot 2^{O_a(\varepsilon^{-3})}.$ Since we only use this estimate for $a=m+1$ and $a=\ell+1$, both applications of Proposition~\ref{prop:rs-profile-exclusion-uniform} are covered by $q\ge n\cdot 2^{O_{m,\ell}(\varepsilon^{-3})}.$ The remaining requirements, namely the requirement in  Lemma~\ref{lem:iid-rs-distance} and the inequalities $K_m<q^m$ and $K^{(\ell)}<q$, are also implied by the same bound. Thus all field-size conditions in the proof follow from $q\ge n\cdot 2^{O_{R,m,\ell}(\varepsilon^{-3})}.$

We prove the affine-space statement first. Let $s:=\left\lceil \frac{4(m+1)(1-R)}{\varepsilon}\right\rceil$. By the definition of $K_m$, $K_m+1>B_{E,E^+}\,q^{m-1}\left(\frac{d}{d-E}\right)^s$. Suppose, for contradiction, that there is an affine syndrome space $S=\mathbf{s}_0+\mathrm{span}(\mathbf{s}_1,\ldots,\mathbf{s}_m)$ such that $|S\cap H_E|>K_m$ and $S$ has no $(H,m,E,E^+)$-correlated agreement. Choose $K_m+1$ points from $S\cap H_E$. Write each selected point uniquely as $\mathbf{s}_0+\sum_{i=1}^m \mathbf{\beta}_i \mathbf{s}_i$ with parameter $\mathbf{\beta}=(\beta_1,\ldots,\beta_m)\in\mathbb F_q^m$, and choose one low-weight preimage for each of them. This gives an $E$-witness matrix with corresponding points in $\mathbb F_q^m$. Write $h:=\dim_{\mathbb F_q}\mathrm{span}(\mathbf{s}_0,\mathbf{s}_1,\ldots,\mathbf{s}_m).$ Since $S$ has dimension $m$, we have $m\le h\le m+1$.

Since $K_m+1>B_{E,E^+}q^{m-1}\left(\frac{d}{d-E}\right)^s\ge q^{m-1},$ the corresponding parameter points in $\mathbb F_q^m$ cannot all lie in an affine hyperplane. Equivalently, the augmented vectors $(1,\mathbf{\beta}_1,\ldots,\mathbf{\beta}_m)$ of these parameter points span $\mathbb F_q^{m+1}$. Hence the corresponding syndrome columns $\mathbf{s}_0+\sum_{i=1}^m \mathbf{\beta}_i \mathbf{s}_i$ span $\mathrm{span}(\mathbf{s}_0,\mathbf{s}_1,\ldots,\mathbf{s}_m)$. In particular, the syndrome matrix has rank $h$, and every resulting witness matrix has rank at least $h$. By Theorem~\ref{thm:min-rank-bad-m-space}, every such witness matrix has rank at least $h+s+1$. Hence we may choose $t:=h+s+1$ linearly independent witness columns and collect them into a matrix $X\in \mathbb F_q^{n\times t}.$ All columns of $X$ have weight at most $E$. The syndrome columns of $X$ lie in $\mathrm{span}(\mathbf{s}_0,\mathbf{s}_1,\ldots,\mathbf{s}_m),$ so $\mathrm{rank}(HX)\le h.$ Therefore, $\dim\ker(HX)\ge t-h=s+1.$ Choose a $s+1$-dimensional subspace $U\subseteq \ker(HX)$. Then $X\cdot U\subseteq \cC$, because for every $\mathbf{u}\in U$, we have $H(X\mathbf{u})=(HX)\mathbf{u}=\mathbf{0}.$ Also, if we take $T:=s+m+2,$ then $s+2\le t\le T.$ Thus, $X$ and $U$ give exactly the type of configuration that Proposition~\ref{prop:rs-profile-exclusion-uniform} shows cannot occur with high probability, with the above parameter choice for $a=m+1$. This contradiction proves the affine-space statement.

The proof for degree-$\ell$ curves is the same, with Theorem~\ref{thm:curve-witness-rank} replacing Theorem~\ref{thm:min-rank-bad-m-space}. Let $s:=\left\lceil \frac{4(\ell+1)(1-R)}{\varepsilon}\right\rceil$. By the definition of $K^{(\ell)}$, $K^{(\ell)}+1>\ell B_{E,E^+}\left(\frac{d}{d-E}\right)^s$. Suppose that a degree-$\ell$ syndrome curve $S(\alpha)=\mathbf{s}_0+\sum_{i=1}^{\ell}\alpha^i \mathbf{s}_i$ has more than $K^{(\ell)}$ points $\alpha\in\mathbb F_q$ satisfying $S(\alpha)\in H_E$, but has no $(H,\ell,E,E^+)$-correlated agreement. Choose $K^{(\ell)}+1$ such distinct parameters $\alpha$, and choose one low-weight preimage for each corresponding syndrome point $S(\alpha)$. This gives the corresponding $E$-witness matrix.

Since $K^{(\ell)}+1>\ell B_{E,E^+}\left(\frac{d}{d-E}\right)^s\ge \ell+1,$ we have selected at least $\ell+1$ distinct parameters $\alpha\in\mathbb F_q$. By the Vandermonde determinant, the vectors $(1,\alpha,\alpha^2,\ldots,\alpha^\ell)$ associated with any $\ell+1$ distinct selected points span $\mathbb F_q^{\ell+1}$. Hence the corresponding syndrome columns $\mathbf{s}_0+\sum_{i=1}^{\ell}\alpha^i \mathbf{s}_i$ span $\mathrm{span}(\mathbf{s}_0,\ldots,\mathbf{s}_\ell)$. Thus the syndrome matrix has rank $h=\dim_{\mathbb F_q}\mathrm{span}(\mathbf{s}_0,\ldots,\mathbf{s}_\ell).$ If $h=0$, then the curve has trivial correlated agreement, so we may assume $1\le h\le \ell+1$. By Theorem~\ref{thm:curve-witness-rank}, every corresponding witness matrix has rank at least $h+s+1$. Choose $t:=h+s+1$ linearly independent witness columns and collect them into a matrix $X\in\mathbb F_q^{n\times t}$. Then all columns of $X$ have weight at most $E$, and $\mathrm{rank}(HX)\le h.$ Hence, $\dim\ker(HX)\ge t-h=s+1.$ Choose a $s+1$-dimensional subspace $U\subseteq \ker(HX)$. Thus, $X$ and $U$ give exactly the type of configuration that Proposition~\ref{prop:rs-profile-exclusion-uniform} shows cannot occur with high probability, with the above parameter choice for $a=\ell+1$. This proves the curve statement, and hence the theorem.
\end{proof}

\begin{rmk}\label{rmk:compare-random-rs}
Theorem~\ref{thm:iid-rs-ca-main} already yields both the no-slack and the fixed-slack regimes. The proof of Theorem~\ref{thm:iid-rs-ca-main} shows that it suffices to take $q \ge n\cdot 2^{O(\varepsilon^{-3})}$ for affine $m$-spaces and $q \ge n\cdot 2^{O_\ell(\varepsilon^{-3})}$ with degree-$\ell$ curves.

Now we compare these parameters with the random-RS consequences of~\cite{goyal2025optimal}. Their theorem applies up to radius $\rho\le 1-R-2\varepsilon.$ Concretely, this gives a sufficient condition for alphabet size $q \ge n\cdot 2^{O(\varepsilon^{-7})}$ in the line/affine-space case, and more generally $q \ge n\cdot 2^{O_\ell(\varepsilon^{-(2\ell+5)})}$ for degree-$\ell$ curves.

\end{rmk}

\begin{table}[H]
\centering
\renewcommand{\arraystretch}{1.18}
\caption{Comparison with the random-RS consequences of~\cite{goyal2025optimal}. The comparison is written in the same $\varepsilon$-parameterization as Theorem~\ref{thm:iid-rs-ca-main}.}
\label{tab:rrs-gg25-comparison}
\begin{tabularx}{\textwidth}{@{}
>{\raggedright\arraybackslash}p{0.22\textwidth}
>{\raggedright\arraybackslash}X
>{\raggedright\arraybackslash}X
@{}}
\toprule
\textbf{Aspect}
&
\textbf{Consequences of~\cite{goyal2025optimal}}
&\qquad\qquad
\textbf{This work}
\\
\midrule

Radius
&
\[
\rho \le 1-R-\varepsilon
\]
&
\[
\rho \le 1-R-\varepsilon
\]
\\

Affine spaces / lines
&
\[
q \ge n \cdot 2^{O(\varepsilon^{-7})}
\]
&
\[
q \ge n \cdot 2^{O(\varepsilon^{-3})}
\]
\\

Degree-$\ell$ curves
&
\[
q \ge n \cdot 2^{O_\ell(\varepsilon^{-(2\ell+5)})}
\]
&
\[
q \ge n \cdot 2^{O_\ell(\varepsilon^{-3})}
\]
\\

\bottomrule
\end{tabularx}
\end{table}

Theorem~\ref{thm:iid-rs-ca-main} was stated for the i.i.d. random RS model. We now explain how its fixed-slack consequences transfer to the more standard model in which the evaluation points are sampled without repetitions.

Let $\cC=\mathrm{RS}_{\mathbb F_q}((\alpha_1,\ldots,\alpha_n);k),$ where $\alpha_1,\ldots,\alpha_n$ are sampled independently and uniformly from $\mathbb F_q$. Let $\cC'=\mathrm{RS}_{\mathbb F_q}((\alpha_1',\ldots,\alpha_n');k),$ where $(\alpha_1',\ldots,\alpha_n')$ is sampled uniformly among all ordered $n$-tuples of distinct elements of $\mathbb F_q$. The following lemma of~\cite{levi2025random} relates these two models.

\begin{lemma}[Lemma~A.2~\cite{levi2025random}]\label{lem:distinct-rs}
Write $k=Rn$. Then, there exists a coupling $(\mathcal{C},\mathcal{C}')$ where
$\mathcal{C}\subseteq \mathbb{F}_q^n$ is a random RS code of dimension $k$ and
$\mathcal{C}'\subseteq \mathbb{F}_q^n$ is a random RS code of dimension $k$
without repetitions, such that there exists a linear bijection
$\varphi:\mathcal{C}\to \mathcal{C}'$ with
\[
\Pr\left[
\max_{x\in\mathcal{C}}
\bigl\{ \mathrm{wt}(x-\varphi(x)) \bigr\}
\geq
n-q\left(1-e^{-n/q}\right)(1-\delta)
\right]
\leq
\left(
\frac{e^{-\delta}}{(1-\delta)^{1-\delta}}
\right)^{q(1-e^{-n/q})}
\]
for all $0<\delta<1$.
\end{lemma}

We now explain how to use Lemma~\ref{lem:distinct-rs} to compare the two random RS models. The lemma allows us to sample the i.i.d. code $\cC$ and the distinct-evaluation code $\cC'$ together, together with a linear bijection $\varphi:\cC\to \cC'.$ Let $\Delta:=\max_{x\in \cC}\mathrm{wt}(\mathbf{x}-\varphi(\mathbf{x})).$ Thus $\Delta$ means that every codeword $\mathbf{x}\in \cC$ differs from its image $\varphi(\mathbf{x})\in \cC'$ in at most $\Delta$ coordinates.

Lemma~\ref{lem:distinct-rs} shows that $\Delta$ is small when $q$ is a sufficiently large constant multiple of $n$. Indeed, $n-q(1-e^{-n/q})=O(n^2/q)$ and hence $n-q(1-e^{-n/q})(1-\delta)\le \delta n+O(n^2/q).$ Therefore, for every fixed $\eta>0$, by first choosing $\delta>0$ sufficiently small and then taking $q\ge C_\eta n$ with $C_\eta$ sufficiently large, we have $\Delta\le \eta n$ with probability at least $1-\exp(-\Omega_\eta(n))$.

The meaning of this bound is simple. Suppose a word $\mathbf{y}\in\mathbb F_q^n$ is $E$-close to a codeword $\mathbf{c}'\in \cC'$. Since $\varphi$ is a bijection, we can write $\mathbf{c}'=\varphi(\mathbf{c})$ for a unique $\mathbf{c}\in \cC$. Because $\mathbf{c}$ and $\varphi(\mathbf{c})$ differ in at most $\Delta$ coordinates, $d(\mathbf{y},\mathbf{c})\le d(\mathbf{y},\mathbf{c}')+d(\mathbf{c}',\mathbf{c})\le E+\Delta.$ Thus every $E$-close codeword explanation over $\cC'$ becomes an $(E+\Delta)$-close codeword explanation over $\cC$.

Consider first an affine $m$-space instance. After pulling all nearby codewords back from $\cC'$ to $\cC$, we may apply Theorem~\ref{thm:iid-rs-ca-main} to the i.i.d. code $\cC$. The conclusion gives codewords $\mathbf{c}_0,\mathbf{c}_1,\ldots,\mathbf{c}_m\in \cC$ such that the affine code space $\mathbf{c}_0+\mathrm{span}(\mathbf{c}_1,\ldots,\mathbf{c}_m)\subseteq \cC$ agrees with the received affine space outside a small set of coordinates.

For degree-$\ell$ curves, the same statement gives codewords $\mathbf{c}_0,\mathbf{c}_1,\ldots,\mathbf{c}_\ell\in \cC$ such that the code curve $\mathbf{c}_0+\sum_{i=1}^{\ell}\alpha^i \mathbf{c}_i$ has small disagreement support.

We then map these codewords forward to $\cC'$. In the affine-space case, this gives $\varphi(\mathbf{c}_0),\varphi(\mathbf{c}_1),\ldots,\varphi(\mathbf{c}_m)$ and in the curve case, this gives $\varphi(\mathbf{c}_0),\varphi(\mathbf{c}_1),\ldots,\varphi(\mathbf{c}_\ell).$ Since $\varphi$ is linear, the affine-space structure and the degree-$\ell$ curve structure are preserved. The only possible loss is in the support size: each coefficient codeword $\mathbf{c}_i$ may change on at most $\Delta$ coordinates when replaced by $\varphi(\mathbf{c}_i)$. Therefore the final disagreement support increases by at most $a\Delta,$ where $a=m+1$ for affine $m$-spaces and $a=\ell+1$ for degree-$\ell$ curves.

We now make the fixed-slack transfer precise. Suppose the desired input radius is $E=\lfloor \rho n\rfloor$ for $\rho\le 1-R-\varepsilon,$ and the desired output radius is $E^+=E+\lceil \alpha n\rceil$ for some fixed $\alpha>0$. We apply Theorem~\ref{thm:iid-rs-ca-main} to the i.i.d. code $\cC$ with parameter $\varepsilon/2$. Let $d':=\left\lfloor (1-R-\varepsilon/8)n\right\rfloor$ be the corresponding distance parameter. Before using Lemma~\ref{lem:distinct-rs}, we take a union bound over all $O(n^2)$ possible integer radius pairs $0\le E_0\le E_1<d'$ with $E_0\le (1-R-\varepsilon/2)n.$ Thus, with probability $1-\exp(-\Omega(n))$, Theorem~\ref{thm:iid-rs-ca-main} holds simultaneously for all such radius pairs.

Choose $q/n$ large enough so that, with high probability, $\Delta\le\min\left\{\frac{\varepsilon n}{4},\frac{\alpha n}{2(a+1)}\right\}.$ On this event, define $\widetilde E:=E+\Delta$ and $\widetilde E^+:=\min\{E^+-a\Delta,\ d'-1\}.$ Then, for all sufficiently large $n$, $\widetilde E\le (1-R-\varepsilon/2)n$ and $\widetilde E\le \widetilde E^+<d'.$ Indeed,
\[
E+\Delta\le (1-R-\varepsilon)n+\frac{\varepsilon n}{4}\le (1-R-\varepsilon/2)n.
\]
Also, the choice $\Delta\le \frac{\alpha n}{2(a+1)}$ implies, up to harmless integer rounding, that $E^+-a\Delta\ge E+\Delta.$ Hence $\widetilde E\le \widetilde E^+$.

Now apply Theorem~\ref{thm:iid-rs-ca-main} to $\cC$ with input radius $\widetilde E$ and output radius $\widetilde E^+$. Mapping the resulting codewords in $\cC$ forward by $\varphi$ gives codewords in $\cC'$. The affine-space or curve structure remains valid in $\cC'$, and the output support increases by at most $a\Delta$. Therefore the final output radius is at most $\widetilde E^+ + a\Delta \le E^+.$

This proves that every fixed-slack correlated-agreement consequence of Theorem~\ref{thm:iid-rs-ca-main} transfers from the i.i.d. random RS model to the distinct-evaluation random RS model, after increasing the constant in the field-size condition. The exact no-slack case does not follow from this comparison, because the final map from $\cC$ to $\cC'$ may increase the support by $a\Delta$.

\bibliographystyle{alpha}
\bibliography{refs}

\appendix
\section{Proofs of Lemma~\ref{lem:line-to-space-gap}, Lemma~\ref{lem:correlated-line-to-space} and Theorem~\ref{thm:no-slack-obstruction}}\label{appendix:A}
\begin{proof}[Proof of Lemma~\ref{lem:line-to-space-gap}]
Let $U\subseteq \mathbb{F}_q^r$ be an affine subspace such that $\frac{|U\cap B_E|}{|U|}>\tau\frac{q}{q-1}$. We show that $U\subseteq H_{E^+}$. If $\dim(U)=0$, then $U$ is a singleton and the claim is immediate. Hence assume
$\dim(U)\ge 1$. Let $\mathbf{u}_0\in U$ be arbitrary. We claim that $\mathbf{u}_0\in H_{E^+}$.

If already $\mathbf{u}_0\in H_{E^+}$, there is nothing to prove. So assume $\mathbf{u}_0\notin H_{E^+}$. Since $E\le E^+$, this also implies $\mathbf{u}_0\notin H_E$. Now choose $\mathbf{u}'$ uniformly from $U\setminus\{\mathbf{u}_0\}$, and consider the affine line $L(\mathbf{u}_0,\mathbf{u}') := \mathbf{u}_0+\mathbb{F}_q(\mathbf{u}'-\mathbf{u}_0)\subseteq U$. For each fixed $\alpha\in \mathbb{F}_q^\ast$, the map $\mathbf{u}' \mapsto \mathbf{u}_0+\alpha(\mathbf{u}'-\mathbf{u}_0)$ is a bijection of $U\setminus\{\mathbf{u}_0\}$ onto itself. Therefore,
\[
\mathbb{E}_{\mathbf{u}'\in U\setminus\{u_0\}}
\Pr_{\alpha\in \mathbb{F}_q^\ast}
\bigl[\mathbf{u}_0+\alpha(\mathbf{u}'-\mathbf{u}_0)\in H_E\bigr]
=
\Pr_{\mathbf{u}\in U\setminus\{\mathbf{u}_0\}}[\mathbf{u}\in H_E].
\]
Since $\mathbf{u}_0\notin H_E$, the right-hand side equals $\frac{|U\cap B_E|}{|U|-1}$. Using $\frac{|U\cap B_E|}{|U|}>\tau q/(q-1)$, we obtain $\frac{|U\cap B_E|}{|U|-1}> \frac{\tau q}{q-1}\cdot \frac{|U|}{|U|-1} > \frac{\tau q}{q-1}$. Hence, there exists some $\mathbf{u}'\in U\setminus\{\mathbf{u}_0\}$ such that
\[
\Pr_{\alpha\in \mathbb{F}_q^\ast}
\bigl[\mathbf{u}_0+\alpha(\mathbf{u}'-\mathbf{u}_0)\in H_E\bigr]
>
\frac{\tau q}{q-1}.
\]
Let $L:=\{\mathbf{u}_0+\alpha(\mathbf{u}'-\mathbf{u}_0):\alpha\in \F_q\}.$ As $\mathbf{u}_0\notin H_E$, the elements of $L\cap H_E$ are all among the $q-1$ nonzero-$\alpha$ elements on $L$. Therefore,
\[
\frac{|L\cap H_E|}{q}
=
\frac{q-1}{q}
\Pr_{\alpha\in \mathbb{F}_q^\ast}
\bigl[\mathbf{u}_0+\alpha(\mathbf{u}'-\mathbf{u}_0)\in H_E\bigr]
>
\tau.
\]
By the assumed line-level proximity-gap property, it follows that $L\subseteq H_{E^+}$. In particular, $\mathbf{u}_0\in H_{E^+}$. Since $\mathbf{u}_0\in U$ was arbitrary, we conclude that $U\subseteq H_{E^+}$, as required.
\end{proof}

\begin{proof}[Proof of Lemma~\ref{lem:correlated-line-to-space}]
Fix an integer $m\ge 2$, and let $U=\mathbf{u}_0+\mathrm{span}(\mathbf{u}_1,\dots,\mathbf{u}_m)\subseteq \F_q^n$ be an affine space. Write $S=U_H=\mathbf{s}_0+\mathrm{span}(\mathbf{s}_1,\dots,\mathbf{s}_m)\subseteq \F_q^r$ where $\mathbf{s}_i:=H\mathbf{u}_i$ for $0\leq i\leq m$. Assume that $\mu_S(E)>\tau\frac{q}{q-1}.$ We must show that there exists a matrix $X=[\mathbf{x}_0|\cdots|\mathbf{x}_m]\in \F_q^{n\times (m+1)}$ such that $H\mathbf{x}_i=\mathbf{s}_i\,(0\leq i\leq m)$ and $\mathrm{rowwt}(X)\le E$.

Since line correlated agreement implies line proximity gap, Lemma~\ref{lem:line-to-space-gap} yields $S\subseteq H_E$. Hence, by Lemma~\ref{lem:syndrome-characterization}, $d(\mathbf{u},\cC)\le E$ for all $\mathbf{u}\in U$. Let $E^\ast:=\max_{\mathbf{u}\in U} d(\mathbf{u},\cC)$, and choose $\mathbf{u}^\ast\in U$ such that $d(\mathbf{u}^\ast,\cC)=E^\ast.$ Then $E^\ast\le E$. Let $\{\mathbf{c}_1,\dots,\mathbf{c}_L\}:=\{\mathbf{c}\in \cC:\ d(\mathbf{u}^\ast,\mathbf{c})=E^\ast\}$. By assumption, $L<q$.

For each $j\in[L]$, define $T_j:=\{i\in[n]: (\mathbf{c}_j)_i=\mathbf{u}^\ast_i\}$, so that $|T_j|=n-E^\ast$. Let $U_j:=\{\mathbf{u}\in U:\ \pi_{T_j}(\mathbf{u})\in \pi_{T_j}(\cC)\}$, where $\pi_{T_j}$ denotes the coordinate projection onto $T_j$. Each $U_j$ is an affine subspace of $U$. We claim that $U=\bigcup_{j=1}^L U_j$.

Indeed, fix any $\mathbf{u}\in U$, and consider the affine line $\ell:=\mathbf{u}^\ast+\F_q(\mathbf{u}-\mathbf{u}^\ast)\subseteq U$. Every element of $\ell$ has distance at most $E^\ast$ from $\cC$, hence the corresponding syndrome line has $\mu(E^\ast)=1>\tau$. By the assumed $(1,E^\ast,E^\ast,\tau)$-line correlated agreement property, there exists a code line $V=\mathbf{c}_0+\F_q \mathbf{c}_1\subseteq \cC$ such that
\[
\left|
\{\,i\in[n]: (\mathbf{c}_{0})_i\neq \mathbf{u}^\ast_i \text{ or } (\mathbf{c}_0+\mathbf{c}_1)_i\neq \mathbf{u}_i\,\}
\right|
\le E^\ast.
\]
Let $T:=\{i\in[n]: (\mathbf{c}_{0})_i=\mathbf{u}^\ast_i \text{ and } (\mathbf{c}_0+\mathbf{c}_1)_i=\mathbf{u}_i\}$. Then $|T|\ge n-E^\ast.$ Since $d(\mathbf{u}^\ast,\mathbf{c}_0)\le E^\ast=d(\mathbf{u}^\ast,\cC)$, we have $d(\mathbf{u}^\ast,\mathbf{c}_0)=E^\ast$, so $\mathbf{c}_0=\mathbf{c}_j$ for some $j\in[L]$. Hence $T\subseteq T_j$, and because $|T_j|=n-E^\ast\le |T|$, we get $T=T_j$. Therefore
\[
\pi_{T_j}(\mathbf{u})=\pi_{T_j}(\mathbf{c}_0+\mathbf{c}_1)\in \pi_{T_j}(\cC),
\]
and $\mathbf{u}\in U_j$. This proves the claim.

Now $U$ is a union of $L<q$ affine subspaces $U_1,\dots,U_L$. Since every proper affine
subspace of $U$ has size at most $|U|/q$, one of them must equal $U$. Thus there exists
$j^\ast\in[L]$ such that $U_{j^\ast}=U$. Let $T:=T_{j^\ast}$. Then for every $\mathbf{u}\in U$, $\pi_T(\mathbf{u})\in \pi_T(\cC).$ In particular, we may choose codewords $\mathbf{c}_0,\mathbf{c}_1,\dots,\mathbf{c}_m\in \cC$ such that $\pi_T(\mathbf{c}_i)=\pi_T(\mathbf{u}_i)$ for $0\leq i\leq m$.

Define $\mathbf{x}_i:=\mathbf{u}_i-\mathbf{c}_i$ where $0\leq i\leq m$ and let $X=[\mathbf{x}_0|\cdots|\mathbf{x}_m]$. Since each $\mathbf{c}_i\in \cC$, we have $H\mathbf{x}_i=\mathbf{s}_i$ for $0\leq i\leq m$. Moreover, $\pi_T(\mathbf{x}_i)=\mathbf{0}$ for every $0\le i\le m$, so every row of $X$ indexed by $T$ is zero. Hence, $\mathrm{rowwt}(X)\le n-|T|=E^\ast\le E$. This proves the claim.
\end{proof}

\begin{proof}[Proof of Theorem~\ref{thm:no-slack-obstruction}]
Choose pairwise distinct coordinates $i_1,\dots,i_K \in [n]$ and further choose pairwise distinct coordinates $j_1,\dots,j_{E+1-K}\in [n]$ disjoint from $\{i_1,\dots,i_K\}$.

Define $\mathbf{x}_2\in\mathbb F_q^n$ by
\[
(\mathbf{x}_2)_{i_t}=1 \quad (t\in[K]),\qquad (\mathbf{x}_2)_u=0 \quad \text{otherwise},
\]
and define $\mathbf{x}_1\in\mathbb F_q^n$ by
\[
(\mathbf{x}_1)_{i_t}=-\alpha_t \quad (t\in[K]),\qquad
(\mathbf{x}_1)_{j_s}=1 \quad (s\in[E+1-K]),\qquad
(\mathbf{x}_1)_u=0 \quad \text{otherwise}.
\]

Fix $\alpha\in\mathbb F_q$. Then for each $s\in[E+1-K]$ we have $x(\alpha)_{j_s}=1$, so these coordinates contribute exactly $E+1-K$ nonzero entries. On the other hand, for each $t\in[K]$, $(\mathbf{x}_1+\alpha\mathbf{x}_2)_{i_t}=-\alpha_t+\alpha$ which vanishes if and only if $\alpha=\alpha_t$.

Hence, if $\alpha=\alpha_t$ for some $t\in[K]$, then exactly one among the coordinates
$i_1,\dots,i_K$ vanishes, while the other $K-1$ remain nonzero. Therefore,
\[
\wt\bigl(\mathbf{x}_1+\alpha_t\mathbf{x}_2\bigr)=(E+1-K)+(K-1)=E.
\]
If instead $\alpha\notin\{\alpha_1,\dots,\alpha_K\}$, then none of the coordinates
$i_1,\dots,i_K$ vanishes, and thus $\wt\bigl(\mathbf{x}_1+\alpha\mathbf{x}_2\bigr)=(E+1-K)+K=E+1$. This proves the first part.

Now fix a linear code $\cC$ with parity check matrix $H$ and $d(\cC)\ge 2E+2$, and let $\mathbf{s}_0:=H\mathbf{x}_1,\,\mathbf{s}_1:=H\mathbf{x}_2,\, L=\{\mathbf{s}_0+\alpha \mathbf{s}_1:\alpha\in\mathbb F_q\}$. For each $j\in[K]$, since $\wt(\mathbf{x}_1+\alpha_j\mathbf{x}_2)=E$, we have $\mathbf{s}_0+\alpha_j \mathbf{s}_1=H(\mathbf{x}_1+\alpha_j\mathbf{x}_2)\in H_E$.

We claim that these $K$ syndromes are pairwise distinct. Indeed, if $H(\mathbf{x}_1+\alpha_i\mathbf{x}_2)=H(\mathbf{x}_1+\alpha_j\mathbf{x}_2)$, then $(\alpha_i-\alpha_j)\mathbf{x}_2 \in \cC$. Since $\alpha_i\neq \alpha_j$ and $\wt(\mathbf{x}_2)=K$, this gives a nonzero codeword of weight $K<d(\cC)$ which implies a contradiction. Hence, $|L\cap H_E|\ge K$.

It remains to prove that $L\not\subseteq H_E$. Since $K<q$, choose $\alpha_\star \in \mathbb F_q\setminus\{\alpha_1,\dots,\alpha_K\}$. By the first part, $\wt\bigl(\mathbf{x}_1+\alpha_\star\mathbf{x}_2)=E+1$. Suppose for contradiction that $\mathbf{s}_0+\alpha_\star \mathbf{s}_1 = H(\mathbf{x}_1+\alpha_\star\mathbf{x}_2)\in H_E$. Then there exists $\mathbf{z}\in B_E$ such that $H\mathbf{z}=H(\mathbf{x}_1+\alpha_\star\mathbf{x}_2)$. Thus, $\mathbf{0}\neq \mathbf{c}:=(\mathbf{x}_1+\alpha_\star\mathbf{x}_2)-\mathbf{z} \in \cC.$

However,
\[
\wt(\mathbf{c})\le \wt\bigl(\mathbf{x}_1+\alpha_\star\mathbf{x}_2\bigr)+\wt(\mathbf{z})\le (E+1)+E=2E+1,
\]
contradicting $d(\cC)\ge 2E+2$. Therefore, $L\not\subseteq H_E$ combining this with $|L\cap H_E|\ge K$ proves the theorem.
\end{proof}

\section{Auxiliary lemmas for probabilistic analysis} \label{appendix:B}
\begin{lemma}\label{lem:random-image}
Let $H\in\mathbb{F}_q^{r\times n}$ be uniformly random matrix. If $X\in\mathbb{F}_q^{n\times t}$
has full column rank $t$, then for every fixed $Y\in\mathbb{F}_q^{r\times t}$, we have $\Pr_H[HX=Y]=q^{-rt}$.
\end{lemma}

\begin{proof}
Write the rows of $H$ as $h_1,\dots,h_r\in\mathbb{F}_q^n$. Then
\[
HX=
\begin{bmatrix}
\mathbf{h}_1X\\
\vdots\\
\mathbf{h}_rX
\end{bmatrix}.
\]
Since $X$ has full column rank, the map $\mathbf{z}\mapsto \mathbf{z}X,\,\mathbf{z}\in\mathbb{F}_q^n,$ is surjective onto $\mathbb{F}_q^t$. Hence, if $\mathbf{h}_i$ is uniformly distributed on
$\mathbb{F}_q^n$, then $\mathbf{h}_iX$ is uniformly distributed on $\mathbb{F}_q^t$. Since the rows
$\mathbf{h}_1,\dots,\mathbf{h}_r$ are independent, the random vectors $\mathbf{h}_1X,\dots,\mathbf{h}_rX$ are independent as well.
Therefore,
\[
\Pr_H[HX=Y]
=
\prod_{i=1}^r \Pr[\mathbf{h}_iX = Y_{i,*}]
=
(q^{-t})^r
=
q^{-rt}.
\]
\end{proof}

\begin{lemma}\label{lem:H_q-bound}
For every $0<\rho<1-1/q$,
\[
H_q(\rho)
=
\rho+\frac{H_2(\rho)+\rho\log_2(1-1/q)}{\log_2 q}
\le
\rho+\frac{H_2(\rho)}{\log_2 q},
\]
where $H_q(x):=\rho\log_x(x-1)-x\log_q x-(1-x)\log_q(1-x)$ denotes the $q$-ary entropy function. In particular, if $q\ge \left(\frac{2}{\varepsilon}\right)^{1/\varepsilon}$, then
\[
H_q(\rho)\le
\rho+\frac{\varepsilon}{1+\log_2(1/\varepsilon)}\,H_2(\rho)
\le
\rho+\frac{\varepsilon}{1+\log_2(1/\varepsilon)}.
\]
\end{lemma}

\begin{proof}
Using $\log_2(q-1)=\log_2 q+\log_2(1-1/q)$, we obtain
\[
H_q(\rho)
=
\rho\frac{\log_2(q-1)}{\log_2 q}
+\frac{H_2(\rho)}{\log_2 q}
=
\rho+\frac{H_2(\rho)+\rho\log_2(1-1/q)}{\log_2 q}.
\]
Since $\log_2(1-1/q)<0$, this yields $H_q(\rho)\le \rho+\frac{H_2(\rho)}{\log_2 q}$. Now if $q\ge (2/\varepsilon)^{1/\varepsilon}$, then $\frac{1}{\log_2 q}\le \frac{\varepsilon}{1+\log_2(1/\varepsilon)}$. Substituting this into the preceding bound proves the claim.
\end{proof}

\begin{lemma}\label{lem:ball-volume-bound}
For every integer $E\ge 0$, $|B_E|=\sum_{i=0}^E \binom{n}{i}(q-1)^i.$ Moreover:
\begin{enumerate}
    \item if $q$ is fixed and $E=\lfloor \rho n\rfloor$ with $0<\rho<1-1/q$, then $|B_E| = q^{H_q(\rho)n + o(n)}\leq q^{(\rho+\log_q 2)n+o(n)}$. In particular, if $q\geq (2/\varepsilon)^{1/\varepsilon}$, then $|B_E|\leq q^{(\rho+\frac{\varepsilon}{1+\log_2 1/\varepsilon})n+o(n)}$.
    \item if $q=\Theta(n)$ and $E=\Theta(n)$, then $|B_E|\le q^{E+o(n)}.$
\end{enumerate}
\end{lemma}
\section{Proof of~\cref{section:curve-correlated-agreement}}\label{appendix:C}
\begin{proof}[Proof of Lemma~\ref{lem:curve-ca-equivalence}]
Assume (1). Define $\mathbf{x}_j := \mathbf{u}_j-\mathbf{c}_j,\, 0\le j\le \ell.$ Since $\mathbf{c}_j\in \cC$, we have $H\mathbf{x}_j = H(\mathbf{u}_j-\mathbf{c}_j)=H\mathbf{u}_j=\mathbf{s}_j$ for every $0\le j\le \ell$. Moreover,
\[
\mathrm{rowsupp}(X)
= \bigcup_{j=0}^\ell \mathrm{supp}(\mathbf{x}_j)= \bigcup_{j=0}^\ell \mathrm{supp}(\mathbf{c}_j-\mathbf{u}_j) 
\]
and $\mathrm{\mathrm{rowwt}}(X)\le E^+$.

Conversely, assume (2), and define $\mathbf{c}_j := \mathbf{u}_j-\mathbf{x}_j,\,0\le j\le \ell$. Then $H\mathbf{c}_j = H\mathbf{u}_j-H\mathbf{x}_j = \mathbf{s}_j-\mathbf{s}_j = 0$, so each $\mathbf{c}_j \in \cC$.
Also,
\[
\bigcup_{j=0}^\ell \mathrm{supp}(\mathbf{c}_j-\mathbf{u}_j)= \bigcup_{j=0}^\ell \mathrm{supp}(\mathbf{x}_j)=\mathrm{rowsupp}(X),
\]
hence the disagreement support has size at most $E^+$. Finally, for every $\alpha\in \F_q$,
\[
\left(
\mathbf{u}_0+\sum_{j=1}^{\ell}\alpha^j \mathbf{u}_j
\right)
-
\left(
\mathbf{c}_0+\sum_{j=1}^{\ell}\alpha^j \mathbf{c}_j
\right)
=
\mathbf{x}_0+\sum_{j=1}^{\ell}\alpha^j \mathbf{x}_j.
\]
The support of the right-hand side is contained in $\mathrm{rowsupp}(X)$, so its weight is at most
$E^+$. This proves the claim.
\end{proof}

\begin{proof}[Proof of Lemma~\ref{lem:curve-rank-reduction}]
For $0\le i\le \ell$, define $\mathbf{u}_i := ((\alpha_1)_i,\dots,(\alpha_K)_i)\in \F_q^K,$ and let $Y :=
\Bigl[
\mathbf{s}_0+\sum_{i=1}^{\ell}(\alpha_1)_i \mathbf{s}_i
\ \big|\
\cdots
\ \big|\
\mathbf{s}_0+\sum_{i=1}^{\ell}(\alpha_K)_i \mathbf{s}_i
\Bigr]
\in \F_q^{r\times K}$. Since $\alpha_1,\dots,\alpha_K$ are pairwise distinct and $K\ge \ell+1$, the vectors $\mathbf{u}_0,\mathbf{u}_1,\dots,\mathbf{u}_\ell$ are linearly independent by the Vandermonde determinant. Moreover, $Y=\sum_{i=0}^{\ell}\mathbf{s}_i \mathbf{u}_i^\top,$ so $\mathrm{Row}(Y)\subseteq \mathrm{span}\{\mathbf{u}_0,\ldots,\mathbf{u}_\ell\}$ and $\mathrm{rank}(Y)=h$. Choose a basis $\mathbf v_1,\ldots,\mathbf v_h$ of $\mathrm{Row}(Y)$, and extend it to a basis $\mathbf v_1,\ldots,\mathbf v_h,\mathbf w_{h+1},\ldots,\mathbf w_t$ of $\mathrm{Row}(X)$.
Write $Y=\sum_{i=1}^h \mathbf b_i\mathbf v_i^\top$ for some $\mathbf b_1,\ldots,\mathbf b_h\in \mathbb F_q^r$. Then there exist vectors
$\mathbf a_1,\ldots,\mathbf a_h,\mathbf c_{h+1},\ldots,\mathbf c_t\in \mathbb F_q^n$ such that $X=\sum_{i=1}^h \mathbf a_i\mathbf v_i^\top+\sum_{j=h+1}^t \mathbf c_j\mathbf w_j^\top$. Multiplying by $H$ and using $HX=Y$, we get $\sum_{i=1}^h (H\mathbf a_i-\mathbf b_i)\mathbf v_i^\top+
\sum_{j=h+1}^t H\mathbf c_j\mathbf w_j^\top=0$. By linear independence, $H\mathbf c_j=0$ for $h+1\le j\le t$. Hence each $\mathbf c_j$ is a codeword.

If $\mathbf{c}_t=\mathbf{0}$, then already $\rank(X)\le t-1$, and there is nothing to prove.
So assume $\mathbf{c}_t\neq \mathbf{0}$, and let $T:=\supp(\mathbf{c}_t)$. Since $\mathbf{c}_t\in \cC\setminus\{0\}$ and $d(\cC)\ge d$, we have $|T|=\wt(\mathbf{c}_t)\ge d.$ Each column $\mathbf{x}_j$ has at least $|T|-E$ zero coordinates inside $T$. By the pigeonhole principle, there exists $i_0\in T$ such that $|\{j\in[K]:(\mathbf x_j)_{i_0}=0\}|\ge K\cdot \frac{|T|-E}{|T|}
\ge K\cdot \frac{d-E}{d}=K\gamma$. Let $J:=\{j\in[K]:(\mathbf x_j)_{i_0}=0\}$. Then $|J|\ge K\gamma$.

For each $j\in J$, the $j$-th column of $X$ equals $\mathbf x_j=\sum_{i=1}^h(\mathbf v_i)_j\mathbf a_i+\sum_{p=h+1}^{t-1}(\mathbf w_p)_j\mathbf c_p+(\mathbf w_t)_j\mathbf c_t$. Since $(\mathbf x_j)_{i_0}=0$ and $(\mathbf c_t)_{i_0}\ne0$, we may solve for $(\mathbf w_t)_j$ and substitute back as in Lemma~\ref{lem:rank-reduction-m}. Thus $X_J$ can be expressed using $(\mathbf v_1)_J,\ldots,(\mathbf v_h)_J, (\mathbf w_{h+1})_J,\ldots,(\mathbf w_{t-1})_J$. Hence $\mathrm{rank}(X_J)\le t-1$. Moreover, $X_J$ remains an $E$-witness matrix covering $|J|$ elements of the same syndrome curve.
\end{proof}

\begin{proof}[Proof of~\cref{thm:curve-witness-rank}]
Assume for contradiction that $K\gamma^{t-h}>\ell B_{E,E^+}$. Since $B_{E,E^+}\ge 1$, this implies $K\gamma^{t-h}>\ell$. By Corollary~\ref{cor:curve-rank-reduction}, there exists a subset $J\subseteq[K]$ with $|J|\ge K\gamma^{t-h}$ and vectors $\mathbf a_0,\ldots,\mathbf a_\ell\in\mathbb F_q^n$ such that $\mathbf x_j=\mathbf a_0+\sum_{i=1}^{\ell}\alpha_j^i\mathbf a_i$ and $H\mathbf a_i=\mathbf s_i\,(0\le i\le \ell).$

Let $A=[\mathbf{a}_0|\cdots|\mathbf{a}_\ell]$. Since $S$ has no $(H,\ell,E,E^+)$-correlated agreement, Lemma~\ref{lem:curve-ca-equivalence} implies $\mathrm{rowwt}(A)>E^+$. Every surviving witness column has weight at most $E$, so $|J|\le\left|\left\{\alpha\in\mathbb F_q:\mathbf{a}_0+\sum_{i=1}^{\ell}\alpha^i \mathbf{a}_i\in B_E\right\}\right|$. By Lemma~\ref{lem:curve-ball-intersection}, $|J|\le \ell B_{E,E^+}$. Thus, $K\gamma^{t-h}\le |J|\le \ell B_{E,E^+}$, contradicting the assumption.
\end{proof}

\begin{proof}[Proof of Theorem~\ref{thm:curve-moment}]
Fix a degree-$\ell$ syndrome curve $S= \mathbf{s}_0+\sum_{i=1}^{\ell}\alpha^i \mathbf{s}_i : \alpha\in \F_q$ and an ordered $K$-tuple $(\alpha_1,\dots,\alpha_K)$ of pairwise distinct elements of $\F_q$. Let $h:=\dim_{\mathbb F_q}\mathrm{span}\{s_0,\ldots,s_\ell\}$. If $h=0$, then $S$ trivially has the required correlated agreement, so we may assume $1\le h\le \ell+1$. Consider the event
\[
\begin{aligned}
\mathcal{A}(S;\alpha_1,\dots,\alpha_K)
:=
&\left[
\mathbf{s}_0+\sum_{i=1}^{\ell}\alpha_j^i \mathbf{s}_i \in H_E
\text{ for all } j\in[K]
\right] \\
&\wedge
\left[
S\text{ has no }(H,\ell,E,E^+)\text{-correlated agreement}
\right]
\wedge
[d(\cC)\ge d].
\end{aligned}
\]

If this event occurs, then there exists an $E$-witness matrix $X=[\mathbf{x}_1|\cdots|\mathbf{x}_K]$ for $(S;\alpha_1,\dots,\alpha_K)$. Since $K > \ell\,\cdot B_{E,E^+}\,\gamma^{-s} \ge \ell$, we have $K\ge \ell+1$. Hence Theorem~\ref{thm:curve-witness-rank} applies and shows that every such witness matrix must satisfy $\rank(X)\ge h+s+1$.

Therefore we can choose a subset of $h+s+1$ columns that is linearly independent. There are at most $\binom{K}{h+s+1}|B_E|^{\,h+s+1}$ possible choices for these columns. By Lemma~\ref{lem:random-image}, the probability that a fixed linearly independent $(h+s+1)$-tuple maps to the prescribed syndrome columns is exactly $q^{-r(h+s+1)}$. Hence
\[
\Pr_{\cC}\!\bigl[\mathcal{A}(S;\alpha_1,\dots,\alpha_K)\bigr]
\le
\binom{K}{h+s+1}|B_E|^{\,h+s+1}q^{-r(h+s+1)}.
\]
Now fix $S$. If $\left|\{\alpha\in\mathbb F_q:S(\alpha)\in H_E\}\right|\ge K,\,S \text{ has no }(H,\ell,E,E^+)\text{-correlated agreement},\, d(\cC)\ge d$, then we can choose a $K$-element subset $\{\alpha_1,\dots,\alpha_K\}\subseteq \mathbb{F}_q$ such that $\mathbf{s}_0+\sum_{i=1}^{\ell}\alpha_j^i \mathbf{s}_i\in H_E$ for $j\in [K]$.

%For each such choice, after ordering these $K$ elements arbitrarily, the event $\mathcal{A}(S;\alpha_1,\dots,\alpha_K)$ occurs. Define the event $E$ is there exists degree-$\ell$ syndrome curve $S\subseteq \F_q^r$ such that
%\[
%\left[|\alpha\in \F_q:\mathbf{s}_0+\sum_{i=1}^{\ell}\alpha_j^i \mathbf{s}_i\in H_E|\geq K\right]\wedge \left[S \text{ has no }%(H,\ell,E,E^+)\text{-correlated agreement}\right]\wedge \left[d(\cC)\ge d\right]
%\]
Hence, by union bound over all $\binom{q}{K}$ possible $K$-subsets of $\mathbb{F}_q$ and all coefficient tuples $(\mathbf{s}_0,\dots,\mathbf{s}_\ell)\in(\mathbb{F}_q^r)^{\ell+1}$, the failure probability is at most
\[
\sum_{S,(\alpha_1,\ldots,\alpha_K)}\Pr_{\cC}\!\bigl[\mathcal{A}(S;\alpha_1,\dots,\alpha_K)\bigr]\leq \binom{q}{K}q^{-r(s+1)}\sum_{h=1}^{\ell+1}q^{h(\ell+1)}\binom{K}{h+s+1}|B_E|^{h+s+1}.
\]
This completes the proof.
\end{proof}

\begin{proof}[Proof of Theorem~\ref{thm:curve-explicit}]
Define $b_\varepsilon:=\frac{\varepsilon}{1+\log_2(1/\varepsilon)}$. Let $d_n:=\lfloor \delta n\rfloor,\,\gamma_n:=\frac{d_n-E}{d_n},\,B_n:=\left\lfloor \frac{E^+ +1}{E^+ - E +1}\right\rfloor$. Since $\delta>\rho$, we have $d_n>E$ for all sufficiently large $n$.

We first estimate the probability that $\cC$ has minimum distance less than $d_n$. By Lemma~\ref{lem:H_q-bound} and Lemma~\ref{lem:ball-volume-bound}, $|B_{d_n}| \le q^{(\delta+b_\varepsilon)n+o(n)}$. Therefore,
\[
\Pr_{\cC}[d(\cC)\le d_n]
\le |B_{d_n}|\,q^{-r}
\le q^{(\delta+b_\varepsilon-(1-R))n+o(n)}.
\]
As $\delta+b_\varepsilon-(1-R)=-(a_\varepsilon-b_\varepsilon)$, we obtain $\Pr_{\cC}[d(\cC)\le d_n]\le q^{-(a_\varepsilon-b_\varepsilon)n+o(n)}=q^{-\Omega(n)}$.

Next we estimate the bad-curve probability on the event $d(\cC)\ge d_n$. Let $K := \tau+1$. Since $B_n \le 1+\frac{\rho}{\varepsilon}$ and $\gamma_n^{-1}=\frac{d_n}{d_n-E}\le \frac{\delta}{\delta-\rho}$, we have $\ell\, B_n\,\gamma_n^{-\lambda}\le\ell\left(1+\frac{\rho}{\varepsilon}\right)\left(\frac{\delta}{\delta-\rho}\right)^\lambda\le\tau<K$.

If a degree-$\ell$ syndrome curve violates the $(\ell,E,E^+,\tau/q)$-correlated agreement property, then $\frac{1}{q} \left|\{\alpha\in\mathbb F_q:\mathbf{s}_0+\sum_{i=1}^{\ell}\alpha^i \mathbf{s}_i\in H_E\}\right|>\frac{\tau}{q}$. We bound the probability that there exists degree-$\ell$ syndrome curve $S\subseteq \F_q^r$ such that
\[
\left[\left|\{\alpha\in\mathbb F_q:S(\alpha)\in H_E\}\right|\ge K\right]\wedge \left[S\text{ has no }(H,\ell,E,E^+)\text{-correlated agreement}\right]\wedge \left[d(\cC)\ge d_n\right]
\]
Therefore, by Theorem~\ref{thm:curve-moment} with $d=d_n$ and $s=\lambda$, the probability is at most
\[
\binom{q}{K}q^{-r(\lambda+1)}\sum_{h=1}^{\ell+1}q^{h(\ell+1)}\binom{K}{h+\lambda+1}|B_E|^{h+\lambda+1}
\]

Since $\ell$, $q$, and $K$ depend only on $(R,\rho,\varepsilon,\ell)$, the two binomial factors contribute only $q^{O(1)}$. On the other hand, Lemma~\ref{lem:ball-volume-bound} gives $|B_E|\le q^{(\rho+b_\varepsilon)n+o(n)}$. Hence the exponent in the preceding probability bound is at most $(\ell+1)(1-R)n+(\ell+\lambda+2)(\rho+b_\varepsilon-(1-R))n+o(n)$. By the assumption on $\rho$, we have $1-R-\rho-b_\varepsilon>\varepsilon$. Since $\ell+\lambda+2\ge \frac{(\ell+1)(1-R)}{\varepsilon}$, it follows that $(\ell+\lambda+2)(1-R-\rho-b_\varepsilon)>(\ell+1)(1-R)$. Therefore $(\ell+1)(1-R)+(\ell+\lambda+2)(\rho+b_\varepsilon-(1-R))<0$ and the probability is at most $q^{-\Omega(n)}$.

Combining this with the minimum-distance bound, we conclude that with probability at least $1-q^{-\Omega(n)}$, every degree-$\ell$ syndrome curve $S$ with $\frac{1}{q} \left|\{\alpha\in\mathbb F_q:S(\alpha)\in H_E\}\right|>\frac{\tau}{q}$ admits a matrix $X=[\mathbf{x}_0|\cdots|\mathbf{x}_\ell]$ with $H\mathbf{x}_i=\mathbf{s}_i\quad(0\le i\le\ell)$ and $\mathrm{rowwt}(X)\le E^+$. The equivalent of correlated-agreement formulation follows from Lemma~\ref{lem:curve-ca-equivalence}.
\end{proof}

\begin{proof}[Proof of Theorem~\ref{thm:curve-one-radius}]
For $E^+=E$, the curve moment theorem gives the corresponding probability bound once $K+1>\ell(E+1)\gamma_n^{-\lambda},\, \gamma_n:=\frac{d_n-E}{d_n}.$

We first estimate the minimum-distance event. Since $q=\Theta(n)$ and $d_n=\Theta(n)$, Lemma~\ref{lem:ball-volume-bound} yields $|B_{d_n}|\le q^{d_n+o(n)}<q^{(1-R-\varepsilon)n+o(n)}$. Therefore, $\Pr_{\cC}[d(\cC)\le d_n]\le |B_{d_n}|\,q^{-r}\le q^{-\varepsilon n+o(n)}$.

Next we estimate the bad-curve probability on the event $d(\cC)\ge d_n$. We bound the probability that there exists degree-$\ell$ syndrome curve $S\subseteq \F_q^r$ such that
\[
\left[\left|\{\alpha\in\mathbb F_q:S(\alpha)\in H_E\}\right|\ge K+1\right]\wedge \left[S \text{ has no }(H,\ell,E,E^+)\text{-correlated agreement}\right]\wedge \left[d(\cC)\ge d_n\right].
\]
Applying the degree-$\ell$ curve moment theorem with $d=d_n$, $E^+=E$, $s=\lambda$, the probability is at most $\binom{q}{K+1}q^{-r(\lambda+1)}\sum_{h=1}^{\ell+1}q^{h(\ell+1)}\binom{K+1}{h+\lambda+1}|B_E|^{h+\lambda+1}.$

Because $q=\Theta(n)$ and $K=\Theta(n)$, we have $\log_q\binom{q}{K+1}=o(n)$ and $\log_q\binom{K+1}{\ell+\lambda+2}=o(n)$. Also, by Lemma~\ref{lem:ball-volume-bound}, $|B_E|\le q^{E+o(n)}\le q^{\rho n+o(n)}$. Hence the exponent in the preceding probability bound is at most $(\ell+1)(1-R)n+(\ell+\lambda+2)(\rho-(1-R))n+o(n)$. Therefore, the exponent is bounded above by $\bigl((\ell+1)(1-R)-(\ell+\lambda+2)\varepsilon\bigr)n+o(n)$. By the definition of $\lambda$, $\ell+\lambda+2=\left\lceil \frac{(\ell+1)(1-R)}{\varepsilon}\right\rceil+1\ge\frac{(\ell+1)(1-R)}{\varepsilon}$, and hence $(\ell+1)(1-R)-(\ell+\lambda+2)\varepsilon\le 0$. Thus, the probability is at most $q^{-\Omega(n)}$.

Combining this with the minimum-distance bound, we conclude that with probability at least $1-q^{-\Omega(n)}$, every degree-$\ell$ syndrome curve $S$ with $\frac{1}{q}\left|\{\alpha\in\mathbb F_q:S(\alpha)\in H_E\}\right|>\frac Kq$ admits a matrix $X=[x_0|\cdots|x_\ell]$ such that $H\mathbf{x}_i=\mathbf{s}_i$ for $(0\le i\le\ell)$ and $\mathrm{rowwt}(X)\le E.$ The equivalent of correlated-agreement formulation follows from the curve analogue of Lemma~\ref{lem:curve-ca-equivalence}.
\end{proof}

\end{document}